\documentclass{article}
\usepackage[a4paper]{geometry}
\usepackage{times}
\usepackage{amsmath}
\usepackage{amssymb}
\usepackage{xcolor}
\usepackage{bibspacing}
\usepackage{mathrsfs}
\usepackage{enumitem}
\usepackage{inputenc}
\usepackage{algorithm}
\usepackage{rotating}
\usepackage{lscape}
\usepackage{bbm}

\usepackage{amsthm}
\usepackage{graphicx}
\usepackage{caption}
\usepackage{subcaption}
\usepackage{algpseudocode}
\usepackage[toc,page]{appendix}
\usepackage[affil-it]{authblk}

\newtheorem{theorem}{Theorem}
\newtheorem{definition}{Definition}
\newtheorem{proposition}{Proposition}
\newtheoremstyle{example}
{3pt}
{3pt}
{\upshape}
{}
{\bf}
{.}
{.5em}
{}
\theoremstyle{example}
\newtheorem{example}{Example}

\renewcommand\refname{References and Notes}

\begin{document}
\title{Conditional Randomization Rank Test}

\author{Yanjie Zhong, Todd Kuffner and Soumendra Lahiri}
\date{}
\affil{Washington University in St. Louis, St. Louis, USA}

\maketitle 

\begin{abstract}
We propose a new method named the Conditional Randomization Rank Test (CRRT) for testing conditional independence of a response variable $Y$ and a covariate variable $X$, conditional on the rest of the covariates $Z$.  The new method generalizes the Conditional Randomization Test (CRT) of  \cite{candes2018panning} by exploiting the knowledge of the conditional distribution of $X|Z$ and is a conditional sampling based method that is easy to implement and interpret. 
In addition to guaranteeing exact type 1 error control, owing to a more flexible framework, the new method markedly outperforms the CRT in computational efficiency. We establish bounds on the probability of type 1 error in terms of total variation norm and also in terms of observed Kullback–Leibler divergence  when the conditional distribution of $X|Z$ is misspecified. We validate our theoretical results by extensive simulations and show that our new method has considerable advantages over other existing conditional sampling based methods when we take both power and computational efficiency into consideration. 
\end{abstract}

\section{Introduction}
We study the problem of testing conditional independence in this paper. To be more specific, let $X\in \mathbb{R}$, $Y\in \mathbb{R}$ and $Z\in \mathbb{R}^p$ be random variables. We consider to test the null hypothesis
\begin{equation}
H_0:Y\perp X|Z,
\label{conditional_test}
\end{equation}
i.e., $Y$ is independent of $X$ conditional on the controlling variables $Z$. Conditional independence test plays an important role in many statistical problems, like causal inference (\cite{dawid1979conditional}, \cite{spirtes2010introduction}, \cite{lechner2001identification}, \cite{kalisch2007estimating}, \cite{spirtes2010introduction}) and Bayesian networks (\cite{singh1995construction}, \cite{cheng1998learning}, \cite{campos2006scoring}). For example, consider the case where $X$ is a binary variable indicating whether the patient received a treatment, with $X=1$ when the patient received the treatment and $X=0$ otherwise. Let $Y$ be an outcome associated with the treatment, like the diabetes risk lagged 6 months. Let $Z$ feature patient's individual characteristics. Usually, it is of interest to see if $\mathbb{E}\left[Y|X=1,Z\right]=\mathbb{E}\left[Y|X=0,Z\right]$ 
or more generally, if the probability distribution of $Y$ given $X=1,Z$ is the same as
that of $Y$ given $X=0,Z$ 
so that we can tell whether the treatment is effective.

\subsection{Our contributions}
In this paper, we propose a versatile and computationally  efficient method,
called the Conditional Randomization Rank Test (CRRT)
to test the null hypothesis (\ref{conditional_test}) . This method generalizes the Conditional Randomization Test (CRT) introduced in \cite{candes2018panning}. Like the CRT, the CRRT is built on the conditional version of model-X assumption, which implies that the conditional distribution of $X|Z$ is known. Such an assumption is in fact feasible 
in modern applications. On one hand, it is feasible in practice because usually we have a large amount of unlabeled data (data without $Y$) (\cite{barber2018robust}, \cite{romano2019deep}, \cite{sesia2019gene}), which combined with some domain knowledge, allow us to learn the conditional distribution of $X|Z$ quite accurately. On the other hand, 
as implied in \cite{shah2018hardness}, domain knowledge is needed to construct non-trivial tests of conditional independence. Since it would be too complicated to know how $Y$ depends on the covariates when the covariates are high-dimensional, marginal information on the covariates  becomes valuable and can be utilized in conditional tests. 

Compared with the CRT, we allow the CRRT to be implemented in a more flexible way. It makes the CRRT computationally more efficient than the CRT.
It is known that the greatest drawback of the CRT is its restrictive computational burden (\cite{liu2020fast}). There are many special forms of the CRT aiming at solving this problem, including the distillation to the CRT (dCRT) in \cite{liu2020fast} and the Holdout Randomization Test (HRT) in \cite{tansey2018holdout}. In addition to having favorable 
computational efficacy compared to these methods, the CRRT is consistently  powerful across a wide range of settings while the HRT suffers from lack of power and the dCRT shows inferior power performance  in complicated models. To highlight the advantages of the CRRT,  we provide a simple interaction model example here.
We consider to test hypothesis (\ref{conditional_test}) when $X$ and $Z$ have a non-ignorable interaction in addition to separate linear effects on $Y$. Results are presented in Figure \ref{intro_plot}. We can see that the CRRT can have comparable power as the CRT and CPT while costing far less computational time; See  Section \ref{sec:empirical}
 for more details. In addition, we also show that the CRRT is  robust to  misspecification of the conditional distribution of $X|Z$, both theoretically and empirically.

\begin{figure}[ht]
    \centering
    \begin{subfigure}[t]{0.49\textwidth}
        \centering
        \includegraphics[width=1\linewidth,height=0.6\linewidth]{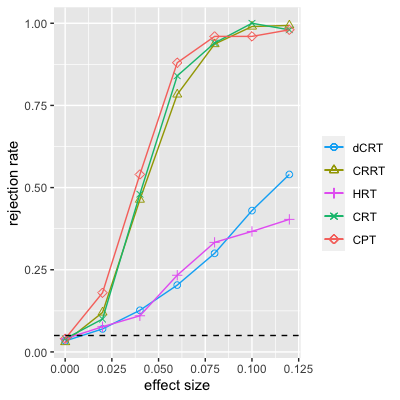} 
        \caption{Rejection Rate} \label{fig:ip1}
    \end{subfigure}
    \hfill
    \begin{subfigure}[t]{0.49\textwidth}
        \centering
        \includegraphics[width=1\linewidth,height=0.6\linewidth]{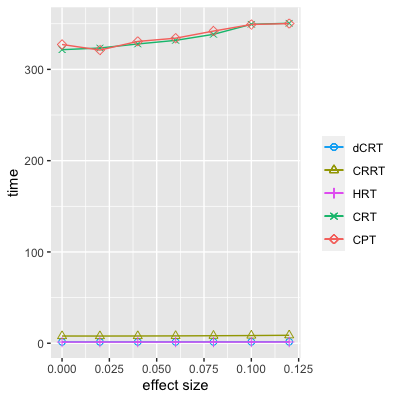} 
        \caption{Time} \label{fig:ip2}
    \end{subfigure}
\caption{(a) Proportions of rejection. The y-axis represents the proportion of times the null hypothesis is rejected. The x-axis represents the true coefficient of $X$. (b) Time in seconds per Monte-Carlo simulation.}
\label{intro_plot}
\end{figure}

\subsection{Related Work}
Conditional tests are ubiquitous in statistical inference. Indeed, 
 many commonly seen statistical problems are essentially conditional independence tests. For example, when we assume a linear model for $Y$ on covariates  $\left(\begin{array}{c} X\\Z\end{array}\right)$:
$$
Y = \beta_0X+\beta^T Z+\varepsilon,
$$
where $\beta_0\in \mathbb{R}$ and $\beta\in \mathbb{R}^p$ are fixed parameters, and $\varepsilon\in \mathbb{R}$ is a random error, it is of interest to test whether $\beta_0=0$, which is equivalent to test (\ref{conditional_test}). If the dimensionality of $Z$ is relatively small compared with the sample size, we can easily construct a valid test based on classical OLS estimator. In a high-dimensional setting, more conditions on sparsity or design matrix (i.e. on the covariates model when we are considering a random design) or coefficients $\beta$ are needed. For example, see \cite{chatterjee2013rates}, \cite{javanmard2014confidence}, \cite{javanmard2014hypothesis}, \cite{van2014asymptotically}, \cite{zhang2014confidence}.

To match the complexity in real application, recently many researchers focused more on the assumptions on covariates while allowing flexibility of the response model to the greatest extent. Our method also belongs to this category. Other methods along this line include the CRT (\cite{candes2018panning}), the CPT (\cite{berrett2019conditional}), the HRT (\cite{tansey2018holdout}) and the dCRT (\cite{liu2020fast}). These methods can usually control the type 1 error in practice because they totally rely on covariates model (which suffer less from model misspecification) when adequate unlabeled data is accessible. We review the CRT and the CPT for completeness 
in Section \ref{connect}. We can naturally relate these model-X methods to semi-supervised learning methods, the key idea of which is also to maximize the use of unlabeled data. There is
some recent work developing reliable methods on variable selection solely based on side-information on the 
covariates; See
\cite{barber2015controlling}, \cite{gimenez2018improving}, \cite{barber2019knockoff} and the references therein. We 
discuss connections  between the CRRT and these methods in the supplementary material. After finishing our work, we became aware of a concurrent work (\cite{li2021deploying}) where the CRRT was briefly mentioned. Compared with their work, we provide a more in-depth analysis on the CRRT beyond validity.

There are many other nonparametric methods for testing conditional independence which we point out here but we do not intend to compare them with our method in this paper because of some fundamental differences in the framework and assumptions. Roughly speaking, these nonparametric methods  can be divided into 4 categories. The first category includes kernel-based methods which require no distributional assumptions (\cite{fukumizu2008kernel}, \cite{zhang2012kernel}, \cite{strobl2019approximate}). 
The Kernel Conditional Independence Test can be viewed as the extension of the Kernel Independence Test (\cite{fukumizu2008kernel}, \cite{gretton2008kernel}). Taking the advantage of the characterization of conditional independence in reproducing kernel Hilbert spaces, one can reduce the original conditional independence test to a low-dimensional test in general. However, kernel-based methods depend on  complex asymptotic distributions of the test statistics. Therefore, they have no finite-sample guarantee and show deteriorating performance as the dimension of $Z$ increases. The second category constructs asymptotically valid tests by measuring the distance between conditional distribution of $Y|Z$ and $Y|X,Z$; See, for example, \cite{su2007consistent}, \cite{su2014testing} and \cite{wang2015conditional}. This approach requires estimating certain  point-wise conditional expectations by kernel methods like the Nadaraya-Watson kernel regression. Hence, the number of samples must be extremely large when the dimension of $Z$ is high. The third category relies on discretization of the conditioning variables, i.e. $Z$. They key idea is that $Y\perp X|Z$ holds if and only if $Y\perp X|Z=z'$ holds for all possible $z'$. Therefore, if one can discretize the space of $Z$, one can test conditional independence on each part and then, can synthesize all test results together (\cite{huang2010testing}, \cite{tsamardinos2010permutation}, \cite{doran2014permutation}, \cite{sen2017model}, \cite{runge2018conditional}). However, it is well-known that discretization of high-dimensional continuous variables is
subject to curse of dimensionality and computationally 
demanding. Note that tests based on within-group permutation also fall in this category because they essentially rely on local independence. The fourth category includes methods testing some weaker null hypothesis (\cite{linton1997conditional}, \cite{bergsma2004testing}, \cite{song2009testing}), where copula estimation would be helpful (\cite{gijbels2011conditional}, \cite{veraverbeke2011estimation}).

\subsection{Paper Outline}
In Section \ref{sec-crrt}, we describe our newly-proposed CRRT in detail and show its effectiveness in controlling finite sample type 1 error. In Section \ref{sec:robust}, we demonstrate that the CRRT is robust to the misspecification of conditional distribution of $X|Z$. In Section \ref{connect}, we analyze the connection among several existing methods based on conditional sampling, including the proposed CRRT. In Section \ref{sec:empirical}, we study finite sample performance of the CRRT and show its advantages
over other conditional sampling based competitors by extensive simulations. Finally, we conclude our paper by Section \ref{sec:discuss} with a discussion.
\subsection{Notation}
The following notation will be used. For a $n\times p$ matrix $x$, its $i$th row is denoted by $x_{[i]}$ and its $j$th column is denoted by $x_j$. For any two random variables $X$ and $Y$, $Y \perp X$ means that they are independent. For two random variables $X$ and $Z$ with $n$ i.i.d. copies $\{(x_i,z_{[i]}):i=1,2,\ldots,n\}$, if the density of $X$ conditional on $Z$ is $Q(\cdot|Z)$, we denote the pdf of $x=(x_1,\ldots,x_n)^T$ conditional on $z=(z_{[1]},\ldots,z_{[n]})^T$ by $Q_n(\cdot|z)$. When there is no ambiguity, we sometimes suppress the subscript and denote it by $Q(\cdot|z)$.

\section{Conditional Randomization Rank Test}
\label{sec-crrt}

Suppose that we have a response variable $Y\in \mathbb{R}$ and predictors $(X,Z)=(X,Z^1,\ldots,Z^p)\in \mathbb{R}^p$. For now, we assume that the conditional distribution of $X|Z$ is known and denote the conditional density by $Q(\cdot|Z)$ throughout this section. Let $y=(y_1,\ldots,y_n)^T$ be an n-dimension i.i.d. realization of $Y$, and similarly, let  $(x,z)=(x,z_1,\ldots,z_p)$ be an $n\times (p+1)$ dimension matrix with rows i.i.d. drawn from the distribution of $(X,Z)$. For a given positive integer $b$ (usually we can take $b=100$), and $k=1,\ldots,b$, let $x^{(k)}$ be an n-dimensional 
vector with elements $x^{(k)}_i\sim Q(\cdot|z_{[i]})$ independently, $i = 1,\ldots,n$, where $z_{[i]}$ is the $i$th row of $z$. We denote $\tilde{x}=(x,x^{(1)},\ldots,x^{(b)})$.

Let $T(y,z,\tilde{x})=(T_0,T_1,\ldots,T_b)^T:\mathbb{R}^n\times \mathbb{R}^{n\times p}\times \mathbb{R}^{n\times (b+1)}\rightarrow \mathbb{R}^{b+1}$ be any function mapping the enlarged dataset to a $(b+1)$-dimension vector. Typically, $T_0$ is a function measuring the importance of  $x$ to $y$ and $T_k$ is a function measuring the importance of $x^{(k)}$ to $y$.

Let $\Sigma_{b+1}$ be the collection of all permutations of $(1,2,\ldots,b,b+1)$. For any $\sigma\in \Sigma_{b+1}$, $\tilde{x}_{permute(\sigma)}$ is a $n\times (b+1)$ matrix with the $j$th column being the $\sigma_j$th column of $\tilde{x}$, $j=1,2,\ldots,b+1$. If $\nu$ is a $(b+1)$-dimension vector, $\nu_{permute(\sigma)}$ is a $(b+1)$-dimension vector with the $j$th element being the $\sigma_j$th element of $\nu$, $j=1,2,\ldots,b+1$. We introduce a property of $T$, which is key to our analysis.

\begin{definition}
Let $T(y,z,\tilde{x})=(T_0,T_1,\ldots,T_b)^T:\mathbb{R}^n\times \mathbb{R}^{n\times p}\times \mathbb{R}^{n\times (b+1)}\rightarrow \mathbb{R}^{b+1}$. We say that $T$ is $X$-symmetric if for any $\sigma\in \Sigma_{b+1}$, $T(y,z,\tilde{x}_{permute(\sigma)}) = T(y,z,\tilde{x})_{permute(\sigma)}$.
\end{definition}

If $X$ is indeed independent with $Y$ given $Z$, we expect that $x,x^{(1)},\ldots,x^{(b)}$ should have similar influence on $Y$. Hence, if $T$ is a X-symmetric function, $T_0(y,z,\tilde{x}),\ldots,T_b(y,z,\tilde{x})$ should also be stochastically similar. 

\begin{proposition}
Assume that the null hypothesis  $H_0:Y\perp X|Z$ is true. If $T(y,z,\tilde{x})=(T_0,T_1,\ldots,T_b)^T:\mathbb{R}^n\times \mathbb{R}^{n\times p}\times \mathbb{R}^{n\times (b+1)}\rightarrow \mathbb{R}^{b+1}$ is X-symmetric, then 
$$
T_0(y,z,\tilde{x})\stackrel{d}{=}T_1(y,z,\tilde{x})\stackrel{d}{=}\ldots\stackrel{d}{=}T_b(y,z,\tilde{x}).
$$ 
Further, they are exchangeable. In particular, for any permutation $\sigma \in \Sigma_{b+1}$, 
$$
(T_0(y,z,\tilde{x}),\ldots,T_b(y,z,\tilde{x}))_{permute(\sigma)}\stackrel{d}{=}(T_0(y,z,\tilde{x}),\ldots,T_b(y,z,\tilde{x})).
$$
\label{xsym}
\end{proposition}

Based on the above exchangeability, a natural choice of p-value for our target test is given by
\begin{equation}
p_{CRRT} = \frac{1}{b+1}\sum\limits_{k=0}\limits^{b}\mathbbm{1}\{T_0(y,z,\tilde{x})\leq T_k(y,z,\tilde{x})\}.
\label{pvalue}
\end{equation}
With this p-value, we will reject $H_0:Y\perp X|Z$ at level $\alpha$ if $p\leq \alpha$. In other words, we will have a strong evidence to believe that Z cannot fully cover X's influence on Y when $T_0$ is among the $\lfloor\alpha(1+b)\rfloor$ largest of  $(T_0,T_1,\ldots,T_b)$. We judge whether to reject $H_0$ based on the rank of $T_0$. That is why we name the test introduced here Conditional Randomization Rank Test. The following proposition and theorem justifies the use of p-value defined in (\ref{pvalue}).

\begin{proposition}
Assume that $T(y,z,\tilde{x})=(T_0,T_1,\ldots,T_b)^T:\mathbb{R}^n\times \mathbb{R}^{n\times p}\times \mathbb{R}^{n\times (b+1)}\rightarrow \mathbb{R}^{b+1}$ is X-symmetric. Denote the ordered statistics of $(T_0,T_1,\ldots,T_b)$ by $T_{(0)}\leq T_{(1)}\leq\ldots\leq T_{(b)}$. When $H_0:Y\perp X|Z$ holds, we have
$$
\mathbb{P}(p_{CRRT}\leq \alpha|z,y,T_{(0)},T_{(1)},\ldots,T_{(b)})\leq \alpha,
$$
for any pre-defined $\alpha\in[0,1]$.
\label{conditionprop}
\end{proposition}

\begin{theorem}
Under the null hypothesis $H_0:Y\perp X|Z$, assume that the test function $T(y,z,\tilde{x})=(T_0,T_1,\ldots,T_b)^T:\mathbb{R}^n\times \mathbb{R}^{n\times p}\times \mathbb{R}^{n\times (b+1)}\rightarrow \mathbb{R}^{b+1}$ is X-symmetric, we can know directly from Proposition \ref{conditionprop}
 that the p-value defined in (\ref{pvalue}) satisfies
$$
\mathbb{P}(p_{CRRT}\leq \alpha) \leq \alpha
$$
for any $\alpha\in[0,1]$. Actually, $p_{CRRT}$ is also conditionally valid:
$$
\mathbb{P}(p_{CRRT}\leq \alpha|y,z)\leq \alpha.
$$
\label{maintheorem}
\end{theorem}

\section{Robustness of the CRRT}
\label{sec:robust}

According to Theorem  \ref{maintheorem}, we know that when the conditional distribution of $X$ given $Z$ is fully known, CRRT can precisely control type 1 error at any desirable level. However, in real applications, this information may
not be available. Luckily, however, the following theorems guarantee that the CRRT is robust to moderate misspecification of the conditional distribution. We state two theorems explaining the robustness of the CRRT from two different angles.

\begin{theorem}
(Type 1 Error Upper Bound 1.) Assume the true conditional distribution of $X|Z$ is $Q^{\star}(\cdot|Z)$. The pseudo variables matrix $(x^{(1)},\ldots,x^{(b)})$ is generated in the same way described in Section \ref{sec-crrt} except that the generation is based on a falsely specific or roughly estimated conditional density $Q(\cdot|Z)$ instead of $Q^{\star}(\cdot|Z)$. The CRRT p-value $p_{CRRT}$ is defined in (\ref{pvalue}). Under the null hypothesis $H_0:Y\perp X|Z$, for any given $\alpha\in[0,1]$, 
$$
\mathbb{P}(p_{CRRT}\leq \alpha|z,y)\leq \alpha+d_{TV}(Q^{\star}_n(\cdot|z),Q_n(\cdot|z)),
$$
$$
\mathbb{P}(p_{CRRT}\leq \alpha)\leq \alpha+\mathbb{E}[d_{TV}(Q^{\star}_n(\cdot|z),Q_n(\cdot|z))],
$$
where the total variation distance is defined as $d_{TV}(Q^{\star}_n(\cdot|z),Q_n(\cdot|z)) = \sup_{A\subset \mathbb{R}^n}|P_Q(x\in A|z)-P_{Q^{\star}}(x\in A|z)|$.
\label{robust1}
\end{theorem}

With the Pinsker's inequality, we have $d_{TV}^2(Q^{\star}_n(\cdot|z),Q_n(\cdot|z))\leq \frac{1}{2}d_{KL}(Q^{\star}_n(\cdot|z),Q_n(\cdot|z))$ (\cite{berrett2019conditional}). Combining it with the results in Theorem \ref{robust1}, we know that controlling the Kullback-Leibler divergence is sufficient to guarantee a small deviation from desired type 1 error level.

\begin{example} (Gaussian Model) We now consider that $\left(\begin{array}{c} X\\Z\end{array}\right)$ is jointly Gaussian. It implies that there exists $\beta\in \mathbb{R}^p$ and $\sigma^2>0$ such that conditional on $z$, $x_i|z\sim N(z_{[i]}^T\beta,\sigma^2)$ independently. 

Suppose that we know the true variance $\sigma^2$ and can estimate $\beta$ by $\hat{\beta}$ from a large amount of unlabeled data of size $N$. For example, $\hat{\beta}$ can be the OLS estimator when $p$ is relatively small, or some estimators based on regularization when $p$ is of a comparative order as $N$. Without the sparsity assumption, under appropriate conditions, we have $||\hat{\beta}-\beta||_2=O_p\left(\sqrt{\frac{p\log p}{N}}\right)$ (\cite{zhang2008sparsity}, \cite{bickel2009simultaneous}, \cite{chernozhukov2018double}). We may simply assume $||\hat{\beta}-\beta||_2=O\left(\sqrt{\frac{p\log p}{N}}\right)$ for convenience. If we additionally assume that $\mathbb{E}|z_{ij}|^2=O(1)$ uniformly for $i=1,\ldots,n$ and $j=1,\ldots,p$, we have
$$
\resizebox{.92\hsize}{!}{$
\begin{array}{ll}
& \left[\mathbb{E}\left[d_{TV}(Q^{\star}_n(\cdot|z),Q_n(\cdot|z))\right]\right]^2 \leq \mathbb{E}\left[d_{TV}^2(Q^{\star}_n(\cdot|z),Q_n(\cdot|z))\right] \leq \frac{1}{2}\mathbb{E}\left[d_{KL}(Q^{\star}_n(\cdot|z),Q_n(\cdot|z))\right]\\

= & \frac{1}{2}\mathbb{E}\left[\sum\limits_{i=1}^n d_{KL}(Q^{\star}_n(\cdot|z_{[i]}),Q_n(\cdot|z_{[i]}))\right]= \frac{1}{2}\mathbb{E}\left[\sum\limits_{i=1}^n \frac{1}{2\sigma^2}(z_{[i]}^T\hat{\beta}-z_{[i]}^T\beta)^2 \right]\\

\leq & \frac{n}{4\sigma^2}\mathbb{E}\left[\|\hat{\beta}-\beta\|_2^2\|z_{[1]}\|_2^2\right]=O\left(\frac{np^2\log p}{N}\right).\\
\end{array}
$}
$$
\label{ex1}
\end{example}

It implies that if $\frac{np^2\log p}{N}$ converges to 0 as $n$ goes to infinity, the type 1 error control would be asymptotically exact. For other models, the conditions may be more complicated. But in general, with sufficiently large amount of unlabeled data, we are able to make the bias of type 1 error control negligible.

Different from the Theorem \ref{robust1}, the following theorem demonstrate the robustness of CRRT from another angle, where the discrepancy between $Q^{\star}(\cdot|z)$ and $Q(\cdot|z)$ will be quantified in the form of observed KL divergence.

\begin{theorem}
(Type 1 Error Upper Bound 2.) Assume the true conditional distribution of $X|Z$ is $Q^{\star}(\cdot|Z)$. The pseudo variables matrix $(x^{(1)},\ldots,x^{(b)})$ is generated in the same way described in section \ref{sec-crrt} except that the generation is based on a falsely specific or roughly estimated conditional density $Q(\cdot|Z)$ instead of $Q^{\star}(\cdot|Z)$. The CRRT p-value $p_{CRRT}$ is defined in (\ref{pvalue}). We define random variables
\begin{equation}
\widehat{KL}_k = \sum\limits_{i=1}\limits^{n}\log \frac{Q^{\star}(x_i|z_{[i]})Q(x_i^{(k)}|z_{[i]})}{Q(x_i|z_{[i]})Q^{\star}(x_i^{(k)}|z_{[i]})},\ k = 1,\ldots,b.
\label{kl}
\end{equation}
Denote the ordered statistics of $\left\{\widehat{KL}_1,\ldots,\widehat{KL}_b\right\}$ by $\widehat{KL}_{(1)}\leq\ldots\leq \widehat{KL}_{(b)}$. For any given  $\varepsilon_1,\ldots,\varepsilon_b$, which are independent of $x$ and $y$ conditional on $z$, and for any given $\alpha\in[0,1]$, under the null hypothesis $H_0:Y\perp X|Z$, 
$$
\mathbb{P}(P_{CRRT}\leq \alpha,\widehat{KL}_{(1)}\leq \varepsilon_1,\ldots,\widehat{KL}_{(b)}\leq \varepsilon_b)\leq \mathbb{E}\left[\frac{\lfloor (b+1)\alpha\rfloor}{1+e^{-\varepsilon_1}+\ldots+e^{-\varepsilon_b}}\right].
$$
Further, we can obtain the following bound directly controlling type 1 error,
\begin{equation}
\mathbb{P}(P_{CRRT}\leq \alpha) \leq \inf\limits_{\varepsilon_1,\ldots,\varepsilon_b} \left\{ \mathbb{E}\left[\frac{\lfloor (b+1)\alpha\rfloor}{1+e^{-\varepsilon_1}+\ldots+e^{-\varepsilon_b}}\right]+ \mathbb{P}\left(\bigcup\limits_{k=1}\limits^{b}\{\widehat{KL}_{(k)}>\varepsilon_k\}\right)\right\}.
\label{klbound}
\end{equation}
\label{thm:robust2}
\end{theorem}

\noindent {\bf Remark:}
\begin{itemize}
\item We can see that, when we have the access to the accurate conditional density, i.e. $Q(\cdot|Z) \equiv Q^{\star}(\cdot|Z)$, $\widehat{KL}_1,\ldots,\widehat{KL}_b$ will be all zero. The last bound for type 1 error in Theorem \ref{thm:robust2} would be $\frac{\lfloor (b+1)\alpha\rfloor}{b+1}$, which is nearly exact when $b$ is sufficiently large. 
\item Though at first glance, it seems that the bound given in (\ref{klbound}) may increase with increasing $b$ and may even grow to 1 due to the second term on the righthand side, actually we can show that for any small $\epsilon>0$, the bound on the righthand side of (\ref{klbound}) can be asymptotically controlled from above, as $b$ grows to infinity, by a function $g(Q,Q^{\star},\alpha)$, which equals to $\alpha(1+\epsilon)$ when $Q=Q^{\star}$. For more details and a generalization of Theorem \ref{thm:robust2}, see the supplement. 
\item In Theorem \ref{robust1} and \ref{thm:robust2}, the proposed conditional density $Q(\cdot|Z)$ cannot depend on $(y,z,x)$. The source of $Q(\cdot|Z)$ is usually a combination of domain knowledge and unlabeled data. It would be interesting to see what theoretical bound we can have when $Q(\cdot|Z)$ is estimated from $(y,z,x)$. 
\end{itemize}

\begin{example}(Gaussian Model) We continue to follow the settings and assumptions in Example \ref{ex1}. For $k=1,\ldots,b$, we can simplify $\widehat{KL}_k$ as 
$$
\widehat{KL}_k=\frac{1}{\sigma^2}(x-x^{(k)})^Tz(\beta-\hat{\beta})=\frac{1}{\sigma^2}(x^{(k)}-z\beta)^Tz(\hat{\beta}-\beta)-\frac{1}{\sigma^2}(x-z\beta)^Tz(\hat{\beta}-\beta).
$$
Denote $M_0\triangleq\frac{1}{\sigma^2}(x-z\beta)^Tz(\hat{\beta}-\beta)$ and $M_k\triangleq\frac{1}{\sigma^2}(x^{(k)}-z\beta)^Tz(\hat{\beta}-\beta)$ for $k=1,\ldots,b$. Besides, denote $H\triangleq (\hat{\beta}-\beta)^Tz^Tz(\hat{\beta}-\beta)$. Under our assumptions, $\mathbb{E}H=O\left(\frac{np^2\log p}{N}\right)$. We know that $M_0|z\sim N(0,\frac{1}{\sigma^2}H)$, $M_1|z,\ldots,M_k|z \sim N(\frac{1}{\sigma^2}H,\frac{1}{\sigma^2}H)$ and they are conditionally independent. Now, we assume that $\frac{np^2\log p}{N}$ converges to $0$ as $n$ goes to infinity, which is the same as in the Example \ref{ex1}. 

First, consider that $b$ is bounded, that is, $b$ does not go to infinity as $n$ goes to infinity. We let $\varepsilon_1=\ldots=\varepsilon_k=\sqrt{\frac{np^2\log p}{N}}f(n)$, where $f(n)$ satisfies that $\lim\limits_{n\rightarrow \infty}f(n)=\infty$ and $\lim\limits_{n\rightarrow \infty}\sqrt{\frac{np^2\log p}{N}}f(n)=0$. Then, according to (\ref{klbound}), we have
$$
\resizebox{.92\hsize}{!}{$
\begin{array}{ll}
& \mathbb{P}\left(P_{CRRT}\leq \alpha\right) \leq \frac{\lfloor (b+1)\alpha \rfloor}{1+b\times exp\left(-\sqrt{\frac{np^2\log p}{N}}f(n)\right)}+\mathbb{P}\left(\bigcup\limits_{k=1}\limits^{b}\left\{\widehat{KL}_{(k)}>\sqrt{\frac{np^2\log p}{N}}f(n)\right\}\right)\\

\leq &  \frac{\lfloor (b+1)\alpha \rfloor}{1+b\times exp\left(-\sqrt{\frac{np^2\log p}{N}}f(n)\right)}+ \mathbb{P}\left(\max\limits_{k=1,\ldots,b}M_k-M_0>\sqrt{\frac{np^2\log p}{N}}f(n) \right).\\
\end{array}
$}
$$

The difference between the first term and $\frac{\lfloor (b+1)\alpha \rfloor}{1+b}$ obviously converges to 0 because $\frac{np^2\log p}{N}f(n)$ converges to $0$. As for the second term, note that
$$
\resizebox{.92\hsize}{!}{$
\mathbb{P}\left(\max\limits_{k=1,\ldots,b}M_k-M_0>\sqrt{\frac{np^2\log p}{N}}f(n) \right) = \mathbb{P}\left(\max\limits_{k=1,\ldots,b}\sqrt{\frac{N}{np^2\log p}}M_k-\sqrt{\frac{N}{np^2\log p}}M_0>f(n) \right),
$}
$$
where $\sqrt{\frac{N}{np^2\log p}}M_0,\sqrt{\frac{N}{np^2\log p}}M_1,\ldots, \sqrt{\frac{N}{np^2\log p}}M_b=O_p(1)$. Besides, considering that $f(n)\rightarrow \infty$ and $b$ is bounded, we can conclude that the second term converges to $0$ as $n\rightarrow \infty$. Therefore, we have
$$
\mathbb{P}\left(P_{CRRT}\leq \alpha\right) - \frac{\lfloor (b+1)\alpha \rfloor}{1+b} \rightarrow 0,
$$
as $n\rightarrow \infty$. As for the case of $\lim\limits_{n\rightarrow \infty}b = \infty$, we relegate it to the supplement.

\end{example}

The above 2 theorems guarantee the robustness of the CRRT. To control the type 1 error at a tolerable level, it is sufficient to control the total variation or the observed KL divergence of $Q^{\star}(\cdot|Z)$ and $Q(\cdot|Z)$. Then, it is natural to ask a conjugate question. Whether controlling the total variation or the observed KL divergence is necessary? The following 2 theorems give positive answers.

\begin{theorem}
(Type 1 Error Lower Bound 1.) Assume the null hypothesis $H_0:Y\perp X|Z$ holds. Suppose the true conditional distribution $Q^{\star}(\cdot|Z)$, the proposed conditional density $Q(\cdot|Z)$ used to sample pseudo variables, and the desired type 1 error level $\alpha\in(0,1]$ are fixed. Suppose $x\sim Q^{\star}(\cdot|z)$, $x^{(k)}\sim Q(\cdot|z),k=1,\ldots,b,$ independently. Suppose that $X|Z$ is a continuous variable under both of $Q(\cdot|Z)$ and $Q^{\star}(\cdot|Z)$. Let $A(z)\subseteq\mathbb{R}^n$ such that 
$$
d_{TV}(Q^{\star}_n(\cdot|z),Q_n(\cdot|z)) = \mathbb{P}_{Q_n^{\star}}(x\in A(z)|z) -\mathbb{P}_{Q_n}(x\in A(z)|z).
$$
For any fixed $0<\varepsilon<\alpha$, when $b$ is sufficiently large, we have
$$
\mathbb{P}(p_{CRRT}\leq \alpha|y,z) \ge \left[\alpha-\varepsilon+d_{TV}(Q^{\star}_n(\cdot|z),Q_n(\cdot|z)) f(z,Q,Q^{\star},\alpha,\varepsilon) \right]\times \left(1-o(1) \right),
$$
where $ f(z,Q,Q^{\star},\alpha,\varepsilon)=\mathbbm{1}\{\alpha_0(z)\ge \alpha-\varepsilon\}\frac{\alpha-\varepsilon}{\alpha_0(z)}+  \mathbbm{1}\{\alpha_0(z)< \alpha-\varepsilon\}\frac{1-(\alpha-\varepsilon)}{1-\alpha_0(z)}$, $\alpha_0(z)=\mathbb{P}_{Q_n}(x\in A(z)|z)$, and the $o(1)$ term decays to 0 exponentially as $b$ grows to infinity.

\label{lowerbound1}
\end{theorem}

It is easy to see that $f(z,Q,Q^{\star},\alpha,\varepsilon)$ is bounded from below by $\max\{\alpha-\varepsilon,1-(\alpha-\varepsilon)\}$. Therefore, the extra error due to misspecification of conditional distribution is at least linear in $d_{TV}(Q^{\star}_n(\cdot|z),Q_n(\cdot|z))$. It demonstrates that when we apply the CRRT, we need to make sure that we propose a conditional distribution with a small bias. Otherwise, bad choice of $X$-symmetric test function $T$ could lead to considerable type 1 error. It is also intuitive that if observed KL divergence is large with high probability, the CRRT would fail to give a good type 1 error control. The following theorem formally demonstrate this point.

\begin{theorem}
(Type 1 Error Lower Bound 2.) Assume the null hypothesis $H_0:Y\perp X|Z$ holds. Suppose the true conditional distribution $Q^{\star}(\cdot|Z)$, the proposed conditional density $Q(\cdot|Z)$ used to sample pseudo variables, and the targeted type 1 error level $\alpha\in[0,1]$ are fixed. Suppose that $X|Z$ is a continuous variable under both of $Q(\cdot|Z)$ and $Q^{\star}(\cdot|Z)$. Denote $\lambda=\lfloor(b+1)\alpha\rfloor$. Suppose the following inequality holds when $x\sim Q^{\star}(\cdot|z)$, $x^{(k)}\sim Q(\cdot|z),k=1,\ldots,b,$ independently,
\begin{equation}
\mathbb{P}(\widehat{KL}_k\ge 0,k=1,\ldots,b,\ and,\ \widehat{KL}_{(\lambda)}\ge \varepsilon)\ge c,
\label{conditionc}
\end{equation}
where $\varepsilon$ and $c$ are some nonnegative values, $\widehat{KL}_k,k=1,\ldots,b$ are defined in (\ref{kl}). Then there exists an X-symmetric function $T$ and a conditional distribution $P_{Y|Z}$, such that:
\begin{enumerate}
\item If $x,x^{(1)},\ldots,x^{(b)}\sim_{i.i.d.}Q^{\star}(\cdot|z)$, CRRT has near-exact type 1 error, i.e.,
$$
\mathbb{P}(P_{CRRT}\leq \alpha) = \frac{\lfloor (b+1)\alpha\rfloor}{b+1}.
$$
\item If $x\sim Q^{\star}(\cdot|z)$, $x^{(k)}\sim Q(\cdot|z),k=1,\ldots,b,$ independently, we have a lower bound for type 1 error,
$$
\mathbb{P}(P_{CRRT}\leq \alpha) \ge \frac{\lfloor (b+1)\alpha\rfloor}{b+1} + c\left(1-\frac{\lfloor (b+1)\alpha\rfloor}{b+1}\right)(1-e^{-\varepsilon}).
$$
\end{enumerate}
\label{lowerbound2}
\end{theorem}

Actually, there could be many alternatives way to state the above theorem. We can replace the key condition ($\ref{conditionc}$) with a more general form like
$$
\mathbb{P}(\widehat{KL}_{(k)}\ge\varepsilon_k,k=1,\ldots,b)\ge c,
$$
where $0 \leq \varepsilon_1\leq \ldots \leq \varepsilon_b$, and then derive a corresponding lower bound for the type 1 error. 


\section{Connections Among Tests Based on Conditional Sampling}
\label{connect}
As we have mentioned, model-X methods can effectively use information contained in the (conditional) distribution of covariates and can shift the burden from learning the model of response on covariates to estimating the (conditional) distribution of covariates. They would be a very suitable choice when we have little knowledge about how the response relies on covariates but know more about how the covariates interact with each other. When it comes to the conditional independence testing problem, following the spirit of model-X, there exist several methods based on conditional sampling, including Conditional Randomization Test (\cite{candes2018panning}), Conditional Permutation Test (\cite{berrett2019conditional}) and the previously introduced Conditional Randomization Rank Test. For completeness, we also propose Conditional Permutation Rank Test, which will be defined in the current section later soon. In this section, we discuss the connections among them.

To make the structure complete, let us briefly review the CRT and CPT. We still suppose that $(y,z,x)$ consists of $n$ i.i.d. copies of random variables $(Y,Z,X)$, where $y$ is of dimension $n\times 1$, $z$ is of dimension $n\times p$ and $x$ is of dimension $n\times 1$. Suppose that we believe the conditional density of $X|Z$ is $Q(\cdot|Z)$. This conditional distribution is the key to the success of conditional tests. It could be either precisely true or biased. As shown in \cite{berrett2019conditional}, as long as $Q$ has a moderate deviation from the true underlying conditional distribution, say $Q^{\star}(\cdot|Z)$, we are able to control the type 1 error. \\

\noindent {\bf $\bullet$ Comparing CRRT and CRT}

The CRT independently samples $b$ pseudo vectors $x^{(1)},\ldots,x^{(b)}$ in the same way as the CRRT. For each $k=1,2,\ldots,b$, $x^{(k)}$ is of dimension $n$, with entries $x_i^{(k)}\sim Q(\cdot|z_{[i]})$ independently, $i=1,2,\ldots,n$. Let $T_m=T(y,z,x)$ be any pre-defined function mapping from $\mathbb{R}^{n}\times \mathbb{R}^{n\times p}\times \mathbb{R}^n$ to $\mathbb{R}$. Here, the requirement of 'pre-defined' is not strict. $T_m$ only needs to be independent with $x$. Then, the p-values of the CRT is defined as
\begin{equation}
p_{CRT}=\frac{1+\sum\limits_{k=1}\limits^b \mathbbm{1}\{T_m(y,z,x^{(k)})\ge T_m(y,z,x)\}}{1+b}.
\label{pcrt}
\end{equation}
It is easy to tell that the CRT is essentially a subclass of the CRRT because the CRT requires that the test statistic $T_m$ is same for $x,x^{(1)},\ldots,x^{(b)}$ while the CRRT relaxes such restriction to $X$-symmetry. Compared with the CRT, the CRRT enjoys greater flexibility and therefore can have greater efficiency. We will give some concrete examples to demonstrate this point in our empirical experiments. Here, we show an intuitive toy example. Suppose that we try to use the LASSO estimator as our test statistics for the CRT. Let
$$
(\hat{\beta}_1^T(y,z,x),\hat{\beta}_2(y,z,x))^T = \mathop{\arg\min}\limits_{\beta_1\in \mathbb{R}^p,\beta_2\in \mathbb{R}}||y-z\beta_1-x\beta_2||_2^2+\hat{\lambda}(||\beta_1||_1+||\beta_2||_1),
$$
where $\hat{\lambda}$ can be given by cross-validation based on $(y,z,x)$ or be fixed before analyzing data. For $k=0,1,2,\ldots,b$, let $T_m(y,z,x^{(k)})=\hat{\beta}_2(y,z,x^{(k)})$. Suppose the dimension of $Z$ is $p=n$ and the number of pseudo variables $b=n$. According to Efron, Hastie, Johnstone and Tibshirani (2004), the time complexity of calculating $b+1$ test statistics will be $O(n^4)$. In contrast, if we apply the CRRT with a similar LASSO statistics, we can see a significant save of time. Let
$$
(\hat{\beta}_3^T(y,z,\tilde{x}),\hat{\beta}^T_4(y,z,\tilde{x}))^T = \mathop{\arg\min}\limits_{\beta_3\in \mathbb{R}^p,\beta_4\in \mathbb{R}^{b+1}}||y-z\beta_3-\tilde{x}\beta_4||_2^2+\hat{\lambda}(||\beta_1||_3+||\beta_2||_4),
$$
where $\hat{\lambda}$ can be given by cross-validation based on $(y,z,\tilde{x})$ or be fixed before analyzing data. Then, we can define test statistics
$$
T(y,z,\tilde{x})=(T_0(y,z,\tilde{x}),\ldots,T_b(y,z,\tilde{x}))^T=\hat{\beta}_4(y,z,\tilde{x}),
$$
whose calculation would cost $O(n^3)$ time. Hence, the order of time complexity  is 
smaller by a factor of $\Theta(n)$. 

In addition to the above LASSO example, there are numerous commonly-seen examples we can use to explain CRRT's computational advantage. \cite{louppe2014understanding} showed that the time complexity of CART is $\Theta(pn\log^2 n)$ in average case, where $p$ is the number of variables and $n$ is the number of samples. Suppose $p=n$ and let the number of pseudo variables $b=n$ in both of CRT and CRRT. The total time complexity of CRT is $\Theta(n^3\log^2 n)$ while the complexity of CRRT is $\Theta(n^2\log^2 n)$. Again, CRRT outperforms CRT by a factor of $\Theta(n)$.\\

\noindent {\bf $\bullet$ Connecting CRRT with CPT}

The CPT proposed by \cite{berrett2019conditional} can be viewed as a robust version of the CRT. It inherits the spirit of traditional permutation tests and adopts a pseudo variable generation scheme different from the CRT. Denote the ordered statistics of $x$ by $x_{ordered}=(x_{(1)},\ldots,x_{(n)})^T$, where $x_{(1)}\leq \ldots \leq x_{(n)}$. Conditional on $x_{ordered}$ and $z$, define a conditional distribution $R_n(\cdot|z,x_{ordered})$ induced by $Q$:
$$
R_n(v|z,x_{ordered})=\left\{
\begin{array}{ll}
\frac{Q_n(v|z)}{\sum\limits_{\dot{v}\in\Gamma(x_{ordered})}Q_n(\dot{v}|z)}, & v\in\Gamma(x_{ordered}),\\
0, & v\notin\Gamma(x_{ordered}),\\
\end{array}
\right.
$$
where $\Gamma(x_{ordered})\triangleq \{x_{permute(\sigma)}:\sigma \in \Sigma_n\}$ and $\Sigma_n$ is the collection of all permutations of $\{1,2,\ldots,n\}$. Similarly, we can also define $R^{\star}_n(v|z,x_{ordered})$ induced by $Q^{\star}$. Let $T_m=T(y,z,x)$ be any function from $\mathbb{R}^{n}\times \mathbb{R}^{n\times p}\times \mathbb{R}^n$ to $\mathbb{R}$. The requirements of $T_m$ is same as those in the CRT. The p-values of the CPT is defined in the same way as the CRT. For the CRT, if $Q$ specifies the true conditional distribution, under the null hypothesis that $H_0:Y\perp X|Z$, we can see that $T_m(y,z,x^{(0)}),\ldots,T_m(y,z,x^{(b)})$ are i.i.d. conditional on $y$ and $z$. Such exchangeability guarantees the control of type 1 error. But for the CPT, we need to take a further step to condition on $x_{ordered},y$ and $z$ to ensure exchangeability. As usual, more conditioning, more robustness and correspondingly, less power. \cite{berrett2019conditional} has some good examples to illustrate these points. The following two propositions, combined with the Theorem 5 in \cite{berrett2019conditional} can give us some clues on the robustness provided by the extra conditioning.

\begin{proposition}
Assume the null hypothesis $H_0:Y\perp X|Z$ holds. Let $x\sim Q_n^{\star}(\cdot|z)$, conditional on $z$. And, conditional on $x_{ordered}$ and $z$, $x^{(1)},\ldots,x^{(b)}\sim R_n(\cdot|z,x_{ordered})$ independently. There exists a statistic $T_m(y,z,x):\mathbb{R}^n\times \mathbb{R}^{n\times p}\times\mathbb{R}^n \rightarrow \mathbb{R}$ such that, for the CPT, the following lower bound holds.
$$
\begin{array}{cl}
 &\mathbb{E}\left[\sup\limits_{\alpha\in[0,1]}(\mathbb{P}(p_{CPT}\leq \alpha|y,z,x_{ordered})-\alpha)|y,z \right] \\
 \ge& \mathbb{E}\left[d_{TV}(R^{\star}_n(\cdot|z,x_{ordered}),R_n(\cdot|z,x_{ordered}))|z\right]- \frac{(1+o(1))}{2}\sqrt{\frac{\log b}{b}}
\end{array}
$$
as $b\rightarrow \infty$, for a.e. $y\in \mathbb{R}^n$, $z\in \mathbb{R}^{n\times p}$, where $x_{ordered}$ is the set of ordered statistics of $x$.
\label{cptlower}
\end{proposition}

\begin{proposition}
Let $x\sim Q_n^{\star}(\cdot|z)$, we have
$$
\mathbb{E}\left[d_{TV}(R^{\star}_n(\cdot|z,x_{ordered}),R_n(\cdot|z,x_{ordered}))|z\right] \leq 2d_{TV}(Q_n^{\star}(\cdot|z),Q_n(\cdot|z)).
$$
\label{cptlower2}
\end{proposition}
The Proposition \ref{cptlower} indicates that under the misspecification of conditional distribution, if we somehow unluckily choose bad test statistics $T_m$ and a bad type 1 error level $\alpha$, the deviation of type one error may be comparable to the deviation of conditional distribution. Though we cannot directly compare results in Proposition \ref{cptlower} to those in the Theorem 5 of \cite{berrett2019conditional}, the Proposition \ref{cptlower2} tells us that if $\mathbb{E}\left[d_{TV}(R^{\star}_n(\cdot|z,x_{ordered}),R_n(\cdot|z,x_{ordered}))|z\right]$ is large, so is $d_{TV}(Q_n^{\star}(\cdot|z),Q_n(\cdot|z))$, which means that if the type 1 error lower bound for the CPT blows up, so does the one for the CRT. Therefore, we can see that the CPT is more robust to the CRT in some sense.

To bridge the CPT and the CRRT, we need to introduce a test incorporating advantages of both, i.e. the conditional permutation rank test (CPRT). The CPRT generalize the CPT, just like that the CRRT is a generalization of the CRT. The CPRT shares a same pseudo variables generation scheme as the CPT. For any prescribed X-symmetric function 
$$
T(y,z,\tilde{x})=(T_0(y,z,\tilde{x}),\ldots,T_b(y,z,\tilde{x})):\mathbb{R}^{n}\times \mathbb{R}^{n\times p}\times\mathbb{R}^{n\times(b+1)}\rightarrow \mathbb{R}^{b+1},
$$
where $T$ is only required to be independent with $\tilde{x}$, we define the p-value for the CPRT,
$$
p_{CPRT} = \frac{1}{b+1}\sum\limits_{k=0}\limits^{b}\mathbbm{1}\{T_0(y,z,\tilde{x})\leq T_k(y,z,\tilde{x})\}.
$$
It is easy to know that the CPRT can also control the type 1 error if the conditional distribution of $X$ given $Z$ is correctly specified, which can be formalized as the following proposition. We skip its proof since the proof would be basically same as the one of Theorem \ref{maintheorem}.

\begin{proposition}
Assume that the null hypothesis $H_0:Y\perp X|Z$ holds. Suppose that $X\sim Q(\cdot|Z)$ is true. For $k=1,2,\ldots,b$, $x^{(k)}$'s are independently generated according to conditional distribution $R_n(\cdot|x_{ordered},z)$ induced by $Q$. Assume that $T(y,z,\tilde{x})$ is $X$-symmetric. Then, the CPRT can successfully control the type 1 error at any given level $\alpha\in[0,1]$:
$$
\mathbb{P}(p_{CPRT}\leq \alpha)\leq \alpha.
$$
\end{proposition}

\begin{figure}[h]
\centering
\includegraphics[width=12cm,height=6cm]{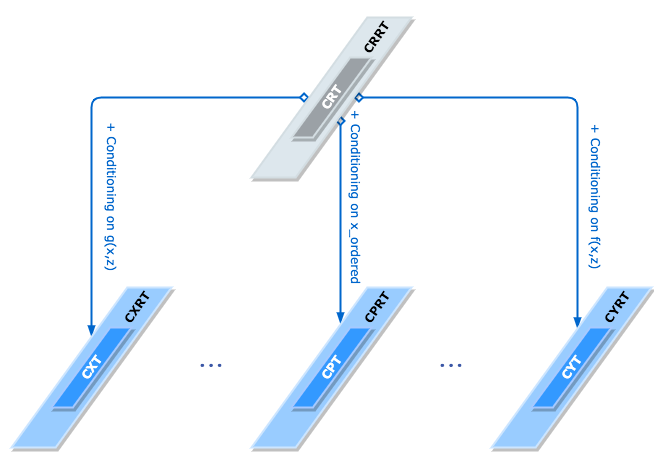}
\caption{Relationships among conditional tests based on pseudo variables. The CRT-CRRT only conditions on $z$. If we further condition on $x_{ordered}$, we get CPT-CPRT. We can also instead condition other statistics to generate pseudo variables, say $f(x,z)$ and $g(x,z)$, then we can have more choices like CXT-CXRT, CYT-CYRT and so on.}
\label{relation}
\end{figure}

The relationships among 4 conditional tests discussed here can be summarized by figure \ref{relation}. Actually, permutation is just one of several ways to enhance robustness at the expense of power loss. For example, we can also generate pseudo variables conditional on $\bar{x}$ and $z$, where $\bar{x}$ is the mean of $x$. In general, more conditioning we put, more robustness we gain. The extreme case is that we generate pseudo variables conditioning on $x$ and $z$. As a result, all $x^{(k)},k=1,\ldots,b$ will be identical with $x$ and the power will decrease to 0. 

When we try to construct a sampling-based conditional test, how much we need to condition on depends on how much we know about the conditional distribution of $X$ given $Z$. If we are confident that the proposed conditional distribution has slight bias or even no bias, we can try to be more aggressive and apply the CRRT.

\section{Empirical Results}
\label{sec:empirical}
In this section, we evaluate the performance of the CRRT proposed in this paper through comparison with other conditional sampling based methods, including CRT, CPT, dCRT, and HRT on simulated data. We also use concrete simulated examples to demonstrate the robustness of the CRRT.

\subsection{Some Preliminary Preparations and Explanations}
All methods considered in our simulations are conditional resampling based and thus can theoretically control the type 1 error regardless of the number of samplings we take. We know that though a small number of samplings $b$ is theoretically valid, a larger $b$ can prevent unnecessary conservativeness. At the same time, we do not want to take an extremely large $b$ since it will bring prohibitive computational burdens. After some preliminary analysis shown in the supplement, we decide to set $b=199$ across all simulations.

Note that the dCRT is a 2-step test. The first step is distillation, where we extract and compress information in $Y$ provided by $Z$. After distillation, the second step, which we call testing step, can be executed efficiently with a dataset of low dimension. For more details, see \cite{liu2020fast}. In the distillation step, it is optional to keep several important elements of $Z$ and incorporate them in the testing step. In such way, \cite{liu2020fast} shows that we can effectively deal with model with slight interactions. We denote the dCRT keeping $k$ elements of $Z$ by $d_kCRT$.  Actually, the dCRT is a special version of the CRT and still enjoys some flexibility because we can take any appropriate approaches in each of the distillation and testing step. In our simulation, we will explicitly specify what methods we use in each step. For example, when we say we use linear Lasso in the distillation step of $d_kCRT$, it means that we decide which $k$ elements would be kept based on the absolute fitted coefficients resulted from linear Lasso regression of $Y$ on $Z$ and use the Lasso regression residues as distillation residues. When we say we use least square linear regression in the testing step, it means that we run a simple linear regression of distillation residues on $X$ (or its pseudo copies) and important variables kept from the last step and use the corresponding absolute fitted coefficient of $X$ (or its pseudo copies) to construct test statistics. Unless specifying additionally, we use least square linear regression in the testing step by default.

The HRT (\cite{tansey2018holdout}) is also a 2-step test belonging to the family of the CRT. It splits data into training set and testing set. With the training set, we can fit a model, which can be arbitrary, of $Y$ on $X$ and $Z$. Then, we use the testing set to generate conditional samplings and feed all copies to the fitted model to obtain predicted errors. These errors work like the test function $T$. In our simulations, when we say we use logistic Lasso in the HRT, it means that we fit a logistic Lasso with the training set.

One-step test, like CRT, CRRT and CPT, would be more simple. As we can see, the test function $T$ is flexible. It would be better to choose suitable $T$'s in different settings. When we say we use the random forest as our test function, it means that we fit a random forest and obtain the feature importance of $X$ (or its pseudo copies) to construct test statistics. 

We propose a trick in implementing the CRRT. Suppose that $b=199$ and $n=400$, then $\tilde{x}$ would be a $400\times 200$ matrix. For some given positive integer $k$, say 4, we separate $\tilde{x}$ into 4 folds, as $\tilde{x}_1,\tilde{x}_2,\tilde{x}_3,\tilde{x}_4\in \mathbb{R}^{400\times 50}$. This step can be deterministic or random. For simplicity, in our simulations, we just let $\tilde{x}_i$ be a submatrix consisting of the $[(i-1)\frac{b}{k}+1]$th column to the $[\frac{bi}{k}]$th column of $\tilde{x}$, $i=1,\ldots,k$. Then, we apply a $X$-symmetric test function $T$ on each $(y,z,\tilde{x}_i)$ and obtain test statistics for $X$ and each pseudo variables. Finally, we can construct a p-value based on these test statistics. We call the CRRT equipped with such trick the $CRRT_k$ with mini-batch size $\frac{b}{k}$, or the $CRRT_k$ with $k$ folds. Though after such extra processing, the procedure is no longer a CRRT, it is not hard to show that $CRRT_k$ can promise type 1 error control. When $k=b$, the $CRRT_b$ degenerates to the $CRT$.

\subsection{The CRRT Shows a Huge Improvement in Computational Efficiency Over the CRT}
\label{sec:general} 
In this part, we assess the performance of all conditional sampling based tests we have mentioned in the setting of linear regression, which is frequently encountered in theoretical statistical analysis. In particular, we compare their ability to control type 1 error when the null hypothesis is true, their power when the null hypothesis is false, and their computational efficiency. Generally speaking, the CRRT and the dCRT are among the best because they are simultaneously efficient and powerful in our simulations.

We let $n=400$, $p=100$ and observations are i.i.d.. We perform 200 Monte-Carlo simulations for each specific setting. For $i=1,2,\ldots,400$, $\left(\begin{array}{c} x_i\\z_{[i]}\end{array}\right)$ is an i.i.d. realization of $\left(\begin{array}{c} X\\Z\end{array}\right)$, where $\left(\begin{array}{c} X\\Z\end{array}\right)$ is a $(p+1)$-dimension random vector following a Gaussian $AR(1)$ process with autoregressive coefficient 0.5. We let
$$
y_i = \beta_0 x_i + z_{[i]}^T\beta+\varepsilon_i,\ i=1,2,\ldots,400,
$$
where $\beta_0\in \mathbb{R}$, $\beta\in \mathbb{R}^{p}$, $\varepsilon_i$'s are i.i.d. standard Gaussian, independent of $x_i$ and $z_{[i]}$. We randomly sample a set of size 20 without replacement from $\{1,2,\ldots,p\}$, say $S$. For $j=1,\ldots,p$, let $\beta_j=0$ if $j\notin S$, and $\beta_j=0.5\mathbbm{1}\{B_j=1\}-0.5\mathbbm{1}\{B_j=0\}$ if $j\in S$, where $B_j$ is Bernoulli. We consider various values of $\beta_0$, including $0,0.05,0.1,\ldots,0.3$. 

\begin{figure}[ht]
    \centering
    \begin{subfigure}[t]{0.49\textwidth}
        \centering
        \includegraphics[width=1\linewidth,height=0.6\linewidth]{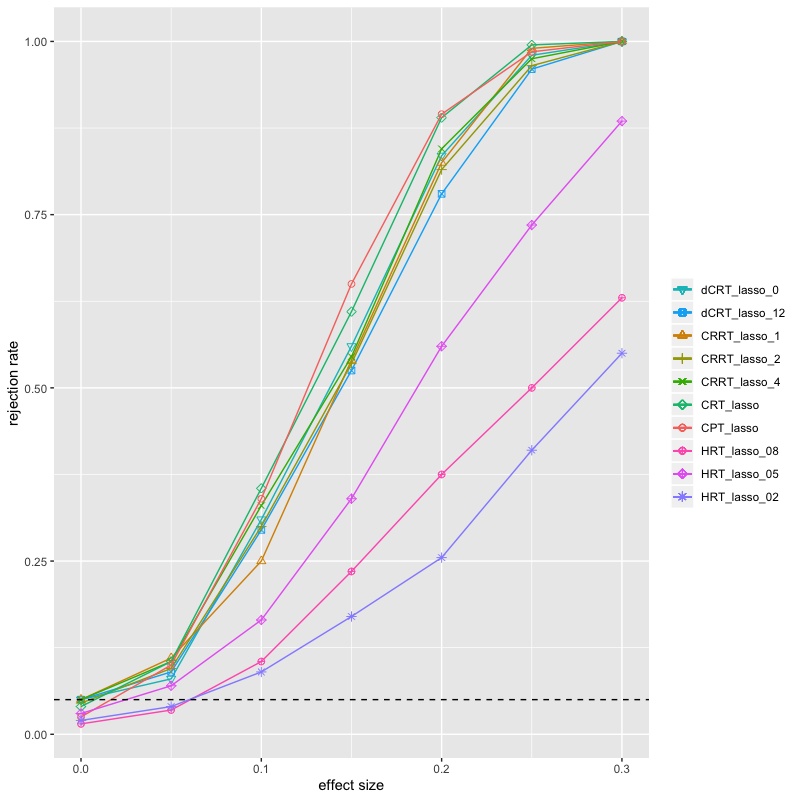} 
        \caption{Rejection Rate} \label{fig:gl1}
    \end{subfigure}
    \hfill
    \begin{subfigure}[t]{0.49\textwidth}
        \centering
        \includegraphics[width=1\linewidth,height=0.6\linewidth]{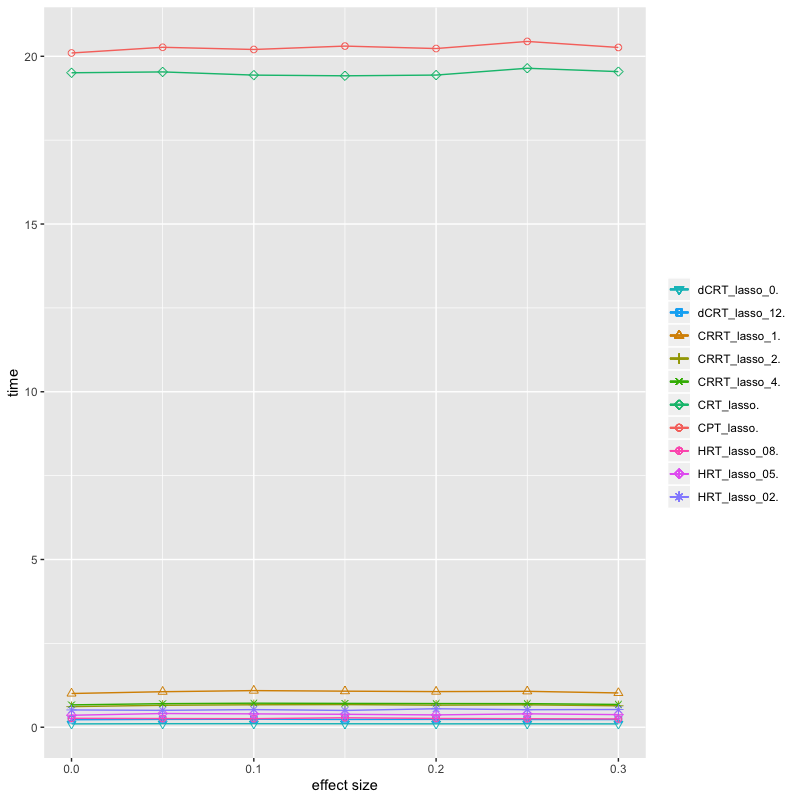} 
        \caption{Time} \label{fig:gl2}
    \end{subfigure}
\caption{(a) Proportions of rejection in the linear regression setting. The y-axis represents the proportion of times the null hypothesis is rejected. The x-axis represents the true coefficient of $X$. (b) Time in seconds per Monte-Carlo simulation. dCRT\_lasso\_k represents the dCRT with linear Lasso in the distillation step keeping k important variables after distillation. CRRT\_lasso\_k represents the CRRT using linear Lasso as the test function and with batch size $(b+1)/k$. HRT\_lasso\_0k represents the HRT fitting a linear Lasso model and using training set of size $n\times k\times 0.1$.}
\label{general_low}
\end{figure}

\noindent\textbf{Results:} Figures \ref{general_low}b exhibits a huge disadvantage of the CRT and the CPT in terms of the computational efficiency. They spend at least 10 times of time than other tests. It validates our assertion in the Section \ref{connect} that the CRRT has a great efficiency improvement over the CRT. We should also note that here we just try some simple methods like Lasso. If the methods we use become more complicated, we expect to see a greater computational efficiency discrepancy. 

From Figures \ref{general_low}a, we can see that basically all tests included in our simulations show good type 1 error control, which coincides with the theoretical results. When the null hypothesis is true, the HRT always shows markedly low rejection. Meanwhile, its power is significantly lower than other methods' when the null hypothesis is false. Though setting the proportion of training set to be $50\%$ can reduce conservativeness, it is still overly conservative in practice. The CRT and the CPT stably show greater power than other methods. However, such advantage is slight when compared with the CRRT and the dCRT. 

$d_0CRT$, $CRRT_2$ and $CRRT_4$ exhibit comparable performances in the sense of test effectiveness and computational efficiency. We should note that $CRRT_1$ does not always own computational advantage over $CRRT_k$ with $k>1$, which can be observed in the Figure \ref{general_low}b. The $d_{12}CRT$ can be viewed as a conservative version of the $d_0CRT$. It is no surprise that the $d_{12}CRT$ gain less power in the above simulations, compared with the $d_0CRT$ and some $CRRT$'s because there is no interactive effect and methods we use in the distillation step and the testing step basically suit with the nature of models, which make the important variables kept from the distillation step noisy instead. 

\subsection{The CRRT Outperforms Other Tests in Complicated Settings}
\label{sec:crrtdcrt} 
In practice, the model of the response variable on covariates is usually not as simple as what we set in the previous subsection. In this subsection, we consider complicated models and assess the performance of CRRT, dCRT and HRT. As impliied by the figure presented in the introduction, the CRT and the CPT spend far more time than the other methods when models get complicated, we do not include them in this subsection. Results show that only the CRRT can remain efficient and powerful in complicated settings.

We set $n=400$ and observations are i.i.d.. We try 200 Monte-Carlo simulations for each specific setting. For $i=1,2,\ldots,400$, $\left(\begin{array}{c} x_i\\z_{[i]}\end{array}\right)$ is an i.i.d. realization of $\left(\begin{array}{c} X\\Z\end{array}\right)$, where $\left(\begin{array}{c} X\\Z\end{array}\right)$ is a 101-dimension random vector following a Gaussian $AR(1)$ process with autoregressive coefficient 0.5. We consider the following response model: 
$$
\begin{array}{ll}
{\bf Interaction\ Model\ 1:}& y_i = \beta_0 x_i\sum\limits_{j=1}\limits^{100}z_{[i]j} + z_{[i]}^T\beta+\varepsilon_i,\ i=1,2,\ldots,400,\\
\end{array}
$$
where $\beta_0\in \mathbb{R}$, $\beta\in \mathbb{R}^{100}$, $\varepsilon_i$'s are i.i.d. standard Gaussian, independent of $x_i$ and $z_{[i]}$. To construct $\beta$, we randomly sample 5 indexes without replacement from $\{1,2,\ldots,100\}$, say $S$. For $j=1,\ldots,100$, let $\beta_j=0$ if $j\notin S$, and $\beta_j=0.5\mathbbm{1}\{B_j=1\}-0.5\mathbbm{1}\{B_j=0\}$ otherwise, where $B_j$ is Bernoulli.  We consider a range of choices of $\beta_0$. When $\beta_0=0$, $Y$ is a function of $Z$ and random errors in all 3 models, which implies that the null hypothesis of $Y\perp X|Z$ holds. When $\beta_0\ne 0$, $X$ is no longer independent of $Y$ conditional on $Z$. We also consider 2 other models with categorical responses, the results of which are provided in the supplement. 

\begin{figure}[h]
    \centering
    \begin{subfigure}[t]{0.49\textwidth}
        \centering
        \includegraphics[width=1\linewidth,height=0.6\linewidth]{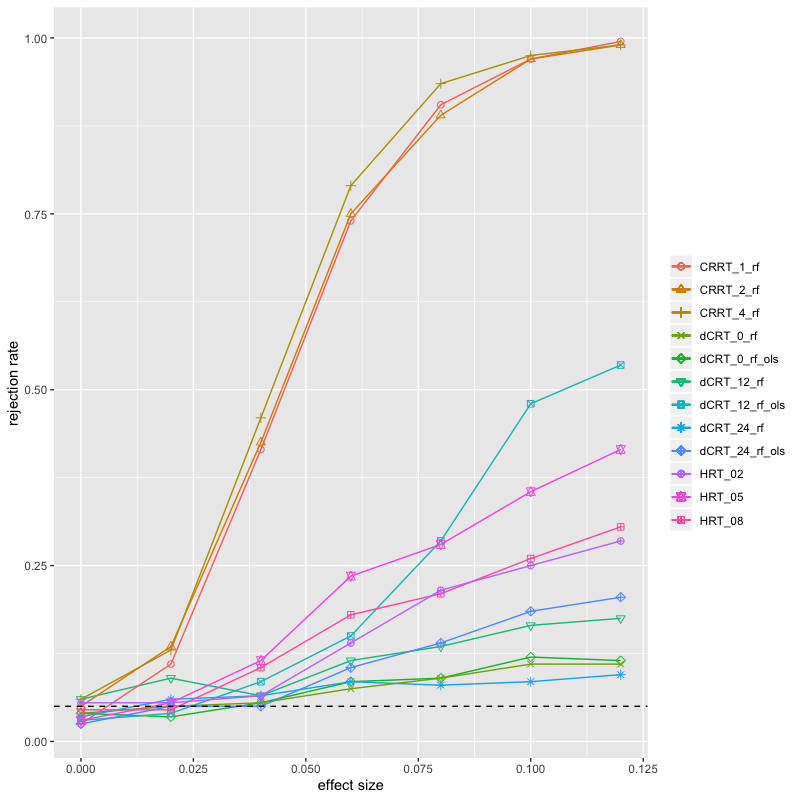} 
        \caption{Rejection Rate} \label{fig:ir1}
    \end{subfigure}
    \hfill
    \begin{subfigure}[t]{0.49\textwidth}
        \centering
        \includegraphics[width=1\linewidth,height=0.6\linewidth]{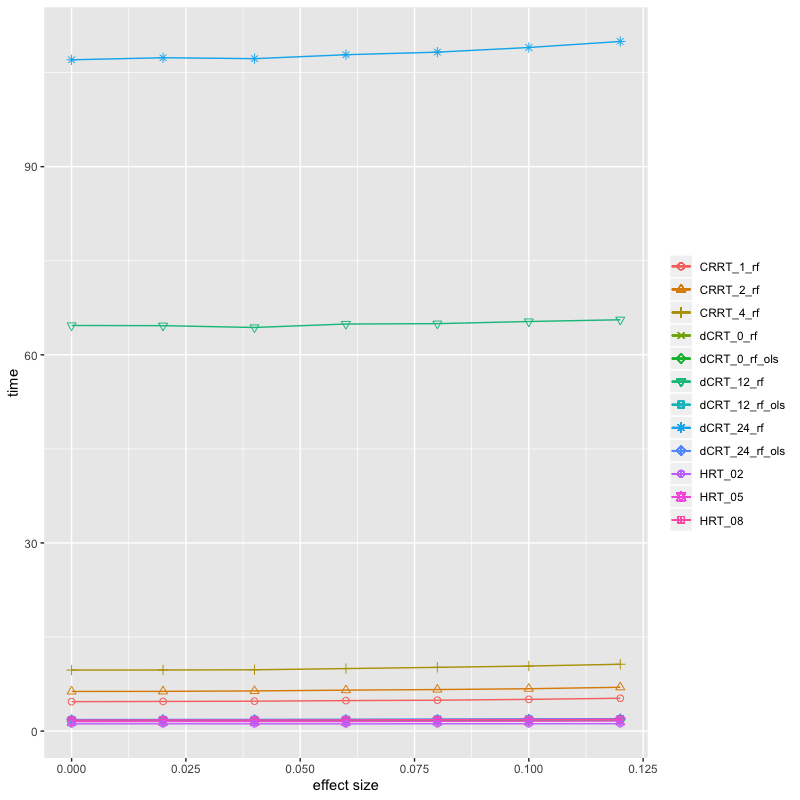} 
        \caption{Time} \label{fig:it1}
    \end{subfigure}
\caption{(a) Proportions of rejection under the Interaction Model 1. The y-axis represents the proportion of times the null hypothesis is rejected. The x-axis represents the true coefficient of $X$. (b) Time in seconds per Monte-Carlo simulation. dCRT\_k\_rf represents the dCRT with random forest in both of the distillation step and the testing step, and keeping k important variables after distillation. dCRT\_k\_rf\_ols represents the dCRT with random forest in the distillation step, least square linear regression in the testing step. CRRT\_k\_rf represents the CRRT using random forest as the test function and with batch size $(b+1)/k$. HRT\_0k represents the HRT fitting a random forest model and using training set of size $n\times k\times 0.1$.}
\label{inter_cont}
\end{figure}

\noindent {\bf Results:} Based on the Figure \ref{inter_cont}a, we can see that the CRRT's have dominating power over the dCRT's and the HRT's. In particular, the CRRT with batch size 50 (i.e. CRRT\_4\_rf) has the most power among all tests considered here. Basically all tests can control the observed type 1 error when the effect size is 0. Figure \ref{inter_cont}b shows that the CRRT's do not sacrifice much computational efficiency. It is worth mentioning that $d_{24}CRT$ does not gain more power than $d_{12}CRT$, which demonstrate that the strategy of keeping important features after the distillation step could fail to alleviate the problem brought by interactions when the model is complicated.

\subsection{The CRRT Is More Robust to Conditional Distribution Misspecification Than the dCRT}
\label{robustsim}
We prove in the Section \ref{sec:robust} that the CRRT is robust to misspecification of conditional density, i.e. $Q(\cdot|Z)\ne Q^{\star}(\cdot|Z)$. In this subsection, we demonstrate it with simulations covering 3 sources of misspecification, including the misspecification from non data-dependent estimation, from using unlabeled data and from reusing labeled data. Note that when the misspecification originates from reusing labeled data, currently there is no theoretical results on robustness. We consider the CRRT and the dCRT, which shows relatively good performance in previous experiments.

\subsubsection{Misspecification From Non Data-Dependent Estimation}
One source of conditional distribution misspecification comes from domain knowledge or some conventions, which does not depend on labeled data or unlabeled data. We consider the following 3 response models. For each specific setting, we run 500 Monte-Carlo simulations.

\begin{figure}[ht!]
    \centering
    \begin{subfigure}[t]{0.49\textwidth}
        \centering
        \includegraphics[width=1\linewidth,height=0.6\linewidth]{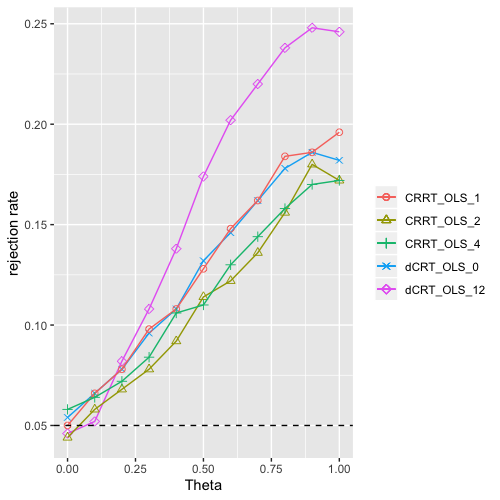} 
        \caption{$n>p$ Linear Regression, OLS-Based Tests} \label{fig:nd1}
    \end{subfigure}
    \hfill
    \begin{subfigure}[t]{0.49\textwidth}
        \centering
        \includegraphics[width=1\linewidth,height=0.6\linewidth]{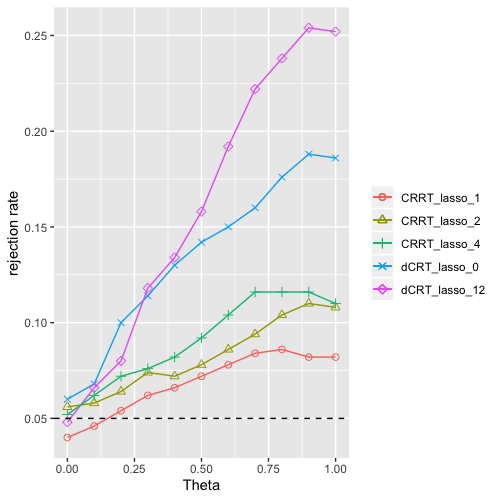} 
        \caption{$n>p$ Linear Regression, Lasso-Based Tests} \label{fig:nd2}
    \end{subfigure}
    \begin{subfigure}[t]{0.49\textwidth}
        \centering
        \includegraphics[width=1\linewidth,height=0.6\linewidth]{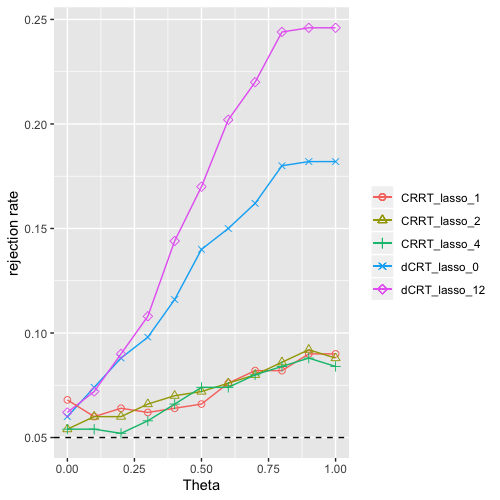} 
        \caption{$n=p$ Linear Regression, Lasso-Based Tests} \label{fig:nd3}
    \end{subfigure}
    \hfill
    \begin{subfigure}[t]{0.49\textwidth}
        \centering
        \includegraphics[width=1\linewidth,height=0.6\linewidth]{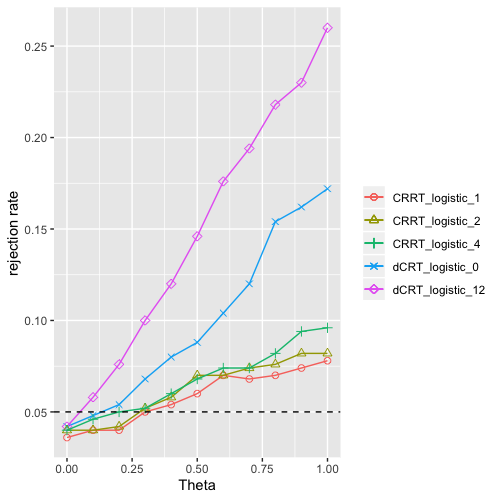} 
        \caption{$n>p$ Logistic Regression, Lasso-Based Tests} \label{fig:nd4}
    \end{subfigure}
\caption{The CRRT is less sensitive to non data-dependent misspecification of conditional distribution than the dCRT. The y-axis represents the proportion of times the null hypothesis of conditional independence is rejected out of 500 simulations. The x-axis is the value of $\theta$, representing the degree of misspecification. CRRT\_OLS\_k, CRRT\_lasso\_k and CRRT\_logistic\_k  represent the CRRT's with respectively least square linear regression, linear Lasso and logistic Lasso as the test function and batch size $(b+1)/k$. dCRT\_OLS\_k, dCRT\_lasso\_k and dCRT\_logistic\_k represents the dCRT's respectively using least square linear regression, linear Lasso and logistic Lasso in the distillation step with k important variables kept after distillation.}
\label{nondata}
\end{figure}

\noindent {\bf ${\bf n>p}$ linear regression:} We set $n=400$, $p=100$ and observations are i.i.d.. For $i=1,2,\ldots,400$, $\left(\begin{array}{c} x_i\\z_{[i]}\end{array}\right)$ is an i.i.d. realization of $\left(\begin{array}{c} X\\Z\end{array}\right)$, where $\left(\begin{array}{c} X\\Z\end{array}\right)$ is a $(p+1)$-dimension random vector. $Z\in \mathbb{R}^{p}$ follows a Gaussian $AR(1)$ process with autoregressive coefficient 0.5 and $X|Z\sim N(\mu(Z),\sigma^2)$, where $\mu$ and $\sigma^2$ will be specified shortly. We let
$$
y_i = 0\cdot x_i + z_{[i]}^T\beta+\varepsilon_i,\ i=1,2,\ldots,400,
$$
where $\beta\in \mathbb{R}^{p}$, $\varepsilon_i$'s are i.i.d. standard Gaussian, independent of $x_i$ and $z_{[i]}$. $\beta$ is constructed as follows. We randomly sample a set of size 20 without replacement from $\{1,2,\ldots,p\}$, say S. For $j=1,\ldots,p$, let $\beta_j=0$ if $j\ne S$, and $\beta_j=0.5\mathbbm{1}\{B_j=1\}-0.5\mathbbm{1}\{B_j=0\}$ if $j\in S$, where $B_j$ is Bernoulli. 

\noindent {\bf ${\bf n=p}$ linear regression:} Basically everything is same as the $n>p$ linear regression setting, except that the dimension of $Z$ is increased from 100 to 400 while the sparsity of $\beta$ remains to be 20.

\noindent {\bf ${\bf n>p}$ logistic regression:} $x$ and $z$ are generated in the same way as the previous model, where $p=100$. We consider a binary model,
$$
y_i = \mathbbm{1}\{0\cdot x_i + z_{[i]}^T\beta+\varepsilon_i\},\ i=1,2,\ldots,400,
$$
where $\beta\in \mathbb{R}^{p}$. $\varepsilon_i$'s and $\beta$ are constructed in a same way as the $n>p$ linear regression model.

Suppose that $\left(\begin{array}{c} \mathring{X}\\Z\end{array}\right)\in \mathbb{R}^{p+1}$ is from a Gaussian AR(1) process with autoregressive coefficient 0.5. The conditional distribution of $\mathring{X}$ given $Z$ must be Gaussian in the form of $N(Z^T\zeta,\mathring{\sigma}^2)$. For $\theta\in [0,1]$, we let $\mu(Z)$ and $\sigma^2$ mentioned above be $\theta Z^T\zeta$ and $\frac{1+(1-\theta)^2}{2}\mathring{\sigma}^2$ respectively. For $k=1,2,\ldots,b$, we set the conditional distribution of pseudo variable $X^{(k)}$ given $Z$ be $N(0,\mathring{\sigma}^2)$. When $\theta$ is exactly 0, the conditional distribution of $X$ given $Z$ would be also $N(0,\mathring{\sigma}^2)$ and there is no misspecification. As the value of $\theta$ increases, the degree of misspecification increases.

\noindent {\bf Results:} Figures \ref{fig:nd1} and \ref{fig:nd2} display results for the $n>p$ linear regression model. Increasing misspecification of conditional distribution can inflate the type 1 error. Besides, it can be seen that the CRT's are always more robust than corresponding dCRT's with same type of test function. In particular, the CRRT's with linear Lasso as the test function are the least sensitive. Even when $\theta=1$, the CRRT\_lasso\_1 can still control the type 1 error below 0.1 while the dCRT\_lasso\_12 has a type 1 error reaching 0.25.

When we turn to the Figure \ref{fig:nd3}, where results for the $n=p$ linear regression model are shown, we can see that the CRRT's show a more pronounced advantage in robustness to misspecification over the dCRT's. The performances of the CRRT's vary little among different choices of batch size.

From the Figure \ref{fig:nd4}, where results for the $n>p$ logistic regression model are shown, we can see that the CRRT's still behave more robustly than the dCRT's and successfully control the observed type 1 error under 0.1 as in the previous setting.

In summary, the CRRT is more robust to the current misspecification type than the dCRT across several models. Besides, the CRRT has comparable computational efficiency as the dCRT, which is not shown directly here but can be inferred from previous results since the misspecification has no impact on computation. We should note that, unlike the monotonic trend shown here, the misspecification of conditional distribution may have involved influence on the type 1 error. For more details, see the supplement.

\subsubsection{Misspecification From Using Unlabeled Data}
Estimating the conditional distribution by using unlabeled data is an effective  and popular approach without sacrificing theoretical validity (\cite{candes2018panning}, \cite{huang2019relaxing}, \cite{romano2019deep}, \cite{sesia2019gene}). Being able to utilize unlabeled data is also a main reason why model-X methods become more and more attractive. Usually, we propose a parametric family for the target conditional distribution and estimates parameters with unlabeled data. When the proposed parametric family is correct, the estimated conditional distribution would be more closed to the true one as we have unlabeled data of larger size. Of course, we can also consider to use unlabeled data in a nonparametric way with some kernel methods. But here, we only show simulations of parametric estimations. 

We still consider same response models as in the previous part, i.e. $n>p$ linear regression, $n=p$ linear regression and $n>p$ logistic regression. Here, we fix $\mu(Z)=Z^T\zeta$ and $\sigma^2=\mathring{\sigma}^2$ such that $\left(\begin{array}{c} X\\Z\end{array}\right)$ follows a Gaussian AR(1) process with autoregressive coefficient 0.5. However, we will not use the exact conditional distribution $Q^{\star}(\cdot|Z)$ to generate pseudo variables. Instead, we assume a Gaussian model and get estimators $\hat{\zeta}$ and $\hat{\sigma}^2$ with unlabeled dataset $(x_{unlabeled},z_{unlabeled})$, where $x_{unlabeled} \in \mathbb{R}^N$ and $z_{unlabeled}\in \mathbb{R}^{N\times p}$ by using scaled Lasso (\cite{sun2012scaled}). Note that each row of the unlabeled dataset is a realization of $\left(\begin{array}{c} X\\Z\end{array}\right)$. Then, we generate pseudo variables based on the estimated conditional distribution $N(Z^T\hat{\zeta},\hat{\sigma}^2)$.

\begin{figure}[ht!]
    \centering
    \begin{subfigure}[t]{0.49\textwidth}
        \centering
        \includegraphics[width=1\linewidth,height=0.6\linewidth]{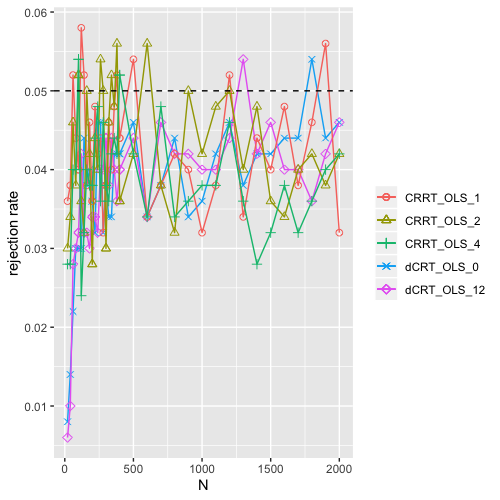} 
        \caption{$n>p$ Linear Regression, OLS-Based Tests} \label{fig:ul1}
    \end{subfigure}
    \hfill
    \begin{subfigure}[t]{0.49\textwidth}
        \centering
        \includegraphics[width=1\linewidth,height=0.6\linewidth]{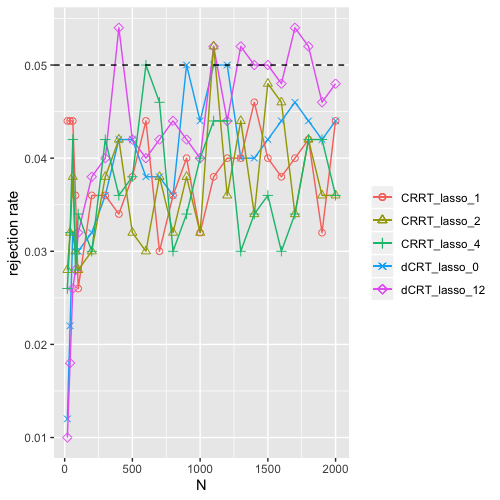} 
        \caption{$n>p$ Linear Regression, Lasso-Based Tests} \label{fig:ul2}
    \end{subfigure}
    \begin{subfigure}[t]{0.49\textwidth}
        \centering
        \includegraphics[width=1\linewidth,height=0.6\linewidth]{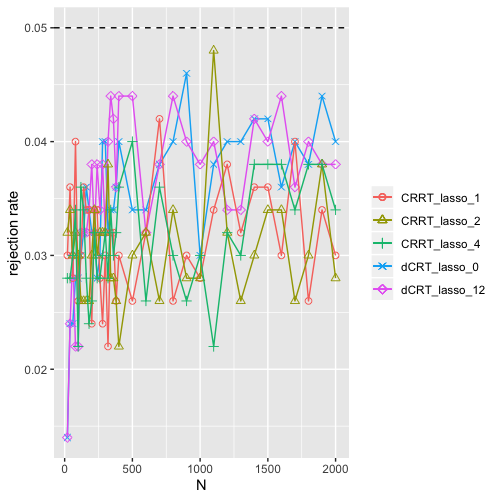} 
        \caption{$n=p$ Linear Regression, Lasso-Based Tests} \label{fig:ul3}
    \end{subfigure}
    \hfill
    \begin{subfigure}[t]{0.49\textwidth}
        \centering
        \includegraphics[width=1\linewidth,height=0.6\linewidth]{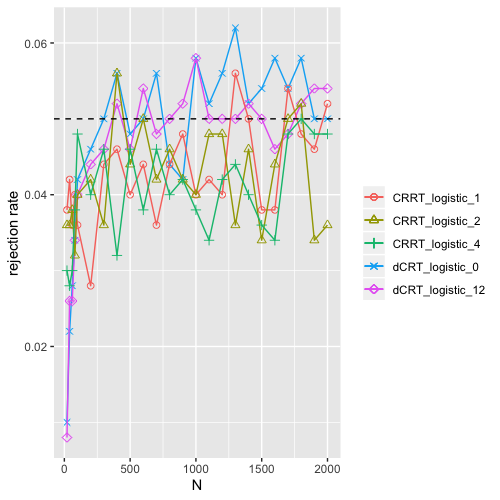} 
        \caption{$n>p$ Logistic Regression, Lasso-Based Tests} \label{fig:ul4}
    \end{subfigure}
\caption{The CRRT is less sensitive to misspecification of conditional distribution originated from unlabeled data estimation than the dCRT. The y-axis represents the proportion of times the null hypothesis of conditional independence is rejected out of 500 simulations. The x-axis represents the sample size of unlabeled data. Definitions of legends can be found in the Figure \ref{nondata}.}
\label{mis:unlabel}
\end{figure}

\noindent {\bf Results:} The plots in the Figure \ref{mis:unlabel} exhibit en encouraging phenomenon that both of the CRRT and the dCRT can basically control the type 1 error under 0.05 when there is a misspecification from unlabeled data estimation. The CRRT apparently has a more stable performance than the dCRT. The dCRT could be too conservative when $N$, the sample size of unlabeled data is relatively small. Besides, in the $n>p$ logistic regression setting, the dCRT's have observed type 1 error slightly exceeding the threshold.

\subsubsection{Misspecification From Reusing Data}
Though losing theoretical guarantee, when we do not have abundant unlabeled data, we can estimate the conditional distribution by reusing data. We consider similar procedures as in the last part except that we estimate $Q(\cdot|Z)$ on labeled data instead of unlabeled data. In such case, $Q$ is no longer independent with our data for constructing test statistics and hence theoretical results on robustness proposed before fail to work here.

\begin{figure}[ht]
    \centering
    \begin{subfigure}[t]{0.49\textwidth}
        \centering
        \includegraphics[width=1\linewidth,height=0.6\linewidth]{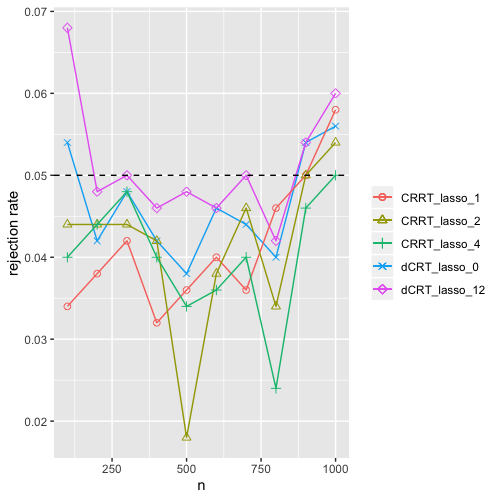} 
        \caption{Linear} \label{fig:label1}
    \end{subfigure}
    \hfill
    \begin{subfigure}[t]{0.49\textwidth}
        \centering
        \includegraphics[width=1\linewidth,height=0.6\linewidth]{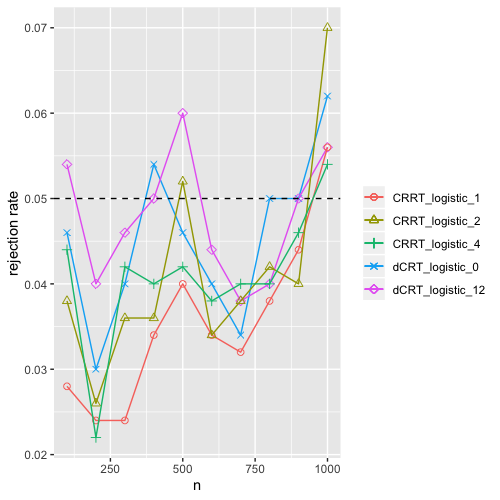} 
        \caption{Logistic} \label{fig:label2}
    \end{subfigure}
\caption{The CRRT and the dCRT have similar performance under misspecification of conditional distribution originated from reusing labeled data. The y-axis represents the proportion of times the null hypothesis of conditional independence is rejected out of 500 simulations. The x-axis represents the sample size of labeled data. Definitions of legends can be found in the Figure \ref{nondata}.}
\label{mis:label}
\end{figure}

\noindent {\bf Results:} We consider linear regression model and logistic regression model, with results shown in the Figure \ref{mis:label}. It is hard to tell whether the CRRT outperforms the dCRT here. They have similar performance while the dCRT is relatively more aggressive. Good news is that, both of these 2 types of test can basically control the observed type 1 error around the desired level. We can also see that their performances are not dependent on $n$, the sample size of data.

\section{Discussions}
\label{sec:discuss}
This paper propose the CRRT, a new method for testing conditional independence, generalizing the CRT proposed in \cite{candes2018panning}. It is more flexible, more efficient and comparably powerful with theoretically well grounded type 1 error control and robustness guarantee. We only discuss the case when the response variable $Y$ is univariate. But, we can trivially extend results presented in this paper to multi-dimensional settings. 

Again, conditional independence is not a testable test if the control variable is continuous (\cite{bergsma2004testing}, \cite{shah2018hardness}). Therefore, we need some pre-known information on the covariates model or the response model. The CRRT focuses on the former, like other model-X methods. However, usually knowledge concerning the response model is not completely absent. It would be pragmatically attractive to develop conditional sampling based methods that can incorporate information from both of the covariates model side and the response model side.

Reusing data is a commonly used approach in practice though usually leading to overfitting. It would be interesting to develop theories to theoretically justify the use of the CRRT combined with this approach, in terms of type 1 error control and robustness. Splitting data is a popular option when we have to perform estimation and testing on a same copy of data (\cite{wasserman2009high}). However, it is well known that simple splitting can cause tremendous power loss. Some cross-splitting approaches are proposed as remedies (\cite{chernozhukov2018double}, \cite{barber2019knockoff}). It seems promising to combine these approaches with the CRRT.


\nocite{*}
{\small \bibliography{Conditional_Randomization_Rank_Test}}

\bibliographystyle{alpha}

\clearpage
\appendix
\noindent\textbf{\Large Appendix}
\section{Supplementary Analysis on the Robustness of CRRT}
\label{supp_robust}

\subsection{Continuous Analysis on the Example 2}

Here, we assume that $\lim\limits_{n\rightarrow \infty}b = \infty$. According to the result provided in Theorem 3, for any measurable functions of $z$, $\varepsilon_1(z),\ldots,\varepsilon_b(z)$, the following inequality holds,
\begin{equation}
\mathbb{P}\left(P_{CRRT}\leq \alpha\right) \leq \mathbb{E}\left[\frac{\lfloor (b+1)\alpha\rfloor}{1+e^{-\varepsilon_1}+\ldots+e^{-\varepsilon_b}}\right] + \mathbb{P}\left(\bigcup\limits_{k=1}\limits^b\left\{M_{(k)}-M_0>\varepsilon_k\right\}\right).
\label{b1}
\end{equation}
We work with the second term on the righthand side firstly:
\begin{equation}
\begin{array}{ll}
& \mathbb{P}\left(\bigcup\limits_{k=1}\limits^b\left\{M_{(k)}-M_0>\varepsilon_k\right\}\right) = \mathbb{E}\left[\mathbb{P}\left(\bigcup\limits_{k=1}\limits^b\left\{M_{(k)}-M_0>\varepsilon_k\right\}|z\right)\right]\\

= &  \mathbb{E}\left[\mathbbm{1}\left\{|M_0|\leq \sqrt{\frac{np^2\log p}{N}}f(n)\right\}\mathbb{P}\left(\bigcup\limits_{k=1}\limits^b\left\{M_{(k)}-M_0>\varepsilon_k\right\}|z\right)\right]\\

& + \mathbb{E}\left[\mathbbm{1}\left\{|M_0|> \sqrt{\frac{np^2\log p}{N}}f(n)\right\}\mathbb{P}\left(\bigcup\limits_{k=1}\limits^b\left\{M_{(k)}-M_0>\varepsilon_k\right\}|z\right)\right]\\

\leq & \mathbb{E}\left[\mathbb{P}\left(\bigcup\limits_{k=1}\limits^b\left\{M_{(k)}>\varepsilon_k-\sqrt{\frac{np^2\log p}{N}}f(n)\right\}|z\right)\right] + \mathbb{P}\left(|M_0|> \sqrt{\frac{np^2\log p}{N}}f(n)\right)\\

= & \mathbb{E}\left[\mathbb{P}\left(\bigcup\limits_{k=1}\limits^b\left\{M_{(k)}>\varepsilon_k-\sqrt{\frac{np^2\log p}{N}}f(n)\right\}|z\right)\right] + o(1),\\
\end{array}
\label{b2}
\end{equation}
where the third step is based on that $\sqrt{\frac{N}{np^2\log p}}M_0=O_p(1)$. Next, we have
\begin{equation}
\begin{array}{ll}
& \mathbb{P}\left(\bigcup\limits_{k=1}\limits^b\left\{M_{(k)}>\varepsilon_k-\sqrt{\frac{np^2\log p}{N}}f(n)\right\}|z\right) \\

\leq & \mathbb{P}\left(\bigcup\limits_{k=1}\limits^b\left\{\sum\limits_{i=1}\limits^b\mathbbm{1}\left\{M_i\leq \varepsilon_k-\sqrt{\frac{np^2\log p}{N}}f(n)\right\}\leq k\right\}|z\right)\\

= & \mathbb{P}\left(\bigcup\limits_{k=1}\limits^b\left\{\frac{1}{b}\sum\limits_{i=1}\limits^b\mathbbm{1}\left\{M_i\leq \varepsilon_k-\sqrt{\frac{np^2\log p}{N}}f(n)\right\} \leq \frac{k}{b} \right\}|z\right)\\

= & \mathbb{P}\left(\bigcup\limits_{k=1}\limits^b\left\{F_b\left(\varepsilon_k- \sqrt{\frac{np^2\log p}{N}}f(n)\right) - F\left(\varepsilon_k- \sqrt{\frac{np^2\log p}{N}}f(n)\right) \right. \right.\\
& \left.\left. \leq \frac{k}{b}-F\left(\varepsilon_k- \sqrt{\frac{np^2\log p}{N}}f(n)\right)\right\}|z\right),\\
\end{array}
\label{b3}
\end{equation}
where $F_b$ is the empirical CDF of $M_1|z$ and $F$ is the true CDF of $M_1|z$. Denote $\eta\triangleq \frac{\log b}{b}$, $\delta\triangleq e^{-\frac{h^2(n)}{2}}$, where $h(n)$ satisfies that $\lim\limits_{n\rightarrow \infty}h(n)=\infty$ and $\lim\limits_{n\rightarrow \infty}\sqrt{\frac{np^2\log p}{N}}f(n)h(n)=0$, and set
$$
\left\{
\begin{array}{ll}
\varepsilon_k=F^{-1}\left(\frac{k}{b}+\eta\right)+ \sqrt{\frac{np^2\log p}{N}}f(n), & if\ 1\leq k\leq b(1-\eta-\delta),\\
\varepsilon_k = \infty, & if\ b(1-\eta-\delta)<k\leq b.
\end{array}
\right.
$$
Under such settings, $\frac{k}{b}-F\left(\varepsilon_k- \sqrt{\frac{np^2\log p}{N}}f(n)\right)=-\eta$ when $1\leq k\leq b(1-\eta-\delta)$. Hence, based on (\ref{b3}), we have
\begin{equation}
\resizebox{.90\hsize}{!}{$
\begin{array}{ll}
& \mathbb{P}\left(\bigcup\limits_{k=1}\limits^b\left\{M_{(k)}>\varepsilon_k-\sqrt{\frac{np^2\log p}{N}}f(n)\right\}|z\right) \\

= & \mathbb{P}\left(\bigcup\limits_{k=1}\limits^{\lfloor b(1-\eta-\delta)\rfloor}\left\{F_b\left(\varepsilon_k- \sqrt{\frac{np^2\log p}{N}}f(n)\right) - F\left(\varepsilon_k- \sqrt{\frac{np^2\log p}{N}}f(n)\right) \leq -\eta\right\}|z\right)\\

\leq & \mathbb{P}\left(\inf\limits_{x\in\mathbb{R}}\left\{F_b(x)-F(x)\right\}\leq -\eta|z \right) \leq e^{-2b\eta^2},\\
\end{array}
$}
\label{b4}
\end{equation}
where the last step is based on the Dvoretzky–Kiefer–Wolfowitz inequality. Combining results in (\ref{b4}) and (\ref{b2}), we know that $\mathbb{P}\left(\bigcup\limits_{k=1}\limits^b\left\{M_{(k)}-M_0>\varepsilon_k\right\}\right) = o(1)$ as $n\rightarrow \infty$. Denote $D\triangleq \frac{1}{\sigma^2}(\hat{\beta}-\beta)^Tz^Tz(\hat{\beta}-\beta)$. Then, for $1\leq k \leq b(1-\eta-\delta)$,
\begin{equation}
\resizebox{.90\hsize}{!}{$
\begin{array}{ll}
& \mathbb{P}\left(\frac{M_1-D}{\sqrt{D}}>\sqrt{-2\log\left(1-\frac{k}{b}-\eta\right)}|z\right)\leq e^{-\frac{1}{2}[-2\log (1-\frac{k}{b}-\eta)]}=1-\frac{k}{b}-\eta,\\
\Rightarrow & \frac{k}{b}+\eta \leq \mathbb{P}\left(\frac{M_1-D}{\sqrt{D}}\leq \sqrt{-2\log\left(1-\frac{k}{b}-\eta\right)}|z\right) = F\left( \sqrt{-2D\log\left(1-\frac{k}{b}-\eta\right)}+D\right),\\
\Rightarrow & F^{-1}\left(\frac{k}{b}+\eta\right) \leq \sqrt{-2D\log\left(1-\frac{k}{b}-\eta\right)}+D,\\
\Rightarrow &  \varepsilon_k \leq  \sqrt{-2D\log\left(1-\frac{k}{b}-\eta\right)}+D + \sqrt{\frac{np^2\log p}{N}}f(n),\\
\end{array}
$}
\label{b5}
\end{equation}
where the first result is based on the concentration inequality. Under the current settings of $\varepsilon_k$'s, we have
\begin{equation}
\resizebox{.90\hsize}{!}{$
\begin{array}{ll}
 & \mathbb{E}\left[\frac{\lfloor (b+1)\alpha\rfloor}{1+e^{-\varepsilon_1}+\ldots+e^{-\varepsilon_b}}\right]= \mathbb{E}\left[\frac{\lfloor (b+1)\alpha\rfloor}{1+\sum\limits_{k=1}\limits^{\lfloor b(1-\eta-\delta)\rfloor }e^{-\varepsilon_k}}\right]\\
 
\leq &  \mathbb{E}\left[\frac{\lfloor (b+1)\alpha\rfloor}{1+\sum\limits_{k=1}\limits^{\lfloor b(1-\eta-\delta)\rfloor}exp\left\{-\left[\sqrt{-2D\log\left(1-\frac{k}{b}-\eta\right)}+D + \sqrt{\frac{np^2\log p}{N}}f(n)\right]\right\}}\right]\\

\leq & \mathbb{E}\left[\frac{\lfloor (b+1)\alpha\rfloor}{1+\sum\limits_{k=1}\limits^{b(1-\eta-\delta)}exp\left\{-\left[\sqrt{-2\frac{np^2\log p}{N}f^2(n)\log\left(1-\frac{k}{b}-\eta\right)}+\frac{np^2\log p}{N}f^2(n) + \sqrt{\frac{np^2\log p}{N}}f(n)\right]\right\}}\right]+o(1)\\

\leq & \mathbb{E}\left[\frac{\lfloor (b+1)\alpha\rfloor}{1+\sum\limits_{k=1}\limits^{b(1-\eta-\delta)}exp\left\{-\left[\sqrt{\frac{np^2\log p}{N}f^2(n)h^2(n)}+\frac{np^2\log p}{N}f^2(n) + \sqrt{\frac{np^2\log p}{N}}f(n)\right]\right\}}\right]+o(1)\\

\rightarrow & \alpha,
\end{array}
$}
\label{b6}
\end{equation}
as $n\rightarrow \infty$. We can see that $F$ depends only on $z$ and therefore $\varepsilon_k$'s are independent of $x$ and $y$ conditional on $z$. Thus, combining results in (\ref{b6}), (\ref{b1}) and that $\mathbb{P}\left(\bigcup\limits_{k=1}\limits^b\left\{M_{(k)}-M_0>\varepsilon_k\right\}\right) = o(1)$ as $n\rightarrow \infty$, we can conclude that $\mathbb{P}(P_{CRRT}\leq \alpha)\rightarrow \alpha$ as $n\rightarrow \infty$.


\subsection{The Impact of Large $b$, the Number of Conditional Samplings}
For convenience, we restate the last bound in the Theorem \ref{thm:robust2} here:
\begin{equation}
\mathbb{P}(P_{CRRT}\leq \alpha) \leq \inf\limits_{\varepsilon_1,\ldots,\varepsilon_b\ge 0} \left\{\frac{\lfloor (b+1)\alpha\rfloor}{1+e^{-\varepsilon_1}+\ldots+e^{-\varepsilon_b}}+ \mathbb{P}\left(\bigcup\limits_{k=1}\limits^{b}\{\widehat{KL}_{(k)}>\varepsilon_k\}\right)\right\}.
\label{klboundd}
\end{equation}
It appears that $\mathbb{P}\left(\bigcup\limits_{k=1}\limits^{b}\{\widehat{KL}_{(k)}>\varepsilon_k\}\right)$ is monotonically increasing with $b$. It is true when all $\varepsilon_1,\ldots,\varepsilon_b$ are identical and fixed. Such phenomenon is counterintuitive and may make this upper bound meaningless in practice. Luckily, $\varepsilon_1,\ldots,\varepsilon_b$ are adjustable and can vary with $b$. We roughly show that for any given small $\epsilon>0$, there exist a function $g(Q,Q^{\star},\alpha)$ such that
$$
\resizebox{.92\hsize}{!}{$
\inf\limits_{\varepsilon_1,\ldots,\varepsilon_b\ge 0} \left\{\frac{\lfloor (b+1)\alpha\rfloor}{1+e^{-\varepsilon_1}+\ldots+e^{-\varepsilon_b}}+ \mathbb{P}\left(\bigcup\limits_{k=1}\limits^{b}\{\widehat{KL}_{(k)}>\varepsilon_k\}\right)\right\} \leq g(Q,Q^{\star},\alpha) + o(1),\ as\ b\rightarrow \infty,
$}
$$
and
$$
g(Q,Q^{\star},\alpha) = \alpha(1+\epsilon),\ if\ Q=Q^{\star}.
$$
It's true that the above conclusion cannot demonstrate that the bound in (\ref{klboundd}) is monotonically decreasing with $b$. But, at least it implies that large $b$ won't have severe negative impact when the conditional distribution is misspecified. We don't intend to give a thorough and rigorous proof. We just roughly sketch our ideas here. For $k=1,\ldots,b$, let
$$
\widehat{KL}_k=\sum\limits_{i=1}\limits^n\log \frac{Q^{\star}(x_i|z_{[i]})Q(x_i^{(k)}|z_{[i]})}{Q(x_i|z_{[i]})Q^{\star}(x_i^{(k)}|z_{[i]})} \triangleq \bar{T}_k - \bar{T}_0,
$$
where $\bar{T}_k = \sum\limits_{i=1}\limits^n\log \frac{Q(x_i^{(k)}|z_{[i]})}{Q^{\star}(x_i^{(k)}|z_{[i]})}$ and $\bar{T}_0 = \sum\limits_{i=1}\limits^n\log \frac{Q(x_i|z_{[i]})}{Q^{\star}(x_i|z_{[i]})}$. Denote the ordered statistics of $\bar{T}_1,\ldots,\bar{T}_b$ by $\bar{T}_{(1)}\leq \ldots \leq \bar{T}_{(b)}$. We observe that the marginal distributions of $\bar{T}_0,\bar{T}_1,\ldots,\bar{T}_b$ do not vary with $b$. Besides, conditional on $z$, $\bar{T}_1,\ldots,\bar{T}_b$ are i.i.d. and independent of $\bar{T}_0$. Denote the distribution of $\bar{T}_1$ conditional on $z$ by $F_1(\cdot|z)$ and the distribution of $\bar{T}_0$ conditional on $z$ by $F_0(\cdot|z)$. Denote the quantile function associated with $F_1$ and $F_0$ by $F_1^{-1}$ and $F_0^{-1}$ respectively. For any fixed positive integer $m>1$, let $\varepsilon_k=\varepsilon_{\lceil\frac{lb}{m}\rceil}$, if $\lceil\frac{(l-1)b}{m}\rceil<k\leq \lceil\frac{lb}{m}\rceil$ for some integer $1\leq l\leq m$, and
$$
\left\{
\begin{array}{ll}
\varepsilon_b=\infty, & \\
\varepsilon_{\lceil\frac{lb}{m}\rceil}=\max\{0,F_1^{-1}(\frac{l}{m})-F_0^{-1}(\mathbb{E}\theta_z)\}, & l=1,\ldots,m-1, 
\end{array}
\right.
$$
where $\theta_z$ can be any measure of the distance between $Q$ and $Q^{\star}$ such that $\theta_z=0$ when $Q$ and $Q^{\star}$ are identical. Then, we have
\begin{equation}
\begin{array}{ll}
& \mathbb{P}\left(\bigcup\limits_{k=1}\limits^{b}\{\widehat{KL}_{(k)}>\varepsilon_k\}\right)\\

\leq & \mathbb{P}\left( \bigcup\limits_{l=1}\limits^{m}\{\widehat{KL}_{\left(\lceil\frac{lb}{m}\rceil\right)}>\varepsilon_{\lceil\frac{lb}{m}\rceil}\}\right)\\
= & \mathbb{P}\left( \bigcup\limits_{l=1}\limits^{m}\{\bar{T}_{(\lceil\frac{lb}{m}\rceil)}-\varepsilon_{\lceil\frac{lb}{m}\rceil}>\bar{T}_0\}\right)\\
=&\mathbb{E}\left[\mathbb{P}\left( \bigcup\limits_{l=1}\limits^{m}\{\bar{T}_{(\lceil\frac{lb}{m}\rceil)}-\varepsilon_{\lceil\frac{lb}{m}\rceil}>\bar{T}_0\}|z\right)\right]\\
\approx &  \mathbb{E}\left[\mathbb{P}\left( \bigcup\limits_{l=1}\limits^{m}\{\bar{T}_0<F_1^{-1}\left(\frac{l}{m}\right)-\varepsilon_{\lceil\frac{lb}{m}\rceil}\}|z\right)\right]\\
\leq& \mathbb{E}\left[\mathbb{P}\left(\bar{T}_0<F_0^{-1}(\mathbb{E}\theta_z)|z\right)\right]\\
= & \mathbb{E}\theta_z,
\end{array}
 \label{mosteller1}
\end{equation}
where the 4th step is based on the asymptotic results of ordered statistics given in Mosteller (1946). For any given $z$, 
$$
\mathbb{P}\left( \bigcup\limits_{l=1}\limits^{m}\{\bar{T}_{(\lceil\frac{lb}{m}\rceil)}-\varepsilon_{\lceil\frac{lb}{m}\rceil}>\bar{T}_0\}|z\right) = \mathbb{P}\left( \bigcup\limits_{l=1}\limits^{m}\{\bar{T}_0<F_1^{-1}\left(\frac{l}{m}\right)-\varepsilon_{\lceil\frac{lb}{m}\rceil}\}|z\right) +o(1)
$$
as $b\rightarrow \infty$. Hence, under mild conditions, for any given small $\nu>0$, there exists $S_Z\in\mathbb{R}^{n\times p}$, such that $\mathbb{P}(z\in S_Z)>1-\nu$ and 
$$
\mathbb{P}\left( \bigcup\limits_{l=1}\limits^{m}\{\bar{T}_{(\lceil\frac{lb}{m}\rceil)}-\varepsilon_{\lceil\frac{lb}{m}\rceil}>\bar{T}_0\}|z\right) = \mathbb{P}\left( \bigcup\limits_{l=1}\limits^{m}\{\bar{T}_0<F_1^{-1}\left(\frac{l}{m}\right)-\varepsilon_{\lceil\frac{lb}{m}\rceil}\}|z\right) +o(1)
$$ 
holds uniformly for $z\in S_Z$. Next, we have
\begin{equation}
\begin{array}{ll}
&\frac{\lfloor (b+1)\alpha\rfloor}{1+e^{-\varepsilon_1}+\ldots+e^{-\varepsilon_b}}\\
 = & \frac{\lfloor (b+1)\alpha\rfloor}{1+\sum\limits_{l=1}\limits^{m-1}\sum\limits_{k=\lceil\frac{(l-1)b}{m}\rceil +1}\limits^{\lceil \frac{lb}{m}\rceil}exp\left(-\max\{0,F_1^{-1}(\frac{l}{m})+F_0^{-1}(\mathbb{E}\theta_z)\}\right)}\\
 = &\frac{\alpha}{\frac{1}{m}\sum\limits_{l=1}\limits^{m-1}exp\left(-\max\{0,F_1^{-1}(\frac{l}{m})-F_0^{-1}(\mathbb{E}\theta_z)\}\right)}+o(1),\\
 \end{array}
\label{mosteller2}
\end{equation}
which will be $\frac{m}{m-1}\alpha+o(1)$ when $Q$ and $Q^{\star}$ are identical as $b\rightarrow \infty$. Results in (\ref{mosteller1}) and (\ref{mosteller2}) combined together can lead to our claim at the beginning.

\subsection{ A Generalization of Theorem \ref{thm:robust2}}
Here, we give a general type 1 error upper bound control result, with the Theorem \ref{thm:robust2} being its special form.
\begin{theorem}
(Type 1 Error Upper Bound 3.) Suppose conditions in the Theorem \ref{thm:robust2} hold. Let $\{u_0,u_1,\ldots,u_b\}$ be the unordered set of $\{x,x^{(1)},\ldots,x^{(b)}\}$. Denote $D(\cdot,z)\triangleq \frac{Q^{\star}(\cdot|z)}{Q(\cdot|z)}$. For a given $s$, we denote the ordered statistics of $\{D(u_k,z):k\ne s\}$ by $D^s_{(1)}\leq D^s_{(2)}\leq \ldots D^s_{(b)}$. For any given measurable functions $$
\varepsilon_1(u_0,u_1,\ldots,u_b,z,y),\ldots,\varepsilon_b(u_0,u_1,\ldots,u_b,z,y),
$$
$$
\eta_1(u_0,u_1,\ldots,u_b,z,y),\ldots,\eta_b(u_0,u_1,\ldots,u_b,z,y),
$$
and any given $\alpha\in[0,1]$, suppose the null hypothesis $H_0:X\perp Y|Z$ holds. For conciseness, we denote
$$
I(s) \triangleq \mathbbm{1}\left\{\bigcap\limits_{k=1}^b \left\{e^{\eta_k}\leq \frac{D(u_{s-1},z)}{D^{s-1}_{(b+1-k)}}\leq e^{\varepsilon_k}\right\}\right\}.
$$
Then, 
\begin{equation}
\begin{array}{ll}
& \mathbb{P}(P_{CRRT}\leq \alpha,\eta_1\leq \widehat{KL}_{(1)}\leq \varepsilon_1,\ldots,\eta_b\leq\widehat{KL}_{(b)}\leq \varepsilon_b|u_0,\ldots,u_b,z,y)\\
\leq & \min\left\{\sum\limits_{s=1}\limits^{\lfloor (b+1)\alpha\rfloor}\frac{I(s)}{I(s)+e^{-\varepsilon_1}+\ldots+e^{-\varepsilon_b}} ,  1-\sum\limits_{s=\lfloor (b+1)\alpha\rfloor+1}\limits^{b+1}\frac{I(s)}{I(s)+e^{\eta_1}+\ldots+e^{\eta_b}} \right\}.
\end{array}
\label{complete1}
\end{equation}
This leads to that,
$$
\begin{array}{ll}
& \mathbb{P}(P_{CRRT}\leq \alpha,\eta_1\leq \widehat{KL}_{(1)}\leq \varepsilon_1,\ldots,\eta_b\leq\widehat{KL}_{(b)}\leq \varepsilon_b)\\
\leq & \min\left\{\mathbb{E}\left[\sum\limits_{s=1}\limits^{\lfloor (b+1)\alpha\rfloor}\frac{I(s)}{I(s)+e^{-\varepsilon_1}+\ldots+e^{-\varepsilon_b}}\right],\mathbb{E}\left[1-\sum\limits_{s=\lfloor (b+1)\alpha\rfloor+1}\limits^{b+1}\frac{I(s)}{I(s)+e^{\eta_1}+\ldots+e^{\eta_b}}\right] \right\}\\
\triangleq  &\min\left\{M_1,M_2 \right\}.
\end{array}
$$
Since $\varepsilon_1,\ldots,\varepsilon_b,\eta_1,\ldots,\eta_b$ are arbitrary, we can obtain the following bound directly controlling type 1 error,
$$
\begin{array}{ll}
& \mathbb{P}(P_{CRRT}\leq \alpha)\\
 \leq & \inf\limits_{\varepsilon_1,\ldots,\varepsilon_b,\eta_1,\ldots,\eta_b} \left\{\min \{M_1,M_2\}+ \mathbb{P}\left(\bigcup\limits_{k=1}\limits^{b}\left\{\{\widehat{KL}_{(k)}>\varepsilon_k\}\cup\{\widehat{KL}_{(k)}<\eta_k\}\right\}\right)\right\}.
\end{array}
$$
\label{robust3}
\end{theorem}
\begin{proof}[\bf Proof of Theorem \ref{robust3}]
We use same notations and similar ideas as in the proof of Theorem \ref{thm:robust2}. For any $s\ne r\in\{1,\ldots,1+b\}$,
\begin{equation}
\resizebox{.90\hsize}{!}{$
\begin{array}{cl}
& \mathbb{P}\left(T_0\ ranks\ s\ th,\eta_1\leq \widehat{KL}_{(1)}\leq \varepsilon_1,\ldots,\eta_b\leq \widehat{KL}_{(b)}\leq \varepsilon_b|u_0,u_1,\ldots,u_b,z,y\right)\\

= & \mathbb{E}\left(\mathbbm{1}\{T_0\ ranks\ s\ th\} \mathbbm{1}\left\{\eta_1\leq\widehat{KL}_{(1)}\leq \varepsilon_1,\ldots,\eta_b\leq\widehat{KL}_{(b)}\leq \varepsilon_b\right\}|u_0,u_1,\ldots,u_b,z,y\right)\\

= & \mathbb{E}\left(\mathbbm{1}\left\{T_0\ ranks\ s\ th\right\} \mathbbm{1}\left\{e^{\eta_1}\leq\frac{D(x^{(0)},z)}{D_{(b)}} \leq e^{\varepsilon_1},\ldots,e^{\eta_b}\leq\frac{D(x^{(0)},z)}{D_{(1)}} \leq e^{\varepsilon_b}\right\}|u_0,u_1,\ldots,u_b,z,y\right)\\

= & \mathbb{E}\left(\mathbbm{1}\left\{T_0\ ranks\ s\ th\right\} I(s)|u_0,u_1,\ldots,u_b,z,y\right)\\

= & I(s) \mathbb{P}\left(T_0\ ranks\ s\ th|u_0,u_1,\ldots,u_b,z,y\right)\\

= & I(s) \frac{D(u_{s-1},z)}{D(u_{r-1},z)}\mathbb{P}\left(T_0\ ranks\ r\ th|u_0,u_1,\ldots,u_b,z,y\right),\\
\end{array}
$}
\label{gen1}
\end{equation}
where the last step holds based on (\ref{goldenratio}). Actually, we have, 
\begin{equation}
\begin{array}{cl}
 & I(s)\frac{D(u_{s-1},z)}{D(u_{r-1},z)}\\
= & I(s)\frac{D(u_{s-1},z)}{D^{s-1}_{(b+1-\Phi_{s-1}(r-1))}}\\
\in  & \left[I(s) e^{\eta_{\Phi_{s-1}(r-1)}}, I(s) e^{\varepsilon_{\Phi_{s-1}(r-1)}}\right],\\ 
\end{array}
\label{gen2}
\end{equation}
where the definition of $\Phi$ is given in the proof of Theorem \ref{thm:robust2}. Combining results in (\ref{gen1}) and (\ref{gen2}), we have
\begin{equation}
\begin{array}{ll}
& \frac{I(s)\mathbb{P}\left(T_0\ ranks\ r\ th|u_0,u_1,\ldots,u_b,z,y\right)}{\mathbb{P}\left(T_0\ ranks\ s\ th,\eta_1\leq \widehat{KL}_{(1)}\leq \varepsilon_1,\ldots,\eta_b\leq \widehat{KL}_{(b)}\leq \varepsilon_b|u_0,u_1,\ldots,u_b,z,y\right)} \\
\in & \left[ e^{-\varepsilon_{\Phi_{s-1}(r-1)}},e^{-\eta_{\Phi_{s-1}(r-1)}}\right].\\
\end{array}
\label{gen3}
\end{equation}
For any given $s$, we have
\begin{equation}
\resizebox{.90\hsize}{!}{$
\begin{array}{rl}
& I(s)\\
= & \sum\limits_{r=1}\limits^{b+1} I(s)\mathbb{P}\left(T_0\ ranks\ r\ th|u_0,u_1,\ldots,u_b,z,y\right)\\

=&  I(s)\mathbb{P}(T_0\ ranks\ s\ th|u_0,u_1,\ldots,u_b,z,y) \\
&+ \sum\limits_{1\leq r \leq b+1,r\ne s}I(s)\mathbb{P}(T_0\ ranks\ r\ th|u_0,u_1,\ldots,u_b,z,y)\\
\ge & \mathbb{P}(T_0\ ranks\ s\ th,\eta_1\leq\widehat{KL}_{(1)}\leq \varepsilon_1,\ldots,\eta_b\leq\widehat{KL}_{(b)}\leq  \varepsilon_b|u_0,\ldots,u_b,z,y)\sum\limits_{\mathop{1\leq r \leq b+1,}\limits_{r\ne s}}e^{-\varepsilon_{\Phi_{s-1}(r-1)}}\\
& + I(s)\mathbb{P}(T_0\ ranks\ s\ th|u_0,u_1,\ldots,u_b,z,y)\\
\ge & \mathbb{P}(T_0\ ranks\ s\ th,\eta_1\leq\widehat{KL}_{(1)}\leq \varepsilon_1,\ldots,\eta_b\leq\widehat{KL}_{(b)}\leq  \varepsilon_b|u_0,\ldots,u_b,z,y)(I(s)+\sum\limits_{r=1}\limits^{b}e^{-\varepsilon_r}),\\
\end{array}
$}
\label{gen4}
\end{equation}
where the third step is based on (\ref{gen3}) and the last step is legitimate due to $\Phi_s$'s one-to-one property, $s=0,1,\ldots,b$. Therefore, we have
\begin{equation}
\begin{array}{cl}
& \mathbb{P}\left(P_{CRRT}\leq \alpha,\eta_1\leq \widehat{KL}_{(1)}\leq \varepsilon_1,\ldots,\eta_b\leq \widehat{KL}_{(b)}\leq  \varepsilon_b|u_0,\ldots,u_b,z,y\right)\\
= & \sum\limits_{s=1}\limits^{\lfloor (b+1)\alpha\rfloor}\mathbb{P}\left(T_0\ ranks\ s\ th,\eta_1\leq \widehat{KL}_{(1)}\leq \varepsilon_1,\ldots,\eta_b\leq\widehat{KL}_{(b)}\leq  \varepsilon_b|u_0,\ldots,u_b,z,y\right)\\
\leq & \sum\limits_{s=1}\limits^{\lfloor (b+1)\alpha\rfloor}\frac{I(s)}{I(s)+e^{-\varepsilon_1}+\ldots+e^{-\varepsilon_b}}.
\end{array}
\label{gen5}
\end{equation}

Based on (\ref{gen3}), similar to (\ref{gen4}), we have
\begin{equation}
\resizebox{.90\hsize}{!}{$
\mathbb{P}(T_0\ ranks\ s\ th,\eta_1\leq\widehat{KL}_{(1)}\leq \varepsilon_1,\ldots,\eta_b\leq\widehat{KL}_{(b)}\leq  \varepsilon_b|u_0,\ldots,u_b,z,y)\ge \frac{I(s)}{(I(s)+\sum\limits_{r=1}\limits^{b}e^{-\eta_r})}.
\label{gen6}
$}
\end{equation}

Hence, we have
\begin{equation}
\resizebox{.90\hsize}{!}{$
\begin{array}{cl}
& \mathbb{P}\left(P_{CRRT}\leq \alpha,\eta_1\leq \widehat{KL}_{(1)}\leq \varepsilon_1,\ldots,\eta_b\leq \widehat{KL}_{(b)}\leq  \varepsilon_b|u_0,\ldots,u_b,z,y\right)\\
= & 1-\sum\limits_{s=\lfloor (b+1)\alpha\rfloor1}\limits^{b+1}\mathbb{P}\left(T_0\ ranks\ s\ th,\eta_1\leq \widehat{KL}_{(1)}\leq \varepsilon_1,\ldots,\eta_b\leq\widehat{KL}_{(b)}\leq  \varepsilon_b|u_0,\ldots,u_b,z,y\right)\\
\leq & 1-\sum\limits_{s=\lfloor (b+1)\alpha\rfloor1}\limits^{b+1}\mathbb{P}\frac{I(s)}{I(s)+e^{-\eta_1}+\ldots+e^{-\eta_b}}.
\end{array}
\label{gen7}
$}
\end{equation}
Combining results in (\ref{gen5}) and (\ref{gen7}), we can conclude that (\ref{complete1}) is valid.

\end{proof}

\section{Connections Between CRRT and Knockoffs}
\label{sec:knockoff}
The CRRT is deeply influenced by the spirit of knockoffs (\cite{barber2015controlling}; \cite{candes2018panning}; \cite{barber2019knockoff}) that utilizing information on the covariates sides can help us avoid the uncertainty of the relationship between the response variable and covariates. In this section, we make this connection more concrete. Specifically speaking, we show that there's a direct comparison between the CRRT and multiple knockoffs (\cite{gimenez2018improving}), which is a variant of the original knockoffs. After comparison, we point out that the version of multiple knockoffs given in \cite{gimenez2018improving} may be too conservative. We provide a modified version that can own greater power and at the same time, maintain finite sample FDR control at a same level. 

Suppose that we have original data $(y,x)$, where $y$ is a $n\times 1$ vector, $x$ is a $n\times p$ matrix, which sometimes is denoted by $x^{(0)}$ alternatively. Each row of $(y,x)$, a realization of $(Y,X^T)^T$, is independently and identically generated from some distribution $F_{Y,X}$. Just be careful that the $x$ in this section is different from those in previous sections. For a given positive integer $b$, we generate a $b$-knockoff copy of $x^{(0)}$, denoted by $(x^{(1)},x^{(2)},\ldots,x^{(b)})$, where $x^{(k)}$ is an $n\times p$ matrix, for $k=1,\ldots,b$. We call $(x^{(1)},x^{(2)},\ldots,x^{(b)})$ a $b$-knockoff copy of $x^{(0)}$ if $(x^{(0)},x^{(1)},\ldots,x^{(b)})$ satisfies the extended exchangeability defined below and $(x^{(1)},x^{(2)},\ldots,x^{(b)})\perp y|x^{(0)}$. 

\begin{definition}
If $\sigma = (\sigma^{1},\ldots,\sigma^{p})\in \Sigma_{b+1}\times\Sigma_{b+1}\times\ldots\times\Sigma_{b+1}$, where $\Sigma_{b+1}$ is the collection of all permutations of $\{1,2,\ldots,b+1\}$, we define $(x^{(0)},x^{(1)},\ldots,x^{(b)})_{swap(\sigma)}\triangleq (u^{(0)},\ldots,u^{(b)})$, where $u^{(k)}_j=x^{(\sigma^j_{k+1})}_j$, for $k=0,1,\ldots,b$ and $j=1,\ldots,p$. We say that $(x^{(0)},x^{(1)},\ldots,x^{(b)}) \in \mathbb{R}^n\times \mathbb{R}^n\times\ldots \times\mathbb{R}^n $ is extended exchangeable, if for any $\sigma^{1},\ldots,\sigma^{p}\in \Sigma_{b+1}$,
$$
(x^{(0)},x^{(1)},\ldots,x^{(b)})_{swap(\sigma)} \mathop{=}\limits^{d} (x^{(0)},x^{(1)},\ldots,x^{(b)}),
$$
where $\sigma = (\sigma^{1},\ldots,\sigma^{p})$.
\label{extexch}
\end{definition}
To construct multiple knockoffs selection, we can use any test function $\bar{T}=(T^0,T^1,\ldots,T^{b})$ satisfying that for any $\sigma = (\sigma^{1},\ldots,\sigma^{p})\in \Sigma_{b+1}\times\Sigma_{b+1}\times\ldots\times\Sigma_{b+1}$,
\begin{equation}
\bar{T}_{swap(\sigma)}=(T^0,\ldots,T^b)_{swap(\sigma)}=\bar{T}(y,(x^{(0)},\ldots,x^{(b)})_{swap(\sigma)})),
\label{multiknockofft}
\end{equation}
where $T^k:\mathbb{R}^{n}\times \mathbb{R}^{n\times[p\times(b+1)]}\rightarrow \mathbb{R}^{p}$ for $k=0,1,\ldots,b$. For $j=1,2,\ldots,p$, denote the ordered statistics of $\{T^0_j,\ldots,T^b_j\}$ by $T^{(1)}_j\leq \ldots,\leq T^{(b+1)}_j$ and define $\tau_j\triangleq T^{(1)}_j-T^{(2)}_j$. As like \cite{gimenez2018improving}, for simplicity, we assume that
$$
\mathbb{P}\left(\bigcap\limits_{j=1}\limits^{p}\{T^{(1)}_j < \ldots,< T^{(b+1)}_j,\ or,\ T^{(1)}_j = \ldots,= T^{(b+1)}_j\}\right)=1.
$$
Then, for $j=1,2,\ldots,p$, we can define 
$$
r_j \triangleq \left\{
\begin{array}{ll}
the\ rank\ of\ T_j^{0}\ among\ \{T_j^{0},\ldots,T_j^{b}\},& if\ T^{(1)}_j < \ldots,< T^{(b+1)}_j,\\
b+1, & if\ T^{(1)}_j = \ldots,= T^{(b+1)}_j.
\end{array}
\right.
$$
In \cite{gimenez2018improving}, for any given desired FDR level $\alpha\in(0,1)$, the selection set is defined as
$$
\hat{S}=\{j\in\{1,\ldots,p\},r_j=1,\tau_j\ge\hat{\tau}\},
$$
where $\hat{\tau}=\min\left\{t\ge \min\{\tau_j:1\leq j\leq p\}:\frac{\frac{1}{b}+\frac{1}{b}\#\{1\leq j\leq p,r_j>1,\tau_j\ge t\}}{\#\{1\leq j\leq p,r_j=1,\tau_j\ge t\}\vee 1}\leq \alpha\right\}$. They prove that this multiple knockoffs selection procedure can control the FDR at level $\alpha$.

Now, let's take a look at a special case of their multiple-knockoff procedure. Assume that $X_1$, the first element of $X$, is our main target and we want to check whether $H_0: X_1\perp Y|X_{-1}$ holds, where $X_{-1} \triangleq (X_2,\ldots,X_p)^T$. For $k=1,2,\ldots,b$, let $x^{(k)}=(x^{(k)}_1,x_2,\ldots,x_p)$, where $x_1^{(1)},\ldots,x_1^{(b)}$ are sampled i.i.d. from distribution of $X_1|X_{-1}$, independent of $x_1$. It's easy to check that under such construction, $(x^{(0)},\ldots,x^{(b)})$ is extended exchangeable. For any qualified multiple knockoffs test function $\bar{T}=(T^0,T^1,\ldots,T^b)$ satisfying the property stated in (\ref{multiknockofft}), we define a function $\bar{V}=(V^0,V^1,\ldots,V^b)\triangleq (T^0_1,T^1_1,\ldots,T^b_1)$, which is essentially a function of $(y,x_2,\ldots,x_p,x^{(0)}_1,x^{(1)}_1,\ldots,x^{(b)}_1)$. In fact, $\bar{V}$ is $X$-symmetric if we view it as a test function for testing conditional independence. For any given $\alpha\in(0,1)$, the definition of $\hat{\tau}$ can be simplified as
$$
\hat{\tau}=\min \left\{t\ge \tau_1:\frac{1}{b}+\frac{1}{b}\mathbbm{1}\{r_1>1,\tau_1\ge t\}\leq \alpha\right\}.
$$
If $\frac{2}{b}\leq \alpha$, we always have $\hat{\tau}=\tau_1$, and hence we'll select the first covariate if and only if $r_1=1$. If $\frac{1}{b}>\alpha$, we have $\hat{\tau}=\infty$ and hence we select no variable. If $\frac{1}{b} \leq \alpha < \frac{2}{b}$, we select the first variable if and only if $r_1=1$.

As comparison, the CRRT based on $\bar{V}$ rejects the null hypothesis and select the first variable if and only if $T_1^0$ ranks at least $\lambda\triangleq \lfloor (b+1)\alpha\rfloor$th among $T^0_1,T^1_1,\ldots,T^b_1$. Thus, when $b<\frac{2}{\alpha}-1$, decision rules given by the multiple knockoffs and the CRRT are basically identical. To be more specific, when $b<\frac{1}{\alpha}-1$, the CRRT select no vairable. When $\frac{1}{\alpha}-1\leq b<\frac{2}{\alpha}-1$, $X_1$ is selected. However, there's a great difference when $b\ge \frac{2}{\alpha}-1$ that the special multiple knockoffs procedure rejects $H_0$ only if $T_1^0$ is the largest one while the CRRT rejects $H_0$ when $T_1^0$ ranks high enough. When $b$ gets larger, the above special multiple knockoffs procedure exhibits greater conservativeness, which is unnecessary. Though such comparison is unfair because the multiple knockoffs doesn't specialize in conditional independence testing, it still gives us some hint to improve the multiple knockoffs. Trying to eliminate such redundant conservativeness, we provide a modified definition of the threshold:
\begin{equation}
\resizebox{.90\hsize}{!}{$
\tilde{\tau} = \min\left\{t\ge \min\{\tau_j:1\leq j\leq p\}:\frac{\eta(\tilde{\lambda})\left(1+\#\{1\leq j\leq p,r_j>\tilde{\lambda},\tau_j\ge t\}\right)}{\#\{1\leq j\leq p,r_j\leq \tilde{\lambda},\tau_j\ge t\}\vee 1}\leq \alpha\right\},
\label{tautilde}
$} 
\end{equation} 
where $\tilde{\lambda}=\lfloor (b+1)\frac{\alpha}{\alpha+1}\rfloor$, $\eta(\tilde{\lambda})=\frac{\tilde{\lambda}}{b+1-\tilde{\lambda}}$, and further we let the selected set be $\tilde{S}=\{j\in\{1,\ldots,p\},r_j\leq \tilde{\lambda},\tau_j\ge \tilde{\tau}\}$. This modified multiple knockoffs procedure roughly matches the CRRT in single variable testing problem. Since $\eta(\tilde{\lambda})\leq \alpha$, the modified selection procedure selects $X_1$ if and only if $r_1\leq \tilde{\lambda} = \lfloor (b+1)\frac{\alpha}{\alpha+1}\rfloor$, which is close to $\lambda$ when $\alpha$ is small. We conjecture that this modified version would be more powerful in general model selection problem. But we don't pursue to show it here because we want to focus on the conditional testing problem. In fact, the modified multiple knockoffs procedure can also control finite sample FDR like the original procedure.
\begin{proposition}
For any given desired FDR level $\alpha\in(0,1)$. The selected set $\tilde{S}$ can control FDR at level $\alpha$. To be more concrete,
$$
FDR=\mathbb{E}\left[\frac{|\tilde{S}\setminus S |}{|\tilde{S}|\vee 1} \right]\leq \alpha,
$$
where $S$ is the Markov blanket for $Y$ (\cite{pearl1988probabilistic}; \cite{candes2018panning}), which can be interpreted as the set of nonnull variables.
\label{multimodified}
\end{proposition}
From the proof of the above proposition, we can see that actually in (\ref{tautilde}), we can replace $\tau_j\ge t$ with some other conditions based on ordered statistics $\check{T} = \{\{T_j^{(1)},T_j^{(2)},\ldots,T_j^{(b+1)}\}:j=1,2,\ldots,p\}$ while remaining the FDR control, which means that we can generalize the multiple knockoffs framework to make it more flexible. We didn't study this generalization in this paper. There could be some alternative conditions enjoying optimal properties in some sense.

\section{Proofs}

\begin{proof}[\bf Proof of Proposition \ref{xsym}]
We first to show that, for any $\sigma\in \Sigma_{b+1}$
$$
(T_0(y,z,\tilde{x}),\ldots,T_b(y,z,\tilde{x}))_{permute(\sigma)}\stackrel{d}{=}(T_0(y,z,\tilde{x}),\ldots,T_b(y,z,\tilde{x})).
$$
Under the null hypothesis $H_0:X\perp Y|Z$, $x,x^{(1)},\ldots,x^{(b)}|z,y=x,x^{(1)},\ldots,x^{(b)}|z$ are i.i.d.. Therefore, we have $\tilde{x}_{permute(\sigma)}|z,y\stackrel{d}{=}\tilde{x}|z,y$. Further, we have $\tilde{x}_{permute(\sigma)}\stackrel{d}{=}\tilde{x}$. Thus,
$$
\begin{array}{cl}
& (T_0(y,z,\tilde{x}),\ldots,T_b(y,z,\tilde{x}))_{permute(\sigma)}\\
= & (T_0(y,z,\tilde{x}_{permute(\sigma)}),\ldots,T_b(y,z,\tilde{x}_{permute(\sigma)}))\\
\stackrel{d}{=} & (T_0(y,z,\tilde{x}),\ldots,T_b(y,z,\tilde{x})),
\end{array}
$$
which is our desired result. This result directly implies that
$$
T_0(y,z,\tilde{x})\stackrel{d}{=}T_1(y,z,\tilde{x})\stackrel{d}{=}\ldots\stackrel{d}{=}T_b(y,z,\tilde{x}).
$$
\end{proof}


\begin{proof}[\bf Proof of Proposition \ref{conditionprop}]
According to Proposition \ref{xsym}, for any $\sigma\in \Sigma_{b+1}$, 
$$
(T_0,\ldots,T_b)_{permute(\sigma)}\stackrel{d}{=}(T_0,\ldots,T_b).
$$ 
Since the ordered statistics of $(T_0,\ldots,T_b)_{permute(\sigma)}$ is also $T_{(0)}\leq T_{(1)}\leq \ldots,\leq T_{(b)}$, we have
$$
(T_0,\ldots,T_b)_{permute(\sigma)}|z,y,T_{(0)},\ldots,T_{(b)}\stackrel{d}{=}(T_0,\ldots,T_b)|z,y,T_{(0)},\ldots,T_{(b)}.
$$
It implies that 
$$
\mathbb{P}(T_{\sigma_1}=T_{(0)},\ldots,T_{\sigma_{b+1}}=T_{(b)})=\mathbb{P}(T_0=T_{(0)},\ldots,T_b=T_{(b)}),
$$
where $\sigma_i$ is the $i$-th element of $\sigma$, $i=1,2,\ldots,b+1$. Because $\sigma$ is arbitrary, we know that conditional on $z,y,T_{(0)},\ldots,T_{(b)}$, $(T_0,\ldots,T_b)$ is uniformly distributed on 
$$
S=\{(T_{(0)},\ldots,T_{(b)})_{permute(\sigma)}:\sigma\in \Sigma_{b+1}\}.
$$
Based on such symmetry, we can easily know that
$$
\mathbb{P}(p_{CRRT}\leq \alpha|z,y,T_{(0)},T_{(1)},\ldots,T_{(b)})\leq \alpha,
$$
for any given $\alpha$.
\end{proof}


\begin{proof}[\bf Proof of Theorem \ref{robust1}]
Let $\dot{x}$ be a n-dimension vector with elements $\dot{x}_i\sim Q(\cdot|z{[i]})$ independently, $i=1,2,\ldots,n$, independent of $y$ and of $\tilde{x}=(x,x^{(1)},\ldots,x^{(b)})$.

Recall that $p_{CRRT} = \frac{\sum\limits_{k=0}\limits^{b}\mathbbm{1}\{T_0(y,z,\tilde{x})\leq T_k(y,z,\tilde{x})\}}{1+b}$, which is a function of $(y,z,\tilde{x})$. Define $S_{\alpha} = \{(y,z,\tilde{x}):p_{CRRT}(y,z,\tilde{x})\leq \alpha\}$. For given $z$ and $y$, define 
$$\bar{S}_{\alpha}(y,z) = \{(x,x^{(1)},\ldots,x^{(b)}):(y,z,x,x^{(1)},\ldots,x^{(b)})\in S_{\alpha}\}.
$$
Then,
$$
\begin{array}{cl}
&\mathbb{P}(p_{CRRT}\leq \alpha|y,z)\\
= & \mathbb{P}((y,z,x,x^{(1)},\ldots,x^{(b)})\in S_{\alpha}|y,z)\\
= & \mathbb{P}((x,x^{(1)},\ldots,x^{(b)})\in \bar{S}_{\alpha}(y,z)|y,z)\\
\leq & \mathbb{P}((\dot{x},x^{(1)},\ldots,x^{(b)})\in \bar{S}_{\alpha}(y,z)|y,z) \\
&+ |\mathbb{P}((\dot{x},x^{(1)},\ldots,x^{(b)})\in \bar{S}_{\alpha}(y,z)|y,z)-\mathbb{P}((x,x^{(1)},\ldots,x^{(b)})\in \bar{S}_{\alpha}(y,z)|y,z)|
\end{array}
$$

Since $\dot{x}$ is sampled based on $Q(\cdot|Z)$, $\dot{x},x^{(1)},\ldots,x^{(b)}$ are i.i.d. conditional on $z$ and $y$. According to Theorem \ref{maintheorem}, $\mathbb{P}((\dot{x},x^{(1)},\ldots,x^{(b)})\in \bar{S}_{\alpha}(y,z)|y,z)=\mathbb{P}(p_{CRRT}(y,z,\dot{x},x^{(1)},\ldots,x^{(b)})\leq \alpha|y,z)\leq \alpha$.

As for the second part, for given $y,z,x^{(1)},\ldots,x^{(b)}$, we define $C_{\alpha}(y,z,x^{(1)},\ldots,x^{(b))}=\{\dot{x}:(y,z,\dot{x},x^{(1)},\ldots,x^{(b)})\in S_{\alpha}$. Then,
$$
\begin{array}{cl}
& |\mathbb{P}((\dot{x},x^{(1)},\ldots,x^{(b)})\in \bar{S}_{\alpha}(y,z)|y,z)-\mathbb{P}((x,x^{(1)},\ldots,x^{(b)})\in \bar{S}_{\alpha}(y,z)|y,z)|\\

= & |\int\limits_{(x^{(1)},\ldots,x^{(b)})\in \mathbb{R}^b} \int\limits_{\dot{x}\in C_{\alpha}(y,z,x^{(1)},\ldots,x^{(b)})} 1 dQ(\dot{x}|y,z)dQ(x^{(1)},\ldots,x^{(b)}|y,z)\}\\

& - \int\limits_{(x^{(1)},\ldots,x^{(b)})\in \mathbb{R}^b} \int\limits_{{x}\in C_{\alpha}(y,z,x^{(1)},\ldots,x^{(b)})} 1 dQ^{\star}({x}|y,z)dQ(x^{(1)},\ldots,x^{(b)}|y,z)\}|\\

\leq & \int\limits_{(x^{(1)},\ldots,x^{(b)})\in \mathbb{R}^b} | \int\limits_{\dot{x}\in C_{\alpha}(y,z,x^{(1)},\ldots,x^{(b)})} 1 dQ(\dot{x}|y,z)\\

& - \int\limits_{{x}\in C_{\alpha}(y,z,x^{(1)},\ldots,x^{(b)})} 1 dQ^{\star}({x}|y,z) | dQ(x^{(1)},\ldots,x^{(b)}|y,z)\\

\leq & \int\limits_{(x^{(1)},\ldots,x^{(b)})\in \mathbb{R}^b} d_{TV}(Q^{\star}_n(\cdot|z),Q_n(\cdot|z)) dQ(x^{(1)},\ldots,x^{(b)}|y,z)\\

= &d_{TV}(Q^{\star}_n(\cdot|z),Q_n(\cdot|z)).
\end{array}
$$

Combining the results above, we have 
$$
\mathbb{P}(p_{CRRT}\leq \alpha|y,z) \leq \alpha + d_{TV}(Q^{\star}_n(\cdot|z),Q_n(\cdot|z)).
$$

Taking expectation with respect to $y$ and $z$ on both sides completes the proof.
\end{proof}


\begin{proof}[\bf Proof of Theorem \ref{thm:robust2}]
In the following proof, sometimes we use $x^{(0)}$ to denote $x$ alternatively. When we say $T_0$ ranks $r$th, we mean that $T_0$ is the $r$th largest value among $\{T_0,\ldots,T_{b}\}$. Recall that 
$$
p_{CRRT} = \frac{\sum\limits_{k=0}\limits^{b}\mathbbm{1}\{T_0(y,z,\tilde{x})\leq T_k(y,z,\tilde{x})\}}{1+b}.
$$
Based on the above definition, we know that $p_{CRRT}\leq \alpha$ is roughly equivalent to that $T_0$ is among the $\lfloor (b+1)\alpha\rfloor$ largest of $T_0,\ldots,T_b$. When all $T_k$'s are distinct, such equivalence is exact. When there are some $T_k$'s, $k\ge 1$ tied with $T_0$, our rule in fact assigns $T_0$ a lower rank (smaller the value, lower the rank) and hence tends to be conservative. And we can see that whether there are repeated values among $T_k,k\ge 1$ doesn't affect $T_0$'s rank. Thus, without loss of generality, we may assume that all $T_k$'s are distinct with probability 1, conditional on $z$ and $y$. That is,
$$
\mathbb{P}(\exists i\ne j,T_i=T_j|z,y)=0.
$$
Denote the unordered set of $\{x,x^{(1)},\ldots,x^{(b)}\}$ by $\{u_0,\ldots,u_b\}$. According to our assumption, with probability 1, conditional on $z$ and $y$, $u_0,\ldots,u_b$ are distinct with probability 1 (no typo here). We have the following claim. \vspace{0.2cm}

\noindent {\bf Claim 1} With probability 1, there exists a one-to-one map $\Psi:\{1,\ldots,b+1\}\longrightarrow \{u_0,\ldots,u_b\}$, such that $T_0$ ranks $r$th if and only if $x=\Psi(r)$, $r = 1,\ldots,b+1$. Such map depends on the unordered set $\{u_0,\ldots,u_b\}$ and the test function $T$, $y$ and $z$. The proof of this claim is given immediately after the current proof.
\vspace{0.2cm}

Given $y,z,u_0,\ldots,u_b$, according to Claim 1, without loss of generality, we assume that $T_0$ ranks $r$th if and only if $x=u_{r-1}$, for $r=1,2,\ldots,b+1$. We are going to show that for any $s,r\in\{1,\ldots,b+1\}$ and $s\ne r$,
\begin{equation}
\frac{\mathbb{P}(T_0\ ranks\ s\ th|u_0,u_1,\ldots,u_b,z,y)}{\mathbb{P}(T_0\ ranks\ r\ th|u_0,u_1,\ldots,u_b,z,y)}=\frac{Q^{\star}(u_{s-1}|z)Q(u_{r-1}|z)}{Q^{\star}(u_{r-1}|z)Q(u_{s-1}|z)}.
\label{goldenratio}
\end{equation}
Actually, based on our assumption, we have
$$
\begin{array}{cl}
& \frac{\mathbb{P}(T_0\ ranks\ s\ th|u_0,u_1,\ldots,u_b,z,y)}{\mathbb{P}(T_0\ ranks\ r\ th|u_0,u_1,\ldots,u_b,z,y)}\\
= & \frac{\mathbb{P}(x=u_{s-1}|u_0,u_1,\ldots,u_b,z,y)}{\mathbb{P}(x=u_{r-1}|u_0,u_1,\ldots,u_b,z,y)}\\
= & \frac{\sum\limits_{k=1}\limits^{b}\mathbb{P}(x=u_{s-1},x^{(k)}=u_{r-1}|u_0,u_1,\ldots,u_b,z,y)}{\sum\limits_{k=1}\limits^{b}\mathbb{P}(x=u_{r-1},x^{(k)}=u_{s-1}|u_0,u_1,\ldots,u_b,z,y)}.
\end{array}
$$
To prove (\ref{goldenratio}), it suffices to show that $\frac{\mathbb{P}(x=u_{s-1},x^{(k)}=u_{r-1}|u_0,u_1,\ldots,u_b,z,y)}{\mathbb{P}(x=u_{r-1},x^{(k)}=u_{s-1}|u_0,u_1,\ldots,u_b,z,y)}=\frac{Q^{\star}(u_{s-1}|z)Q(u_{r-1}|z)}{Q^{\star}(u_{r-1}|z)Q(u_{s-1}|z)}$ for any $k\in\{1,\ldots,b\}$. We have the following observation,
\begin{equation}
\begin{array}{cl}
& \frac{\mathbb{P}(x=u_{s-1},x^{(k)}=u_{r-1}|u_0,u_1,\ldots,u_b,z,y)}{\mathbb{P}(x=u_{r-1},x^{(k)}=u_{s-1}|u_0,u_1,\ldots,u_b,z,y)}\\
=&\frac{\sum\limits_{\sigma\in\Sigma_{b+1},\sigma_1=s,\sigma_{k+1}=r}\mathbb{P}(x^{(t-1)}=u_{\sigma_t-1},t=1,\ldots,b+1|u_0,u_1,\ldots,u_b,z,y)}{\sum\limits_{\sigma'\in\Sigma_{b+1},\sigma'_1=r,\sigma'_{k+1}=s}\mathbb{P}(x^{(t-1)}=u_{\sigma'_t-1},t=1,\ldots,b+1|u_0,u_1,\ldots,u_b,z,y)}.
\end{array}
\label{goldenratio2}
\end{equation}
It's obvious that there's a one-to-one map $\Gamma$ from $\{\sigma\in\Sigma_{b+1}:\sigma_1=s,\sigma_{k+1}=r\}$ to $\{\sigma'\in\Sigma_{b+1}:\sigma'_1=r,\sigma'_{k+1}=s\}$ such that if $\sigma'=\Gamma(\sigma)$, for some $\sigma\in\{\sigma\in\Sigma_{b+1}:\sigma_1=s,\sigma_{k+1}=r\}$, then $\sigma'_t=\sigma_t$, for $t\ne 1\ or\ k+1$. If $\sigma'=\Gamma(\sigma)$,
$$
\begin{array}{cl}
& \frac{\mathbb{P}(x^{(t-1)}=u_{\sigma_t-1},t=1,\ldots,b+1|u_0,u_1,\ldots,u_b,z,y)}{\mathbb{P}(x^{(t-1)}=u_{\sigma'_t-1},t=1,\ldots,b+1|u_0,u_1,\ldots,u_b,z,y)}\\
= & \left(\frac{\mathbb{P}(x^{(t-1)}=u_{\sigma_t-1},t=1,\ldots,b+1|z,y)}{\mathbb{P}(\tilde{x}\in U|z,y)} \right)/ \left(\frac{\mathbb{P}(x^{(t-1)}=u_{\sigma'_t-1},t=1,\ldots,b+1|z,y)}{\mathbb{P}(\tilde{x}\in U|z,y)} \right)\\
= & \frac{\mathbb{P}(x^{(t-1)}=u_{\sigma_t-1},t=1,2,\ldots,b+1|z,y)}{\mathbb{P}(x^{(t-1)}=u_{\sigma'_t-1},t=1,2,\ldots,b+1|z,y)} = \frac{Q^{\star}(u_{s-1}|z)\prod\limits_{t=2}\limits^{b+1}Q(u_{\sigma_t-1}|z)}{Q^{\star}(u_{r-1}|z)\prod\limits_{t=2}\limits^{b+1}Q(u_{\sigma'_t-1}|z)}\\
= & \frac{Q^{\star}(u_{s-1}|z)Q(u_{r-1}|z)}{Q^{\star}(u_{r-1}|z)Q(u_{s-1}|z)},
\end{array}
$$
where $U=\{(u_0,\ldots,u_b)_{permute(\hat{\sigma})}:\hat{\sigma}\in\Sigma_{b+1}\}$. Given the above results and that $\Gamma$ is one-to-one, we have
$$
\resizebox{.92\hsize}{!}{$
\frac{\sum\limits_{\sigma\in\Sigma_{b+1},\sigma_1=s,\sigma_{k+1}=r}\mathbb{P}(x^{(t-1)}=u_{\sigma_t-1},t=1,\ldots,b+1|u_0,u_1,\ldots,u_b,z,y)}{\sum\limits_{\sigma'\in\Sigma_{b+1},\sigma'_1=r,\sigma'_{k+1}=s}\mathbb{P}(x^{(t-1)}=u_{\sigma'_t-1},t=1,\ldots,b+1|u_0,u_1,\ldots,u_b,z,y)} =  \frac{Q^{\star}(u_{s-1}|z)Q(u_{r-1}|z)}{Q^{\star}(u_{r-1}|z)Q(u_{s-1}|z)},
$}
$$
which, combined with (\ref{goldenratio2}), results in
$$
\frac{\mathbb{P}(x=u_{s-1},x^{(k)}=u_{r-1}|u_0,u_1,\ldots,u_b,z,y)}{\mathbb{P}(x=u_{r-1},x^{(k)}=u_{s-1}|u_0,u_1,\ldots,u_b,z,y)} = \frac{Q^{\star}(u_{s-1}|z)Q(u_{r-1}|z)}{Q^{\star}(u_{r-1}|z)Q(u_{s-1}|z)}. 
$$
And this implies (\ref{goldenratio}) is validate. Denote $D(x^{(k)},z) = \frac{Q^{\star}(x^{(k)}|z)}{Q(x^{(k)}|z)}$, for $k=0,1,\ldots,b$. Denote the ordered statistics of $\{D(x^{(k)},z):k=1,\ldots,b\}$ by $D_{(1)}\leq D_{(2)}\leq \ldots \leq D_{(b)}$. For a given $s$, we denote the ordered statistics of $\{D(u_k,z):k\ne s\}$ by $D^s_{(1)}\leq D^s_{(2)}\leq \ldots \leq D^s_{(b)}$. Next, for any $s\ne r\in\{1,\ldots,1+b\}$,
\begin{equation}
\resizebox{.90\hsize}{!}{$
\begin{array}{cl}
& \mathbb{P}(T_0\ ranks\ s\ th,\widehat{KL}_{(1)}\leq \varepsilon_1,\ldots,\widehat{KL}_{(b)}\leq \varepsilon_b|u_0,u_1,\ldots,u_b,z,y)\\
= & \mathbb{E}(\mathbbm{1}\{T_0\ ranks\ s\ th\} \mathbbm{1}\{\widehat{KL}_{(1)}\leq \varepsilon_1,\ldots,\widehat{KL}_{(b)}\leq \varepsilon_b\}|u_0,u_1,\ldots,u_b,z,y)\\
= & \mathbb{E}(\mathbbm{1}\{T_0\ ranks\ s\ th\} \mathbbm{1}\{\frac{D(x^{(0)},z)}{D_{(b)}} \leq e^{\varepsilon_1},\frac{D(x^{(0)},z)}{D_{(b-1)}} \leq e^{\varepsilon_2},\ldots,\frac{D(x^{(0)},z)}{D_{(1)}} \leq e^{\varepsilon_b}\}|u_0,u_1,\ldots,u_b,z,y)\\
= & \mathbb{E}(\mathbbm{1}\{T_0\ ranks\ s\ th\} \mathbbm{1}\{\frac{D(u_{s-1},z)}{D^{s-1}_{(b)}} \leq e^{\varepsilon_1},\frac{D(u_{s-1},z)}{D^{s-1}_{(b-1)}} \leq e^{\varepsilon_2},\ldots,\frac{D(u_{s-1},z)}{D^{s-1}_{(1)}} \leq e^{\varepsilon_b}\}|u_0,u_1,\ldots,u_b,z,y)\\
= & \mathbbm{1}\{\frac{D(u_{s-1},z)}{D^{s-1}_{(b)}} \leq e^{\varepsilon_1},\ldots,\frac{D(u_{s-1},z)}{D^{s-1}_{(1)}} \leq e^{\varepsilon_b}\} \mathbb{P}(T_0\ ranks\ s\ th|u_0,u_1,\ldots,u_b,z,y)\\
= & \mathbbm{1}\{\frac{D(u_{s-1},z)}{D^{s-1}_{(b)}} \leq e^{\varepsilon_1},\ldots,\frac{D(u_{s-1},z)}{D^{s-1}_{(1)}} \leq e^{\varepsilon_b}\} \frac{D(u_{s-1},z)}{D(u_{r-1},z)} \mathbb{P}(T_0\ ranks\ r\ th|u_0,u_1,\ldots,u_b,z,y),\\
\end{array}
$}
\label{mid1}
\end{equation}
where the 3rd step is valid because we previously assume that $T_0$ ranks $s$th if and only if $x^{(0)}=u_{s-1}$, the second to the last step is valid due to the fact that $\mathbbm{1}\{\frac{D(u_{s-1},z)}{D^{s-1}_{(b)}} \leq e^{\varepsilon_1},\ldots,\frac{D(u_{s-1},z)}{D^{s-1}_{(1)}} \leq e^{\varepsilon_b}\}$ is a function of $u_0,\ldots,u_b,z$, the last step holds based on (\ref{goldenratio}). For any given s, it's obvious that there's a one-to-one map $\Phi_s$ from $\{0,1,\ldots,b\}\setminus\{s\}$ to $\{1,\ldots,b\}$ such that $D(u_{r},z)=D^{s}_{(b+1-\Phi_s(r))}$, $r\in\{0,1,\ldots,b\}\setminus\{s\}$. Then, we can resume our previous derivation,
$$
\begin{array}{cl}
 & \mathbbm{1}\{\frac{D(u_{s-1},z)}{D^{s-1}_{(b)}} \leq e^{\varepsilon_1},\ldots,\frac{D(u_{s-1},z)}{D^{s-1}_{(1)}} \leq e^{\varepsilon_b}\} \frac{D(u_{s-1},z)}{D(u_{r-1},z)}\\
= & \mathbbm{1}\{\frac{D(u_{s-1},z)}{D^{s-1}_{(b)}} \leq e^{\varepsilon_1},\ldots,\frac{D(u_{s-1},z)}{D^{s-1}_{(1)}} \leq e^{\varepsilon_b}\}\frac{D(u_{s-1},z)}{D^{s-1}_{(b+1-\Phi_{s-1}(r-1))}}\\
 \leq  & e^{\varepsilon_{\Phi_{s-1}(r-1)}}.\\ 
\end{array}
$$
Plugging it back to (\ref{mid1}), we have
$$
\begin{array}{cl}
& \mathbb{P}(T_0\ ranks\ s\ th,\widehat{KL}_{(1)}\leq \varepsilon_1,\ldots,\widehat{KL}_{(b)}\leq  \varepsilon_b|u_0,u_1,\ldots,u_b,z,y) \\
\leq & e^{\varepsilon_{\Phi_{s-1}(r-1)}} \mathbb{P}(T_0\ ranks\ r\ th|u_0,u_1,\ldots,u_b,z,y).
\end{array}
$$
With the above results, for any given $s$, we have
$$
\resizebox{.92\hsize}{!}{$
\begin{array}{rl}
1= & \sum\limits_{r=1}\limits^{b+1} \mathbb{P}(T_0\ ranks\ r\ th|u_0,u_1,\ldots,u_b,z,y)\\

=&  \mathbb{P}(T_0\ ranks\ s\ th|u_0,u_1,\ldots,u_b,z,y)+ \sum\limits_{1\leq r \leq b+1,r\ne s}\mathbb{P}(T_0\ ranks\ r\ th|u_0,u_1,\ldots,u_b,z,y)\\

\ge & \mathbb{P}(T_0\ ranks\ s\ th,\widehat{KL}_{(1)}\leq \varepsilon_1,\ldots,\widehat{KL}_{(b)}\leq  \varepsilon_b|u_0,\ldots,u_b,z,y)\left(1+\sum\limits_{1\leq r \leq b+1,r\ne s}e^{-\varepsilon_{\Phi_{s-1}(r-1)}}\right)\\

= & \mathbb{P}(T_0\ ranks\ s\ th,\widehat{KL}_{(1)}\leq \varepsilon_1,\ldots,\widehat{KL}_{(b)}\leq  \varepsilon_b|u_0,\ldots,u_b,z,y)\left(1+\sum\limits_{r=1}\limits^{b}e^{-\varepsilon_r}\right),\\
\end{array}
$}
$$
where the last step is true due to $\Phi_s$'s one-to-one property, $s=0,1,\ldots,b$. Based on this result, we have
$$
\begin{array}{cl}
& \mathbb{P}(P_{CRRT}\leq \alpha,\widehat{KL}_{(1)}\leq \varepsilon_1,\ldots,\widehat{KL}_{(b)}\leq  \varepsilon_b|u_0,\ldots,u_b,z,y)\\
= & \sum\limits_{s=1}\limits^{\lfloor (b+1)\alpha\rfloor}\mathbb{P}(T_0\ ranks\ s\ th,\widehat{KL}_{(1)}\leq \varepsilon_1,\ldots,\widehat{KL}_{(b)}\leq  \varepsilon_b|u_0,\ldots,u_b,z,y)\\
\leq & \sum\limits_{s=1}\limits^{\lfloor (b+1)\alpha\rfloor}\left(1+\sum\limits_{r=1}\limits^{b}e^{-\varepsilon_r}\right)^{-1}\\
= & \frac{\lfloor (b+1)\alpha\rfloor}{1+e^{-\varepsilon_1}+\ldots,+e^{-\varepsilon_b}}.
\end{array}
$$
By taking expectations on both sides, the above result immediately gives us the desired conclusions stated in Theorem \ref{thm:robust2}.

\noindent {\bf Proof of Claim 1} First, we show that with probability 1, for fixed $z$ and $y$, for any $r\in \{1,\ldots,b+1\}$, there exists $u\in \{u_0,\ldots,u_b\}$, such that $x=u$ sufficiently lead to that $T_0$ ranks $r$th. We only need to focus on $y,z,u_0,u_1,\ldots,u_b$ such that for any $i\ne j$, $T_i\ne T_j$. Such event happens with probability 1 based on our assumption mentioned at the beginning of the proof. Let permutation $\sigma\in\Sigma_{b+1}$ satisfy that when $\tilde{x}=(u_0,\ldots,u_b)_{permute(\sigma)}=(u_{\sigma_1-1},\ldots,u_{\sigma_{b+1}-1})$, $T_0$ ranks $r$th. (Such permutation $\sigma$ does exist. Suppose that when $x=u_0,\ldots,x^{(b)}=u_b$, $T_0$ ranks $\hat{r}$th, and $T_{i_r}$ ranks $r$th. We just need to let $x=u_{i_r}$ and $x^{(i_r)}=u_0$. Based on $T$'s X-symmetry, now $T_0$ ranks $r$th.)

For any $\hat{\sigma}\in\Sigma_{b+1}$ such that $\hat{\sigma}_1=\sigma_1$, there exists a $\hat{\hat{\sigma}}\in\Sigma_{b+1}$ such that 
$$
[(u_0,\ldots,u_b)_{permute(\sigma)}]_{permute(\hat{\hat{\sigma}})}=(u_0,\ldots,u_b)_{permute(\hat{\sigma})}.
$$
Then, 
$$
\begin{array}{cl}
& T(y,z,\tilde{x}=(u_0,\ldots,u_b)_{permute(\hat{\sigma})})\\
= & T(y,z,\tilde{x}=[(u_0,\ldots,u_b)_{permute(\sigma)}]_{permute(\hat{\hat{\sigma}})})\\
= & T(y,z,\tilde{x}=(u_0,\ldots,u_b)_{permute(\sigma)})_{permute(\hat{\hat{\sigma}})}.
\end{array}
$$
Since permutation $\hat{\hat{\sigma}}$ doesn't change the 1st element, we have
$$
T_0(y,z,\tilde{x}=(u_0,\ldots,u_b)_{permute(\hat{\sigma})})=T_0(y,z,\tilde{x}=(u_0,\ldots,u_b)_{permute(\sigma)}),
$$
which implies that $T_0(y,z,\tilde{x}=(u_0,\ldots,u_b)_{permute(\hat{\sigma})})$ still ranks $r$th among $\{T_k,0\leq k\leq b\}$ as $T_0(y,z,\tilde{x}=(u_0,\ldots,u_b)_{permute(\sigma)})$. Hence, due to the arbitrariness of $\hat{\sigma}$, $x=u_{\sigma_1-1}$ is a sufficient condition that $T_0$ ranks $r$th.

Next, we show that for any permutations $\sigma,\sigma'\in\Sigma_{b+1}$, if $\sigma_1\ne \sigma'_1$, the rank of $T_0(y,z,\tilde{x}=(u_0,\ldots,u_b)_{permute(\sigma)})$, denoted by $r_{\sigma}$, is different from the rank of $T_0(y,z,\tilde{x}=(u_0,\ldots,u_b)_{permute(\sigma')})$, denoted by $r_{\sigma'}$. Suppose on the contrary that, $r_{\sigma}=r_{\sigma'}=r\in\{1,\ldots,b+1\}$. There exists $k\in \{2,\ldots,b+1\}$, such that $\sigma'_k=\sigma_1$. Let $\sigma''\in\Sigma_{b+1}$ satisfying that $\sigma''_g=\sigma'_g$, $g\ne1\ or\ k$, and $\sigma''_1=\sigma'_k$, $\sigma''_k=\sigma'_1$. Then, based on what've just shown, the rank of $T_0(y,z,\tilde{x}=(u_0,\ldots,u_b)_{permute(\sigma'')})$, denoted by $r_{\sigma''}$, equals $r_{\sigma}$. Because, $r_{\sigma}=r_{\sigma'}$, now we have $r_{\sigma''}=r_{\sigma'}$, which is impossible, since if we switch $x$ and $x_{k-1}$, we also switch the rank of $T_0$ and $T_{k-1}$. Combining 2 parts shown above, we can conclude that Claim 1 is true.
\end{proof}

\begin{proof}[\bf Proof of Theorem \ref{lowerbound1}]

For brevity, we denote $d_{TV}(Q_n^{\star}(\cdot|z),Q_n(\cdot|z))$ by $d_{TV}(z)$. According to the definition of $d_{TV}(z)$, there exists a set $A(z)\subseteq \mathbb{R}^n$ such that
$$
\mathbb{P}_{Q_n^{\star}}(x\in A(z)|z) = \mathbb{P}_{Q_n}(x\in A(z)|z) + d_{TV}(z),
$$
and 
$$
\left\{
\begin{array}{ll}
\frac{Q_n^{\star}(x|z)}{Q_n(x|z)}\ge 1, & \forall x\in A(z),\\
\frac{Q_n^{\star}(x|z)}{Q_n(x|z)}\leq 1, & \forall x\in A^c(z).
\end{array}
\right.
$$
Denote $\mathbb{P}_{Q_n}(x\in A(z)|z)$ by $\alpha_0(z)$. Let's first to consider the case when $\alpha_0(z)\ge \alpha-\varepsilon$. Let $\tilde{A}(z)\subseteq A(z)$ such that $\mathbb{P}_{Q_n}(x\in \tilde{A}(z)|z)=\alpha-\varepsilon$, and
$$
\mathbb{P}_{Q^{\star}_n}(x\in \tilde{A}(z)|z) \ge \sup\limits_{B(z)\subseteq A(z):\mathbb{P}_{Q_n}(x\in B(z)|z)=\alpha-\varepsilon} \mathbb{P}_{Q^{\star}_n}(x\in B(z)|z).
$$
Such $\tilde{A}(z)$ exists a.s. due to the assumption that $X$ is conditionally continuous under conditional distribution $Q(\cdot|Z)$ and $Q^{\star}(\cdot|Z)$ a.s.. Then, it's not hard to know that
$$
\tilde{\alpha}(z)\triangleq \mathbb{P}_{Q^{\star}_n}(x\in \tilde{A}(z)|z) \ge \alpha-\varepsilon + d_{TV}(z)\frac{\alpha-\varepsilon}{\alpha_0(z)}.
$$
Because $x^{(0)}=x\sim Q_n(\cdot|z)$ and $x^{(1)}\ldots,x^{(k)}\sim Q_n^{\star}(\cdot|z)$ independently, we have
$$
\left\{
\begin{array}{l}
\left( \sum\limits_{k=1}\limits^{b} \mathbbm{1}\{x^{(k)}\in\tilde{A}(z)\}|y,z\right) \sim Binomial(b,\alpha-\varepsilon),\\
\mathbbm{1} \{x^{(0)}\in \tilde{A}(z)\}|y,z \sim Bernoulli(\tilde{\alpha}(z)).
\end{array}
\right.
$$
For $k=0,1,\ldots, b$, let $T_k(y,z,\tilde{x})=\mathbbm{1}\{x^{(k)}\in\tilde{A}(z)\}$, and $T(y,z,\tilde{x})=(T_0,\ldots,T_b)^T$. Define $p_{CRRT}$ as in (\ref{pvalue}). Then, denoting $h(u)=(1+u)\log(1+u)-u$, as $b$ goes to infinity, we have
$$
\begin{array}{ll}
& \mathbb{P}(p_{CRRT}\leq \alpha|y,z)\\
= & \mathbb{P}\left(\frac{\sum\limits_{k=0}\limits^{b}\mathbbm{1}\{T_0(y,z,\tilde{x})\leq T_k(y,z,\tilde{x})\}}{1+b} \leq \alpha|y,z\right)\\
\ge & \mathbb{P}\left(T_0(y,z,\tilde{x})=1\ and\ \sum\limits_{k=1}\limits^{b}\mathbbm{1}\{T_k(y,z,\tilde{x})=1\} \leq (b+1)\alpha-1|y,z\right)\\
= & \tilde{\alpha}(z)\mathbb{P}\left(Binomial(b,\alpha-\varepsilon)\leq \alpha(b+1)-1\right)\\
= & \tilde{\alpha}(z)\left(1- \mathbb{P}\left(Binomial(b,\alpha-\varepsilon)-b(\alpha-\varepsilon)> \alpha-1+b\varepsilon\right) \right)\\
\ge &  \tilde{\alpha}(z)\left(1- exp\left\{-\frac{b(\alpha-\varepsilon)(1-(\alpha-\varepsilon))}{(1-(\alpha-\varepsilon))^2} h\left(\frac{(1-(\alpha-\varepsilon))(\alpha-1+b\varepsilon)}{b((\alpha-\varepsilon)(1-(\alpha-\varepsilon))} \right) \right\} \right)\\
= & \tilde{\alpha}(z)\left(1-O\left(exp\left\{-\frac{\alpha-\varepsilon}{1-(\alpha-\varepsilon)}h\left(\frac{\varepsilon}{\alpha-\varepsilon}\right)b \right\}\right) \right)\\
\ge & \left( \alpha-\varepsilon + d_{TV}(z)\frac{\alpha-\varepsilon}{\alpha_0(z)}\right)\left(1-O\left(exp\left\{-\frac{\alpha-\varepsilon}{1-(\alpha-\varepsilon)}h\left(\frac{\varepsilon}{\alpha-\varepsilon}\right)b \right\}\right) \right),\\
\end{array}
$$
where we use the conditional independence of $x,x^{(1)},\ldots,x^{(b)}$ in the 3rd step, we apply the Bennett's inequality in the 5th step and require $b$ to be sufficiently large such that $b>\frac{1-\alpha}{\varepsilon}$.

In the case when $\alpha_0(z)<\alpha-\varepsilon$, we proceed in a similar way. Let $\tilde{A}(z) \supseteq A(z)$ such that $\mathbb{P}_{Q_n}(x\in \tilde{A}(z)|z)=\alpha-\varepsilon$, and
$$
\mathbb{P}_{Q^{\star}_n}(x\in \tilde{A}(z)|z) \ge \sup\limits_{B(z)\supseteq A(z):\mathbb{P}_{Q_n}(x\in B(z)|z)=\alpha-\varepsilon} \mathbb{P}_{Q^{\star}_n}(x\in B(z)|z).
$$
Then, we have
$$
\begin{array}{rl}
& \tilde{\alpha}(z)=\mathbb{P}_{Q_n^{\star}}(x\in \tilde{A}(z)|z)\\

= & \mathbb{P}_{Q_n^{\star}}(x\in {A}(z)|z) + \mathbb{P}_{Q_n^{\star}}(x\in \tilde{A}(z) \cap [A(z)]^c|z) \\

\ge & [\alpha_0(z) + d_{TV}(z)] + \left\{ [\alpha-\varepsilon - \alpha_0(z)] - \frac{\alpha-\varepsilon - \alpha_0(z)}{1-\alpha_0(z)}d_{TV}(z)\right\}\\

= &\alpha-\varepsilon + d_{TV}(z)\frac{1-(\alpha-\varepsilon)}{1-\alpha_0(z)}.\\
\end{array}
$$
For $k=0,1,\ldots,b$, let $T_k(y,z,\tilde{x})=\mathbbm{1}\{x^{(k)}\in\tilde{A}(z)\}$, and $T(y,z,\tilde{x})=(T_0,\ldots,T_b)^T$. Similarly, we can show that
$$
\begin{array}{ll}
& \mathbb{P}(p_{CRRT}\leq \alpha|y,z)\\
\ge & \left(\alpha-\varepsilon + d_{TV}(z)\frac{1-(\alpha-\varepsilon)}{1-\alpha_0(z)}\right)\left(1-O\left(exp\left\{-\frac{\alpha-\varepsilon}{1-(\alpha-\varepsilon)}h\left(\frac{\varepsilon}{\alpha-\varepsilon}\right)b \right\}\right) \right).
\end{array}
$$
We want to point out that though the definitions of $\tilde{A}(z)$ in the above two situations are different, they depend only on $z,Q$ and $Q^{\star}$. We can summarize a compact definition of $\tilde{A}(z)$ as follows,
$$
\tilde{A}(z)=\left\{
\begin{array}{ll}
\mathop{\arg\sup}\limits_{B(z)\subseteq A(z):\mathbb{P}_{Q_n}(x\in B(z)|z)=\alpha-\varepsilon} \mathbb{P}_{Q^{\star}_n}(x\in B(z)|z), & if\ \alpha_0(z)\ge \alpha-\varepsilon,\\
\mathop{\arg\sup}\limits_{B(z)\supseteq A(z):\mathbb{P}_{Q_n}(x\in B(z)|z)=\alpha-\varepsilon} \mathbb{P}_{Q^{\star}_n}(x\in B(z)|z), & if\ \alpha_0(z)< \alpha-\varepsilon.\\
\end{array}
\right.
$$
For $k=0,1,\ldots,b$, let $T_k(y,z,\tilde{x})=\mathbbm{1}\{x^{(k)}\in\tilde{A}(z)\}$, $T(y,z,\tilde{x})=(T_0,\ldots,T_b)^T$, and $p_{CRRT}$ defined as in (\ref{pvalue}). It's easy to check that $T$ is $X$-symmetric. Combining the above results together, we have
$$
\begin{array}{ll}
& \mathbb{P}(p_{CRRT}\leq \alpha|y,z)\\
\ge & \left[ \alpha-\varepsilon + \mathbbm{1}\{\alpha_0(z)\ge \alpha-\varepsilon\}d_{TV}(z)\frac{\alpha-\varepsilon}{\alpha_0(z)} + \mathbbm{1}\{\alpha_0(z)< \alpha-\varepsilon\} d_{TV}(z)\frac{1-(\alpha-\varepsilon)}{1-\alpha_0(z)}\right]\\
& \times \left(1-O\left(exp\left\{-\frac{\alpha-\varepsilon}{1-(\alpha-\varepsilon)}h\left(\frac{\varepsilon}{\alpha-\varepsilon}\right)b \right\}\right) \right).
\end{array}
$$
Note that the definition of $T$ doesn't depend on $b$. If we allow $T$ to depend on $b$, we can replace $\varepsilon$ with some other functions of $b$ like $\frac{1}{2}\sqrt{\frac{\log b}{b}}$ and can have a tighter asymptotic bound.
\end{proof}

\begin{proof}[\bf Proof of Theorem \ref{lowerbound2}]

In the following proof, sometimes we use $x^{(0)}$ to denote $x$ alternatively. When we say $T_0$ ranks $r$th, we mean that $T_0$ is the $r$th largest value among $\{T_0,\ldots,T_{b}\}$. Since the main condition is a function of $\widehat{KL}_k,k=1,\ldots,b$, it's natural to consider
$$
T_k \triangleq D(x^{(k)},z) = \frac{Q^{\star}(x^{(k)}|z)}{Q(x^{(k)}|z)},k=0,1,\ldots,b.
$$
It's easy to see that $T=(T_0,T_1\ldots,T_b)$ is $X$-symmetric. As before, we denote the unordered set of $\{x,x^{(1)},\ldots,x^{(b)}\}$ by $\{u_0,\ldots,u_b\}$. Based on the assumption that $X|Z$ is a continuous variable under both of $Q(\cdot|Z)$ and $Q^{\star}(\cdot|Z)$, $T_0,T_1,\ldots,T_b$ are distinct with probability 1. 

We'll frequently use equation (\ref{goldenratio}). Therefore, we restate it here. For any $s,r\in\{1,\ldots,b+1\}$ and $s\ne r$,

\begin{equation}
\frac{\mathbb{P}(T_0\ ranks\ s\ th|u_0,u_1,\ldots,u_b,z,y)}{\mathbb{P}(T_0\ ranks\ r\ th|u_0,u_1,\ldots,u_b,z,y)}=\frac{Q^{\star}(u_{s-1}|z)Q(u_{r-1}|z)}{Q^{\star}(u_{r-1}|z)Q(u_{s-1}|z)}.
\label{goldenratio11}
\end{equation}

In the first setting, that is, when $x,x^{(1)},\ldots,x^{(b)}\sim_{i.i.d.}Q^{\star}(\cdot|z)$, since $T_0,T_1,\ldots,T_b$ are distinct with probability 1, events $\{\{T_0\ ranks\ s\ th\}:s=1,\ldots,b+1\}$ are of equal probability. Therefore,
$$
\mathbb{P}(P_{CRRT}\leq \alpha)=\sum\limits_{s=1}\limits^{\lfloor (b+1)\alpha\rfloor}\mathbb{P}(T_0\ ranks\ s\ th)=\frac{\lfloor (b+1)\alpha\rfloor}{b+1}.
$$
Next, we show that if $x\sim Q^{\star}(\cdot|z),x^{(k)}\sim Q(\cdot|z),k=1,\ldots,b$, independently, we can promise a nontrivial lower bound for the type 1 error. Denote the ordered statistics of $T_k,k=0,1,\ldots,b$, by $T_{(0)}\leq T_{(1)}\leq \ldots \leq T_{(b)}$. We start from our probability condition,
\begin{equation}
\begin{array}{rl}
c\leq & \mathbb{P}(\widehat{KL}_k\ge 0,k=1,\ldots,b,\ and,\ \widehat{KL}_{(\lambda)}\ge \varepsilon)\\
= & \mathbb{E}\mathbb{E}[\mathbbm{1}\{T_{(b)}/T_{(b-\lambda)}\ge e^{\varepsilon}\}\mathbbm{1}\{T_0\ ranks\ 1st\}|u_0,\ldots,u_b,z,y]\\
= &  \mathbb{E}\{\mathbbm{1}\{T_{(b)}/T_{(b-\lambda)}\ge e^{\varepsilon}\}\mathbb{E}[\mathbbm{1}\{T_0\ ranks\ 1st\}|u_0,\ldots,u_b,z,y]\}\\
=& \mathbb{E}\left[\mathbbm{1}\{T_{(b)}/T_{(b-\lambda)}\ge e^{\varepsilon}\}\frac{D(u_0,z)}{\sum\limits_{k=0}\limits^{b}D(u_k,z)}\right],\\
\end{array}
\label{inequalityc}
\end{equation}
where $D(u,z)\triangleq \frac{Q^{\star}(u|z)}{Q(u|z)}$ and the last step is based on (\ref{goldenratio11}). Next,
$$
\begin{array}{cl}
& \mathbb{P}(P_{CRRT}\leq \alpha|u_0,\ldots,u_b,z,y)\\
= & \mathbb{P}(\bigcup\limits_{s=1}\limits^{\lambda}\{T_0\ ranks\ s\ th\}|u_0,\ldots,u_b,z,y)\\
= & {\left[\sum\limits_{k=0}\limits^{\lambda-1}\frac{Q^{\star}(u_k|z)}{Q(u_k|z)}\right]}/{\left[\sum\limits_{k=0}\limits^{b}D(u_k,z)\right]},\\
\end{array}
$$
based on which, we further have
\begin{equation}
\begin{array}{cl}
& \mathbb{P}(P_{CRRT}\leq \alpha)\\
= &\mathbb{E}\left\{ {\left[\sum\limits_{k=0}\limits^{\lambda-1}D(u_k,z)\right]}/{\left[\sum\limits_{k=0}\limits^{b}D(u_k,z)\right]} \right\} \\
= & \frac{\lambda}{b+1}+ \mathbb{E}\left\{ {\left[\sum\limits_{k=0}\limits^{\lambda-1}D(u_k,z)\right]}/{\left[\sum\limits_{k=0}\limits^{b}D(u_k,z)\right]}-\frac{\lambda}{b+1} \right\} \\
\ge & \frac{\lambda}{b+1}+ \mathbb{E}\left\{ \left\{{\left[\sum\limits_{k=0}\limits^{\lambda-1}D(u_k,z)\right]}/{\left[\sum\limits_{k=0}\limits^{b}D(u_k,z)\right]}-\frac{\lambda}{b+1} \right\} \mathbbm{1}\{T_{(b)}/T_{(b-\lambda)}\ge e^{\varepsilon}\}\right\} \\
= & \frac{\lambda}{b+1}+ \mathbb{E}\left\{ \left\{\frac{\left[\sum\limits_{k=0}\limits^{\lambda-1}D(u_k,z)\right]}{\left[\sum\limits_{k=0}\limits^{b}D(u_k,z)\right]}-\frac{\lambda}{b+1} \right\} \frac{\sum\limits_{k=0}\limits^{b}D(u_k,z)}{D(u_0,z)}  \frac{D(u_0,z)}{\sum\limits_{k=0}\limits^{b}D(u_k,z)}\mathbbm{1}\left\{\frac{D(u_0,z)}{D(u_{\lambda},z)}\ge e^{\varepsilon}\right\}\right\}. \\
\end{array}
\label{mid2}
\end{equation}
Actually, we can see that
$$
\begin{array}{cl}
& \min\limits_{M_0\ge M_1 \ge\ldots \ge M_b\ge 0,\frac{M_0}{M_{\lambda}}\ge e^{\varepsilon}} \left\{\left(\frac{M_0+\ldots+M_{\lambda-1}}{M_0+\ldots+M_{b}}-\frac{\lambda}{b+1} \right)\frac{M_0+\ldots+M_{b}}{M_0} \right\}\\
=& \min\limits_{M_0\ge M_1 \ge\ldots \ge M_b\ge 0,\frac{M_0}{M_{\lambda}}\ge e^{\varepsilon}} \left\{ \frac{M_0+\ldots+M_{\lambda-1}}{M_0}- \frac{\lambda}{b+1}\frac{M_0+\ldots+M_{b}}{M_0} \right\} \\
= & \min\limits_{M_0\ge M_1 \ge\ldots \ge M_b\ge 0,\frac{M_0}{M_{\lambda}}\ge e^{\varepsilon}} \left\{\left(1-\frac{\lambda}{b+1} \right) \frac{M_0+\ldots+M_{\lambda-1}}{M_0} - \frac{\lambda}{b+1}\frac{M_{\lambda}+\ldots+M_{b}}{M_0} \right\} \\
= &  \min\limits_{M_0\ge M_{\lambda}\ge 0,\frac{M_0}{M_{\lambda}}\ge e^{\varepsilon}} \left\{\left(1-\frac{\lambda}{b+1} \right) \frac{M_0+(\lambda-1)M_{\lambda}}{M_0} - \frac{\lambda}{b+1}\frac{(b-\lambda+1)M_{\lambda}}{M_0} \right\}\\
= &  \min\limits_{M_0\ge M_{\lambda}\ge 0,\frac{M_0}{M_{\lambda}}\ge e^{\varepsilon}} \left[ \left(1-\frac{\lambda}{b+1} \right) - \left(1-\frac{\lambda}{b+1} \right)\frac{M_{\lambda}}{M_0} \right]\\
= & \left(1-\frac{\lambda}{b+1} \right) (1-e^{-\varepsilon}).\\
\end{array}
$$
Based on this result, we can proceed (\ref{mid2}) as
$$
\begin{array}{rl}
\mathbb{P}(P_{CRRT}\leq \alpha) \ge & \frac{\lambda}{b+1}+ \left(1-\frac{\lambda}{b+1} \right) (1-e^{\varepsilon}) \mathbb{E}\left[\frac{D(u_0,z)}{\sum\limits_{k=0}\limits^{b}D(u_k,z)}\mathbbm{1}\left\{\frac{D(u_0,z)}{D(u_{\lambda},z)}\ge e^{\varepsilon}\right\} \right]\\
\ge &  \frac{\lambda}{b+1} + c \left(1-\frac{\lambda}{b+1} \right) (1-e^{-\varepsilon}),\\
\end{array}
$$
where the last step holds based on (\ref{inequalityc}).

\end{proof}

\begin{proof}[\bf Proof of Proposition \ref{cptlower2}]
For given $x_{ordered}$ and $z$, let $A_R(x_{ordered},z)$ be the set satisfying that 
$$
\resizebox{.92\hsize}{!}{$
\int\limits_{x\in A_R(x_{ordered},z)}R^{\star}_n(x|x_{ordered},z)-R_n(x|x_{ordered},z)dx = d_{TV}(R^{\star}_n(\cdot|z,x_{ordered}),R_n(\cdot|z,x_{ordered})). 
$}
$$
Let $A(z)\triangleq \{x:x\in A_R(x_{ordered},z)\}$. Denote $\Omega_o=\{(r_1,\ldots,r_n)^T\in\mathbb{R}:r_1\leq\ldots\leq r_n\}$. Denote $DF(x_{ordered},z)\triangleq \sum\limits_{x\in\Gamma(x_{ordered})}(Q_n(x|z)-Q^{\star}_n(x|z))$. Then we have
$$
\resizebox{.92\hsize}{!}{$
\begin{array}{ll}
& \mathbb{E}\left[d_{TV}(R^{\star}_n(\cdot|z,x_{ordered}),R_n(\cdot|z,x_{ordered})) |z\right]\\
= & \int\limits_{x_{ordered}\in\Omega_o}\left[\int\limits_{x\in A_R(x_{ordered},z)}R^{\star}_n(x|x_{ordered},z)-R_n(x|x_{ordered},z)dx\right]\sum\limits_{x\in\Gamma(x_{ordered})}Q^{\star}_n(x|z)dx_{ordered}\\
\leq & \left| \int\limits_{x_{ordered}\in\Omega_o}\int\limits_{x\in A_R(x_{ordered},z)}R^{\star}_n(x|x_{ordered},z)dx\sum\limits_{x\in\Gamma(x_{ordered})}Q^{\star}_n(x|z)dx_{ordered}\right.\\
& \left.- \int\limits_{x_{ordered}\in\Omega_o}\int\limits_{x\in A_R(x_{ordered},z)}R_n(x|x_{ordered},z)dx\sum\limits_{x\in\Gamma(x_{ordered})}Q_n(x|z)dx_{ordered}\right|\\
& + \left| \int\limits_{x_{ordered}\in\Omega_o}\int\limits_{x\in A_R(x_{ordered},z)}R_n(x|x_{ordered},z)dx\sum\limits_{x\in\Gamma(x_{ordered})}(Q_n(x|z)-Q^{\star}_n(x|z))dx_{ordered} \right|\\
= & \left| \mathbb{P}_{Q^{\star}}(x\in A(z)|z)- \mathbb{P}_{Q}(x\in A(z)|z)  \right|\\
& + \left| \int\limits_{x_{ordered}\in\Omega_o}\int\limits_{x\in A_R(x_{ordered},z)}R_n(x|x_{ordered},z)dx\sum\limits_{x\in\Gamma(x_{ordered})}(Q_n(x|z)-Q^{\star}_n(x|z))dx_{ordered} \right|\\
\leq & d_{TV}(Q_n^{\star}(\cdot|z),Q_n(\cdot|z))\\
& + \left| \int\limits_{\mathop{x_{ordered}\in\Omega_o}\limits_{DF(x_{ordered},z)>0}}\int\limits_{x\in A_R(x_{ordered},z)}R_n(x|x_{ordered},z)dx DF(x_{ordered},z) dx_{ordered} \right.\\
& + \left. \int\limits_{\mathop{x_{ordered}\in\Omega_o}\limits_{DF(x_{ordered},z)<0}}\int\limits_{x\in A_R(x_{ordered},z)}R_n(x|x_{ordered},z)dx DF(x_{ordered},z) dx_{ordered} \right|\\
\leq & d_{TV}(Q_n^{\star}(\cdot|z),Q_n(\cdot|z))\\
& + \max\left\{\left|\int\limits_{\mathop{x_{ordered}\in\Omega_o}\limits_{DF(x_{ordered},z)>0}} DF(x_{ordered},z) dx_{ordered}\right|,\left|\int\limits_{\mathop{x_{ordered}\in\Omega_o}\limits_{DF(x_{ordered},z)<0}} DF(x_{ordered},z) dx_{ordered}\right|\right\}\\
\leq & 2d_{TV}(Q_n^{\star}(\cdot|z),Q_n(\cdot|z)),
\end{array}
$}
$$
where the last step is valid because
$$
\begin{array}{cl}
& \left|\int\limits_{\mathop{x_{ordered}\in\Omega_o}\limits_{DF(x_{ordered},z)>0}} DF(x_{ordered},z) dx_{ordered}\right|\\
= & \left| \mathbb{P}_{Q^{\star}}(x\in\{v\in\mathbb{R}^n:v\in \Gamma(v_{ordered}),v_{ordered}\in\Omega_o,DF(v_{ordered},z)>0\}|z) \right.\\
& \left. - \mathbb{P}_{Q}(x\in\{v\in\mathbb{R}^n:v\in \Gamma(v_{ordered}),v_{ordered}\in\Omega_o,DF(v_{ordered},z)>0\}|z) \right|\\
\leq & d_{TV}(Q_n^{\star}(\cdot|z),Q_n(\cdot|z)),
\end{array}
$$
and similar result holds for $\left|\int\limits_{\mathop{x_{ordered}\in\Omega_o}\limits_{DF(x_{ordered},z)<0}} DF(x_{ordered},z) dx_{ordered}\right|$.
\end{proof}

\begin{proof}[\bf Proof of Proposition \ref{multimodified}]
Since tied $T_j$ values can only make the given modified multiple knockoffs procedure more conservative, for simplicity, we assume that for $j=1,2,\ldots,p$, $\mathbb{P}(\exists 1\leq k<k'\leq b+1, T_j^{k}=T_j^{k'})=0$. Then, we can make a claim similar to the Lemma 3.2 in Gemenez and Zou (2018):

\noindent {\bf Claim 2}: The ranks corresponding to null variables, $\{r_j:j\in S^c\}$, are i.i.d. uniformly on the set $\{1,2,\ldots,b+1\}$, and are independent of the other ranks $\{r_j:j\in S\}$, and of $\check{T}=\{\{T^{(1)}_j,\ldots,T^{(b+1)}_j\}:j=1,2,\ldots,p\}$.

We skip the proof of this claim because we can show it in an almost same way as in the Lemma 3.2 of Gemenez and Zou (2018). Recall that
$$
\tilde{\tau} = \min\left\{t\ge \min\{\tau_j:1\leq j\leq p\}:\frac{\frac{\tilde{\lambda}}{b+1-\tilde{\lambda}}\left(1+\#\{1\leq j\leq p,r_j>\tilde{\lambda},\tau_j\ge t\}\right)}{\#\{1\leq j\leq p,r_j\leq \tilde{\lambda},\tau_j\ge t\}\vee 1}\leq \alpha\right\}.
$$
Similar to Barber and Cand{\`e}s (2015) and Gemenez and Zou (2018), we construct one-bit p-values for $j=1,\ldots,p$:
$$
p_j=\left\{
\begin{array}{cc}
\frac{\tilde{\lambda}}{b+1}, & r_j\leq \tilde{\lambda},\\
1, & r_j>\tilde{\lambda}.
\end{array}
\right.
$$
Based on the Claim 2, for $j\in \{1,\ldots,p\}\setminus S$, we have 
$$
\left\{
\begin{array}{l}
\mathbb{P}(p_j=\frac{\tilde{\lambda}}{b+1}|\check{T})=\mathbb{P}(p_j=\frac{\tilde{\lambda}}{b+1})=\mathbb{P}(r_j\leq \tilde{\lambda})=\frac{\tilde{\lambda}}{b+1},\\
\mathbb{P}(p_j=1|\check{T})=\mathbb{P}(p_j=1)=\mathbb{P}(r_j>\tilde{\lambda})=\frac{b+1-\tilde{\lambda}}{b+1},
\end{array}
\right.
$$
which implies that $p_j|\check{T} \mathop{\ge}\limits^{d} \mathscr{U}([0,1])$. Since for $j\in\{1,\ldots,p\}$, $p_j$ is a function of $r_j$, based on the Claim 2, all $p_j$'s corresponding to null variables are mutually independent and are independent from the nonnulls, conditional on $\check{T}$.

Before defining a Selective SeqStep+ sequential test, we adjust the order of variables. Because all of our arguments are conditional on $\check{T}$, it's legitimate to rearrange the order of $X_1,\ldots,X_p$ based on some functions of $\check{T}$. In particular, we can adjust the order of variables based on $\tau_1,\ldots,\tau_p$. Hence, we can directly assume that
$$
\tau_1\ge \tau_2 \ge \ldots \ge \tau_p \ge0.
$$
Next, we define a Selective SeqStep+ threshold
$$
\tilde{k}=\max \left\{1\leq k\leq p:\frac{1+\#\{j\leq k:p_j>\frac{\tilde{\lambda}}{b+1}\}}{\#\{j\leq k:p_j\leq \frac{\tilde{\lambda}}{b+1}\}\vee 1}\leq \alpha \frac{b+1-\tilde{\lambda}}{\tilde{\lambda}}\right\}
$$
and $\tilde{S}'=\left\{j\leq \tilde{k}:p_j\leq \frac{\tilde{\lambda}}{b+1}\right\}$. Noting that $\frac{\tilde{\lambda}}{b+1}/(1-\frac{\tilde{\lambda}}{b+1})=\frac{\tilde{\lambda}}{b+1-\tilde{\lambda}}$, we can assert that
$$
\mathbb{E}\left[\frac{|\tilde{S}'\setminus S|}{|\tilde{S}'|\vee 1}\right]=\mathbb{E}\mathbb{E}\left[\frac{|\tilde{S}'\setminus S|}{|\tilde{S}'|\vee 1}|\check{T}\right]\leq \alpha
$$
by applying the Theorem 3 in Barber and Cand{\`e}s (2015). Our last job is to show that $\tilde{S}'$ is identical with $\tilde{S}$. We have
$$
\begin{array}{ll}
\tilde{k} & =\max \left\{1\leq k\leq p:\frac{1+\#\{j\leq k:p_j>\frac{\tilde{\lambda}}{b+1}\}}{\#\{j\leq k:p_j\leq \frac{\tilde{\lambda}}{b+1}\}\vee 1}\leq \alpha \frac{b+1-\tilde{\lambda}}{\tilde{\lambda}}\right\}\\
& = \max \left\{1\leq k\leq p:\frac{1+\#\{j\leq k:r_j>\tilde{\lambda},\tau_j\ge \tau_k\}}{\#\{j\leq k:r_j\leq \tilde{\lambda},\tau_j\ge \tau_k\}\vee 1}\leq \alpha \frac{b+1-\tilde{\lambda}}{\tilde{\lambda}}\right\}\\
& =  \max \left\{1\leq k\leq p:\frac{\frac{\tilde{\lambda}}{b+1-\tilde{\lambda}}(1+\#\{j\leq k:r_j>\tilde{\lambda},\tau_j\ge \tau_k\})}{\#\{j\leq k:r_j\leq \tilde{\lambda},\tau_j\ge \tau_k\}\vee 1}\leq \alpha \right\},
\end{array}
$$
from which we can tell $\tau_{\tilde{k}}=\tilde{\tau}$. Hence,
$$
\begin{array}{ll}
\tilde{S} & = \{1\leq j\leq p:r_j\leq \tilde{\lambda}, \tau_j\ge \tilde{\tau}\}\\
& = \{1\leq j\leq p:r_j\leq \tilde{\lambda}, \tau_j\ge\tau_{\tilde{k}}\}\\
& = \{1\leq j\leq \tilde{k}:r_j\leq \tilde{\lambda}\}\\
& = \{1\leq j\leq \tilde{k}: p_j = \frac{\tilde{\lambda}}{b+1}\}\\
& = \tilde{S}'.
\end{array}
$$
Thus,
$$
FDR = \mathbb{E}\left[\frac{|\tilde{S}\setminus S|}{|\tilde{S}|\vee 1}\right] = \mathbb{E}\left[\frac{|\tilde{S}'\setminus S|}{|\tilde{S}'|\vee 1}\right] \leq \alpha,
$$
which concludes the proof.
\end{proof}

\section{Supplementary Simulations}

\subsection{Impact of the Number of Conditional Samplings}
\label{appendix:b}
In this part, we inspect the impact of $b$, the number of conditional samplings on the performance of the CRT and the dCRT, which are 2 main competitors of our method, in several different model settings. Based on the Theorem \ref{maintheorem}, we know that the choice of $b$ doesn't affect the validity of type 1 error control of the CRRT. Actually, this conclusion also applies to other conditional sampling based methods considered in this paper, like CRT and dCRT. As mentioned in the first remark of the Theorem \ref{thm:robust2}, larger $b$ leads to less conservative tests. We demonstrate by the following simulations that $b=200$ is sufficiently large for the CRT and the dCRT in the sense that the amount of removable conservativeness is negligible. Based on results in this subsection, we fix $b=200$ in further simulations for all conditional sampling based methods.
\subsubsection{Impact of the Number of Conditional Samplings on the Performance of Distillation-Based Tests}
Here, we assess the impact of $b$ on the performance of distillation-based tests (Liu, Katsevich, Janson and Ramdas, 2020), including $d_0CRT$, $d_2CRT$ and $d_{12}CRT$. We consider 3 settings, including $n>p$ linear regression, $n=p$ linear regression and $n>p$ logistic regression. For each specific setting, we run 100 Monte-Carlo simulations.

\noindent {\bf ${\bf n>p}$ linear regression}: We set $n=500$, $p=100$ and observations are i.i.d.. For $i=1,2,\ldots,500$, $(x_i,z_{[i]}^T)^T$ is an i.i.d. realization of $(X,Z^T)^T$, where $(X,Z^T)^T$ is a $(p+1)$-dimension random vector following a Gaussian $AR(1)$ process with autoregressive coefficient 0.5. We let
$$
y_i = \beta_0 x_i + z_{[i]}^T\beta+\varepsilon_i,\ i=1,2,\ldots,500,
$$
where $\beta_0\in \mathbb{R}$, $\beta\in \mathbb{R}^{p}$, $\varepsilon_i$'s are i.i.d. standard Gaussian, independent of $x_i$ and $z_{[i]}$. For simplicity, we fix $\beta$ by letting $\beta_1=\beta_2=0.5$ and $0$ elsewhere. We consider various values of $\beta_0$, including $0,0.1,0.2,0.3,0.4$. We use linear Lasso in the distillation step, with regularized parameter tuned by cross-validation, and use linear least square regression in the testing step.

\noindent {\bf ${\bf n=p}$ linear regression}: Basically all things are same as those in the previous setting except that the dimension of $Z$ is 500 instead of 100. 

\noindent {\bf ${\bf n>p}$ logistic regression}: Models for covariates are same as those in the first setting while the model for response variable is different. We let
$$
y_i = \mathbbm{1}\left\{\beta_0 x_i + z_{[i]}^T\beta+\varepsilon_i>0\right\},\ i=1,2,\ldots,500,
$$
where $\beta_0\in \mathbb{R}$, $\beta\in \mathbb{R}^{p}$, $\varepsilon_i$'s are i.i.d. standard Gaussian, independent of $x_i$ and $z_{[i]}$. We use logistic Lasso in the distillation step, with regularized parameter tuned by cross-validation, and use linear least square regression in the testing step.

\begin{figure}[h]
    \centering
    \begin{subfigure}[t]{0.3\textwidth}
        \centering
        \includegraphics[width=1\linewidth,height=0.7\linewidth]{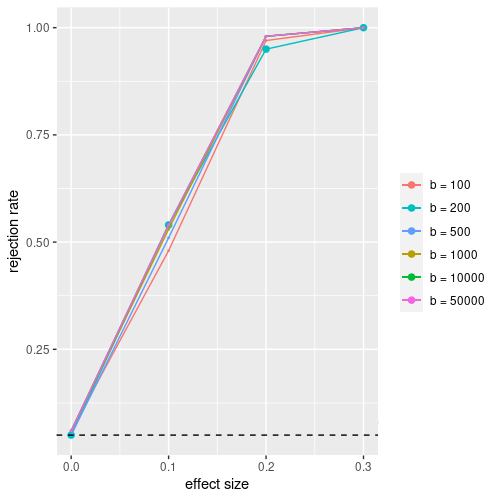} 
        \caption{$n>p$ linear regression} \label{fig:testb1}
    \end{subfigure}
    \hfill
    \begin{subfigure}[t]{0.3\textwidth}
        \centering
        \includegraphics[width=1\linewidth,height=0.7\linewidth]{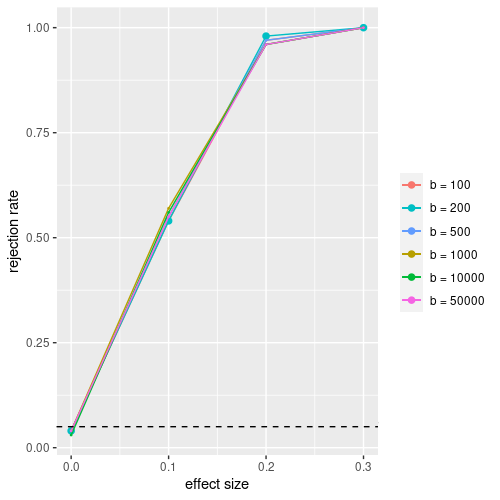} 
        \caption{$n=p$ linear regression} \label{fig:testb2}
    \end{subfigure}
    \hfill
    \begin{subfigure}[t]{0.3\textwidth}
        \centering
        \includegraphics[width=1\linewidth,height=0.7\linewidth]{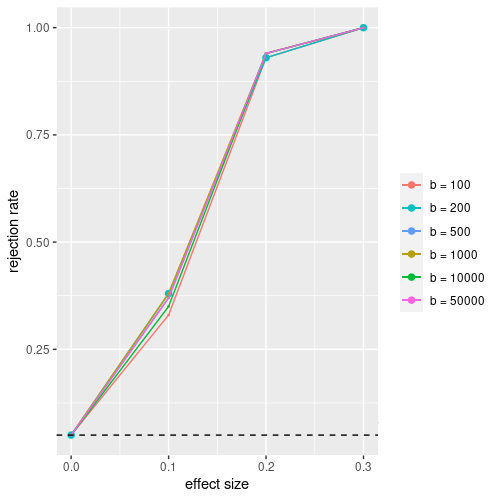} 
        \caption{$n>p$ logistic regression} \label{fig:testb3}
    \end{subfigure}
\caption{Proportions of rejection given by the $d_0CRT$ under 3 choices of $b$ in 3 settings. The y-axis represents the proportion of times the null hypothesis of conditional independence is rejected out of 100 simulations. The x-axis represents the true coefficient of $X$ in the models.}
\label{testbd0crt}
\end{figure}

\begin{figure}[h]
    \centering
    \begin{subfigure}[t]{0.3\textwidth}
        \centering
        \includegraphics[width=1\linewidth,height=0.7\linewidth]{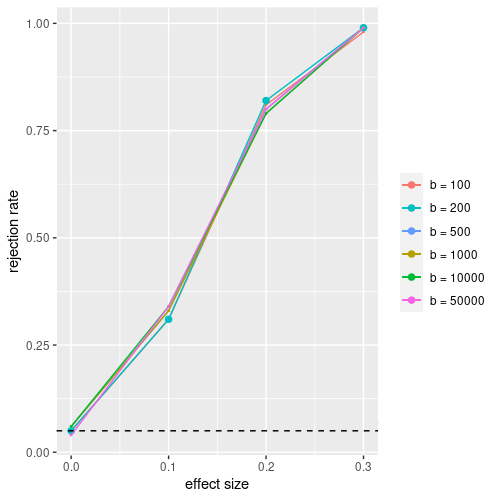} 
        \caption{$n>p$ linear regression} \label{fig:testb4}
    \end{subfigure}
    \hfill
    \begin{subfigure}[t]{0.3\textwidth}
        \centering
        \includegraphics[width=1\linewidth,height=0.7\linewidth]{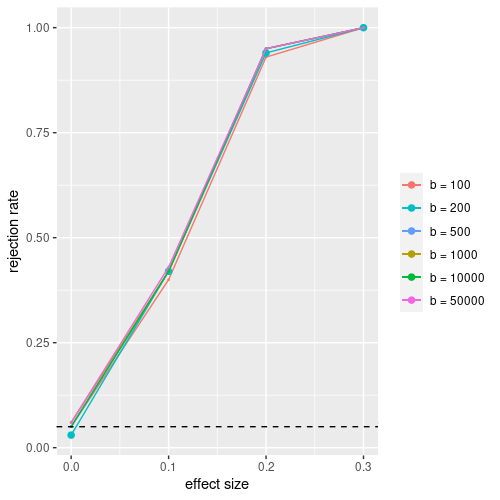} 
        \caption{$n=p$ linear regression} \label{fig:testb5}
    \end{subfigure}
    \hfill
    \begin{subfigure}[t]{0.3\textwidth}
        \centering
        \includegraphics[width=1\linewidth,height=0.7\linewidth]{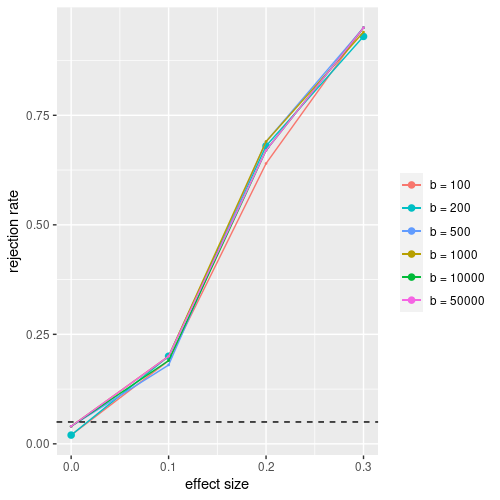} 
        \caption{$n>p$ logistic regression} \label{fig:testb6}
    \end{subfigure}
\caption{Proportions of rejection given by the $d_2CRT$ under 3 choices of $b$ in 3 settings. The y-axis represents the proportion of times the null hypothesis of conditional independence is rejected out of 100 simulations. The x-axis represents the true coefficient of $X$ in the models.}
\label{testbd2crt}
\end{figure}

\begin{figure}[htp]
    \centering
    \begin{subfigure}[t]{0.3\textwidth}
        \centering
        \includegraphics[width=1\linewidth,height=0.7\linewidth]{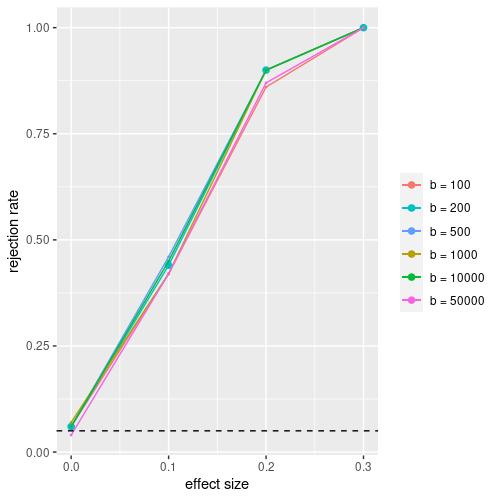} 
        \caption{$n>p$ linear regression} \label{fig:testb7}
    \end{subfigure}
    \hfill
    \begin{subfigure}[t]{0.3\textwidth}
        \centering
        \includegraphics[width=1\linewidth,height=0.7\linewidth]{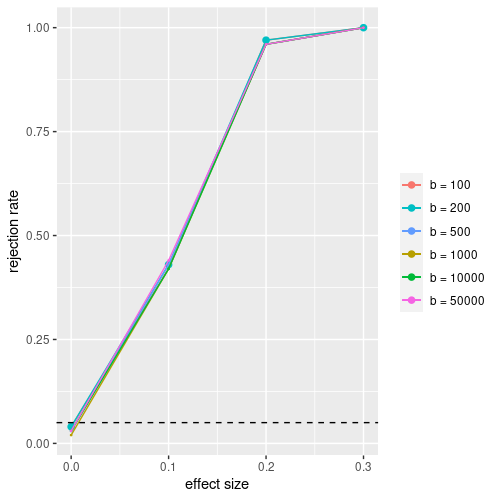} 
        \caption{$n=p$ linear regression} \label{fig:testb8}
    \end{subfigure}
    \hfill
    \begin{subfigure}[t]{0.3\textwidth}
        \centering
        \includegraphics[width=1\linewidth,height=0.7\linewidth]{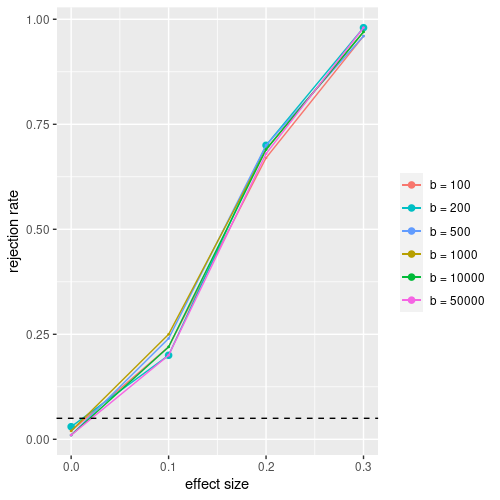} 
        \caption{$n>p$ logistic regression} \label{fig:testb9}
    \end{subfigure}
\caption{Proportions of rejection given by the $d_{12}CRT$ under 3 choices of $b$ in 3 settings. The y-axis represents the proportion of times the null hypothesis of conditional independence is rejected out of 100 simulations. The x-axis represents the true coefficient of $X$ in the models.}
\label{testbd12crt}
\end{figure}

\noindent {\bf Results:} From Figure \ref{testbd0crt}, \ref{testbd2crt} and \ref{testbd12crt}, we can observe that each of the 3 distillation-based tests considered here show similar performance under different choices of $b$ varying from 100 to 50000. In most cases, $b=200$ leads to desirable observed type 1 error, which is less than 0.05, while $b=100$ sometimes brings type 1 error exceeding the threshold. There's no obvious evidence showing that we can gain much greater power by increasing $b$ from 200 to 50000. Hence, we decide to take $b=200$ in subsequent analysis.
\subsubsection{Impact of the Number of Conditional Samplings on the Performance of the CRT}
We also assess the impact of $b$ on the performance of the CRT with same simulation settings as those for distillation-based tests. Unlike distillation-based tests, the original CRT introduced in Cand{\`e}s, Fan, Janson and Lv (2018) is not a two-stage test. In the 2 linear regression settings, we evaluate the importance of variables (i.e. test statistics) by linear Lasso fitted coefficients. In the logistic regression settings, we obtain the variables importance by logistic Lasso fitted coefficients. Due to the limitation of computation, we are not able to simulate the CRT with extremely large $b$. For each specific setting, we run 100 Monte-Carlo simulations. 

\begin{figure}[h]
    \centering
    \begin{subfigure}[t]{0.3\textwidth}
        \centering
        \includegraphics[width=1\linewidth,height=0.7\linewidth]{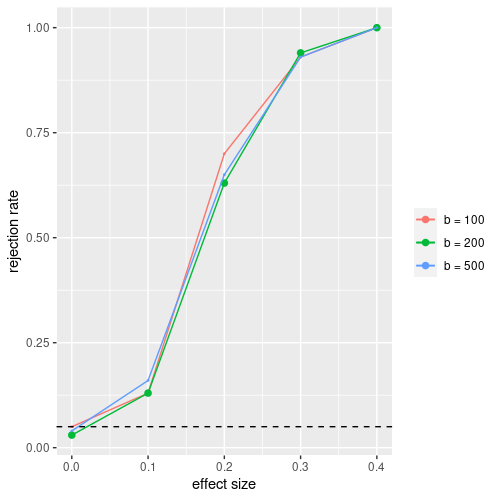} 
        \caption{$n>p$ linear regression} \label{fig:testbcrt1}
    \end{subfigure}
    \hfill
    \begin{subfigure}[t]{0.3\textwidth}
        \centering
        \includegraphics[width=1\linewidth,height=0.7\linewidth]{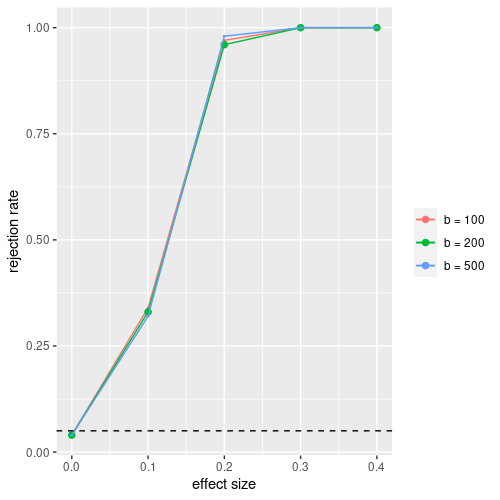} 
        \caption{$n=p$ linear regression} \label{fig:testbcrt2}
    \end{subfigure}
    \hfill
    \begin{subfigure}[t]{0.3\textwidth}
        \centering
        \includegraphics[width=1\linewidth,height=0.7\linewidth]{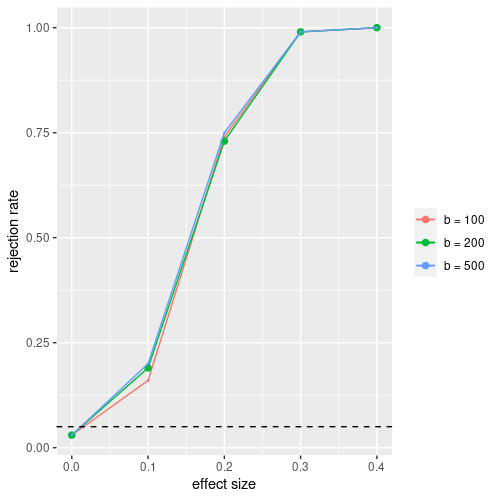} 
        \caption{$n>p$ logistic regression} \label{fig:testbcrt3}
    \end{subfigure}
\caption{Proportions of rejection given by the CRT under 3 choices of $b$ in 3 settings. The y-axis represents the proportion of times the null hypothesis of conditional independence is rejected out of 100 simulations. The x-axis represents the true coefficient of $X$ in the models.}
\label{testbcrt}
\end{figure}

\noindent {\bf Results:} From Figure \ref{testbcrt}, we can see that there's no significant difference in performance among 3 choices of $b$ in 3 considered settings. In particular, no matter what's the value of $b$, the CRT always shows observed type 1 error control under the null hypothesis of conditional independence, which is encouraging. Of course, as mentioned in Cand{\`e}s, Fan, Janson and Lv (2018), there could be some settings in which we need an extremely large $b$ to guarantee that the CRT has good observed type 1 error control in practice. Besides, the CRT under larger $b$ doesn't show obvious advantage in power. Therefore, out of simultaneous consideration of theoretical power and computational burden, we set $b=200$ for further simulations.

\subsection{More Complicated Models}

In the main content, we consider the case of continuous response variable. Here, we provide further investigations in cases of categorical response variable.

Recall that $n=400$ and the number of Monte-Carlo simulations is 200 for each specific setting. For $i=1,2,\ldots,400$, $(x_i,z_{[i]}^T)^T$ is an i.i.d. realization of $(X,Z^T)^T$, where $(X,Z^T)^T$ is a 101-dimension random vector following a Gaussian $AR(1)$ process with autoregressive coefficient 0.5. We consider the following categorical response model with different forms of interactions: 
$$
\resizebox{.92\hsize}{!}{$
\begin{array}{ll}
{\bf Interaction\ Model\ 2:}& \left\{
\begin{array}{l}
\ddot{y}_i = \beta_0 x_i\sum\limits_{j=1}\limits^{100}z_{[i]j} + z_{[i]}^T\ddot{\beta}+\ddot{\varepsilon_i},\ i=1,2,\ldots,400,\\
\mathring{y}_i = exp\left(\beta_0 x_i\sum\limits_{j=1}\limits^{100}z_{[i]j}\right) + (z_{[i]}^T\mathring{\beta})^3-1+\mathring{\varepsilon_i},\ i=1,2,\ldots,400,\\
y_i=\left\{\begin{array}{ll}
0, & if\ \ddot{y}_i>0\ and\ \mathring{y}_i>0,\\
1, & if\ \ddot{y}_i>0\ and\ \mathring{y}_i\leq0,\\
2, & if\ \ddot{y}_i\leq0\ and\ \mathring{y}_i>0,\\
3, & if\ \ddot{y}_i\leq0\ and\ \mathring{y}_i\leq0,\\
\end{array} \right.\  i=1,2,\ldots,400,\\
\end{array} \right.\\
{\bf Interaction\ Model\ 3:}& \left\{
\begin{array}{l}
\ddot{y}_i = \beta_0 x_i\sum\limits_{j=1}\limits^{100}z_{[i]j} + z_{[i]}^T\ddot{\beta}+\ddot{\varepsilon_i},\ i=1,2,\ldots,400,\\
\mathring{y}_i = \beta_0 exp\left(\beta_0 x_i\sum\limits_{j=1}\limits^{100}z_{[i]j}z_{[i]j+1}\right) + z_{[i]}^T\mathring{\beta}+\mathring{\varepsilon_i},\ i=1,2,\ldots,400,\\
y_i=\left\{\begin{array}{ll}
0, & if\ \ddot{y}_i>0\ and\ \mathring{y}_i>0,\\
1, & if\ \ddot{y}_i>0\ and\ \mathring{y}_i\leq0,\\
2, & if\ \ddot{y}_i\leq0\ and\ \mathring{y}_i>0,\\
3, & if\ \ddot{y}_i\leq0\ and\ \mathring{y}_i\leq0,\\
\end{array} \right.\  i=1,2,\ldots,400,\\
\end{array} \right.\\
\end{array}
$}
$$
where $\beta_0\in \mathbb{R}$, $\beta,\ddot{\beta},\mathring{\beta}\in \mathbb{R}^{100}$, $\varepsilon_i,\ddot{\varepsilon}_i,\mathring{\varepsilon}_i$'s are i.i.d. standard Gaussian, independent of $x_i$ and $z_{[i]}$. To construct $\beta$, we randomly sample a set of size 5 without replacement from $\{1,2,\ldots,100\}$, say $S$. For $j=1,\ldots,100$, let $\beta_j=0$ if $j\notin S$, and $\beta_j=0.5\mathbbm{1}\{B_j=1\}-0.5\mathbbm{1}\{B_j=0\}$ if $j\in S$, where $B_j$ is Bernoulli. $\ddot{\beta}$ and $\mathring{\beta}$ are also constructed in the same way independently. We consider various values of $\beta_0$ for different models. The difference between the 3rd and the 2nd model is that, in the 3rd one, covariates entangle with each other in a more involved way.

\begin{figure}[h]
    \centering
    \begin{subfigure}[t]{0.49\textwidth}
        \centering
        \includegraphics[width=1\linewidth,height=0.7\linewidth]{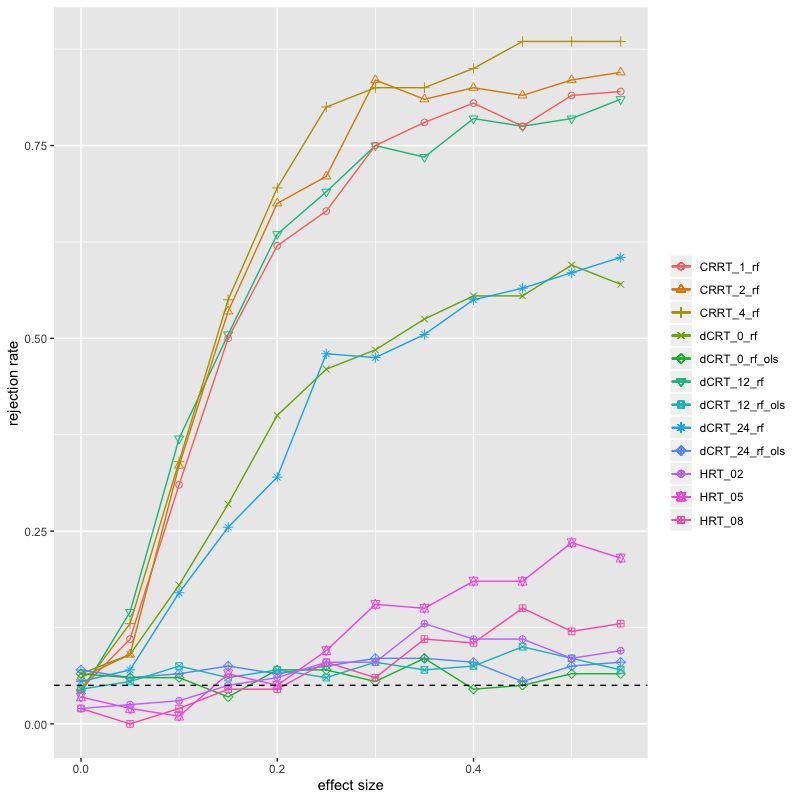} 
        \caption{Rejection Rate} \label{fig:ir2}
    \end{subfigure}
    \hfill
    \begin{subfigure}[t]{0.49\textwidth}
        \centering
        \includegraphics[width=1\linewidth,height=0.7\linewidth]{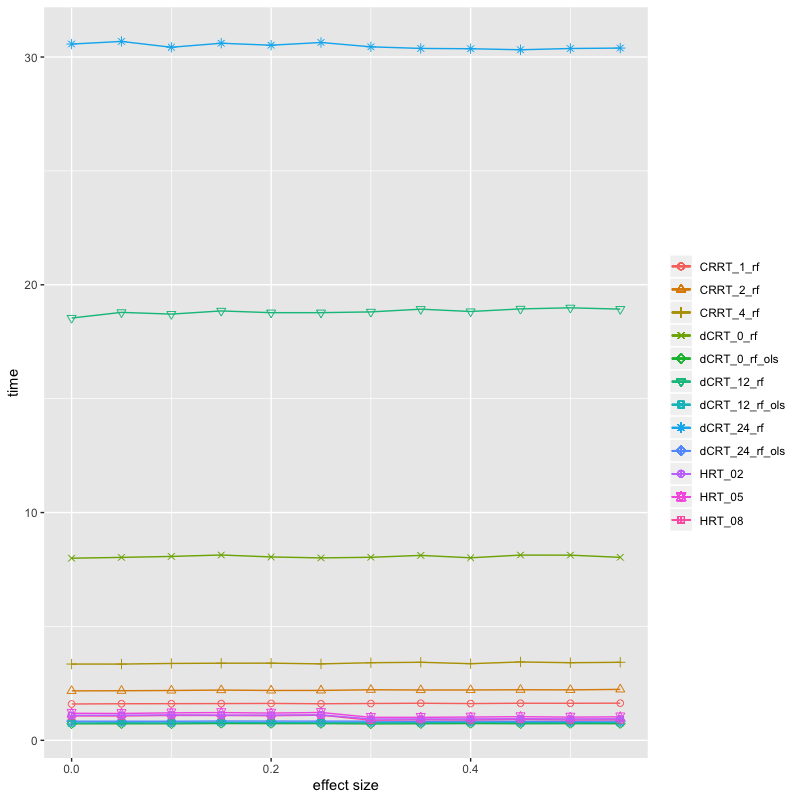} 
        \caption{Time} \label{fig:it2}
    \end{subfigure}
\caption{(a) Proportions of rejection under the Interaction Model 2. The y-axis represents the proportion of times the null hypothesis of conditional independence is rejected out of 200 simulations. The x-axis represents the true coefficient of $X$ in the models. (b) Time in seconds per Monte-Carlo simulation. dCRT\_k\_rf represents the dCRT with random forest in both of the distillation step and the testing step, and keeping k important variables after distillation. dCRT\_k\_rf\_ols represents the dCRT with random forest in the distillation step, least square linear regression in the testing step, and keeping k important variables after distillation. CRRT\_k\_rf represents the CRRT using random forest as the test function and with batch size $(b+1)/k$. HRT\_0k represents the HRT fitting a random forest model and using training set of size $n\times k\times 0.1$.}
\label{inter_class1}
\end{figure}

\begin{figure}[h]
    \centering
    \begin{subfigure}[t]{0.49\textwidth}
        \centering
        \includegraphics[width=1\linewidth,height=0.7\linewidth]{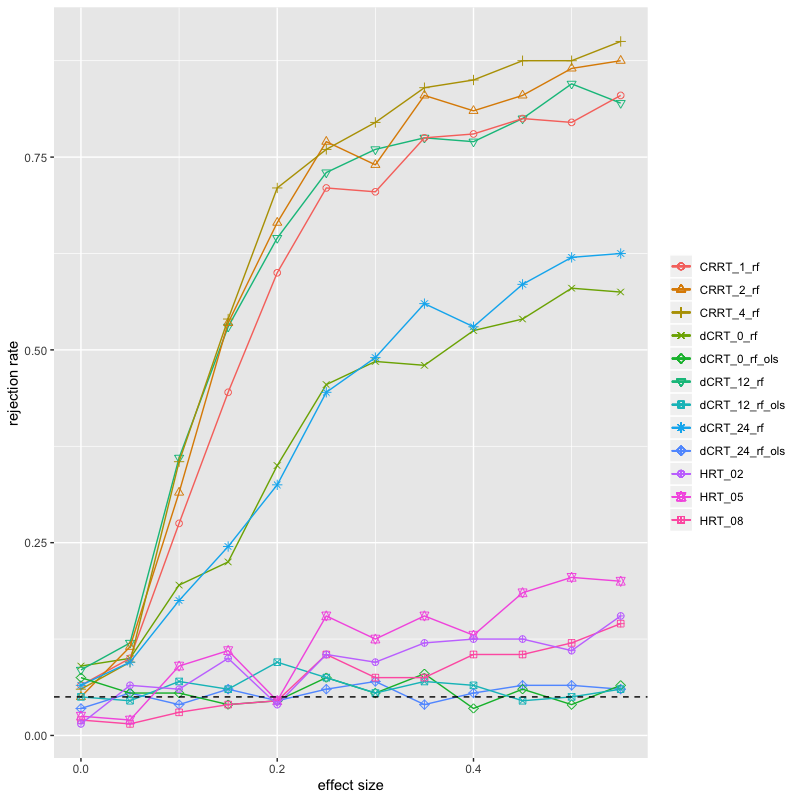} 
        \caption{Rejection Rate} \label{fig:ir3}
    \end{subfigure}
    \hfill
    \begin{subfigure}[t]{0.49\textwidth}
        \centering
        \includegraphics[width=1\linewidth,height=0.7\linewidth]{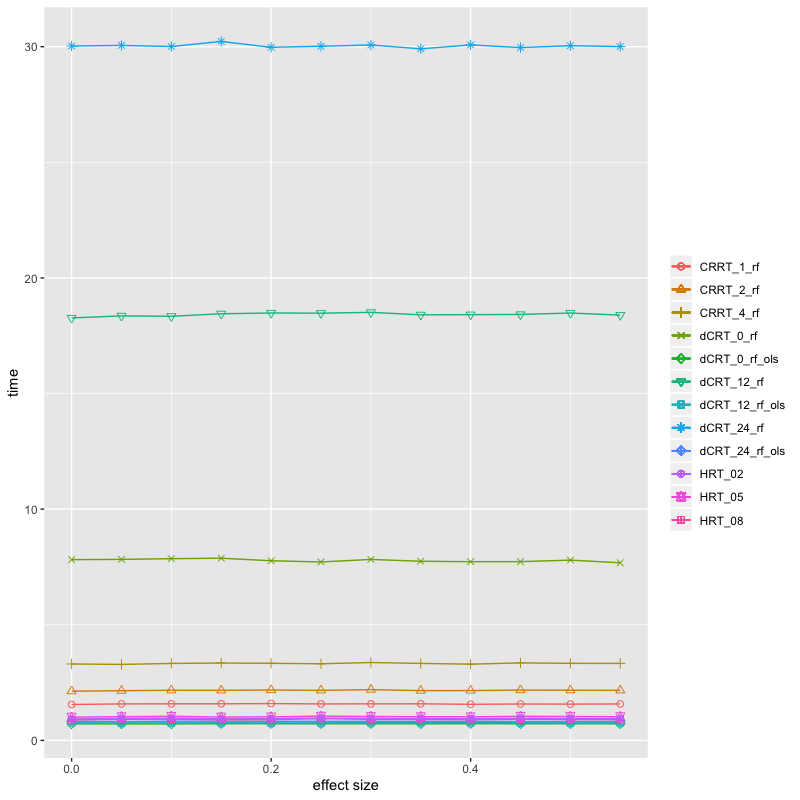} 
        \caption{Time} \label{fig:it3}
    \end{subfigure}
\caption{(a) Proportions of rejection under the Interaction Model 3. The y-axis represents the proportion of times the null hypothesis of conditional independence is rejected out of 200 simulations. The x-axis represents the true coefficient of $X$ in the models. (b) Time in seconds per Monte-Carlo simulation. Definitions of legends can be found in the Figure \ref{inter_class1}.}
\label{inter_class2}
\end{figure}

\noindent {\bf Results:} Results under the Interaction Model 2 and 3 are provided in Figure \ref{inter_class1} and \ref{inter_class2}, where the response variable $Y$ is categorical. Because results under these 2 models are similar, we analyze them together. The CRRT's and the dCRT\_12\_rf are among the most powerful tests. However, the dCRT\_12\_rf is in average at least 4 times slower than the CRRT\_4\_rf, which is consistently the most powerful test. Again, increasing the number of important variables kept from the distillation step doesn't help to improve power while costing more computations. Taking both power and efficiency into consideration, the CRRT with batch size 50 is obviously the best one.

\subsection{Supplementary Results to the Subsection \ref{robustsim}}
\label{mixedeffect}
It's well comprehensible that misspecification of the conditional distribution of $X|Z$ can not only inflate type 1 error like what we've shown in the Subsection \ref{robustsim}, but can sometimes drive it downwards. This could happen and we can easily construct simple intuitive example as follows. Suppose that $Y=Z+\varepsilon$, where both $Y$ and $Z$ are scalars and $\varepsilon\sim N(0,1)$. Let $X\sim N(0,1)$, independent of $Z$ and $\varepsilon$ while $X^{(k)}=\theta Z$ for $k=1,\ldots,b$. As long as $\theta\ne 0$, pseudo variables are more related with $Y$ than $X$ and most methods would tend to believe that pseudo variables are more conditionally important, which brings the type 1 error downwards. Here, we give more examples under non data-dependent misspecification to show the mixed influence of conditional distribution misspecification. We consider the following 4 models.

\noindent {\bf ${\bf n>p}$ linear regression:} We set $n=400$, $p=100$ and observations are i.i.d.. For $i=1,2,\ldots,400$, $(x_i,z_{[i]}^T)^T$ is an i.i.d. realization of $(X,Z^T)^T$, where $(X,Z^T)^T$ is a $(p+1)$-dimension random vector. $Z\in \mathbb{R}^{p}$ follows a Gaussian $AR(1)$ process with autoregressive coefficient 0.5 and $X|Z\sim N(\mu(Z),\sigma^2)$, where the definition of $\mu$ and $\sigma^2$ will be given shortly. We let
$$
y_i = 0\cdot x_i + z_{[i]}^T\beta+\varepsilon_i,\ i=1,2,\ldots,400,
$$
where $\beta\in \mathbb{R}^{p}$, $\varepsilon_i$'s are i.i.d. standard Gaussian, independent of $x_i$ and $z_{[i]}$. $\beta$ is constructed as follows. We randomly sample a set of size 20 without replacement from $\{1,2,\ldots,p\}$, say S. For $j=1,\ldots,p$, let $\beta_j=0$ if $j\ne S$, and $\beta_j=0.5\mathbbm{1}\{B_j=1\}-0.5\mathbbm{1}\{B_j=0\}$ if $j\in S$, where $B_j$ is Bernoulli. 

\noindent {\bf ${\bf n>p}$ cubic regression:} $x$ and $z$ are generated in the same way as the previous model. Instead of assuming a linear model, we let
$$
y_i = 0\cdot x_i + (z_{[i]}\circ z_{[i]}\circ z_{[i]})^T\beta+\varepsilon_i,\ i=1,2,\ldots,400,
$$
where $\beta\in \mathbb{R}^{p}$. $\varepsilon_i$'s and $\beta$ are constructed in a same way as the previous model.

\noindent {\bf ${\bf n=p}$ linear regression:} We set $n=400$, $p=400$ and observations are i.i.d.. For $i=1,2,\ldots,400$, $(x_i,z_{[i]}^T)^T$ is an i.i.d. realization of $(X,Z^T)^T$, where $(X,Z^T)^T$ is a $(p+1)$-dimension random vector. $Z\in \mathbb{R}^{p}$ follows a Gaussian $AR(1)$ process with autoregressive coefficient 0.5 and $X|Z\sim N(\mu(Z),sigma^2)$. We let
$$
y_i = 0\cdot x_i + z_{[i]}^T\beta+\varepsilon_i,\ i=1,2,\ldots,400,
$$
where $\beta\in \mathbb{R}^{p}$, $\varepsilon_i$'s are i.i.d. standard Gaussian, independent of $x_i$ and $z_{[i]}$. $\beta$ is constructed as follows. We randomly sample a set of size 20 without replacement from $\{1,2,\ldots,p\}$, say S. For $j=1,\ldots,p$, let $\beta_j=0$ if $j\ne S$, and $\beta_j=0.5\mathbbm{1}\{B_j=1\}-0.5\mathbbm{1}\{B_j=0\}$ if $j\in S$, where $B_j$ is Bernoulli.

\noindent {\bf ${\bf n>p}$ logistic regression:} $x$ and $z$ are generated in the same way as the previous model, where $p=100$. We consider a binary model,
$$
y_i = \mathbbm{1}\{0\cdot x_i + z_{[i]}^T\beta+\varepsilon_i\},\ i=1,2,\ldots,400,
$$
where $\beta\in \mathbb{R}^{p}$. $\varepsilon_i$'s and $\beta$ are constructed in a same way as the $n>p$ linear regression model.

\begin{figure}[ht]
    \centering
    \begin{subfigure}[t]{0.32\textwidth}
        \centering
        \includegraphics[width=1\linewidth,height=0.7\linewidth]{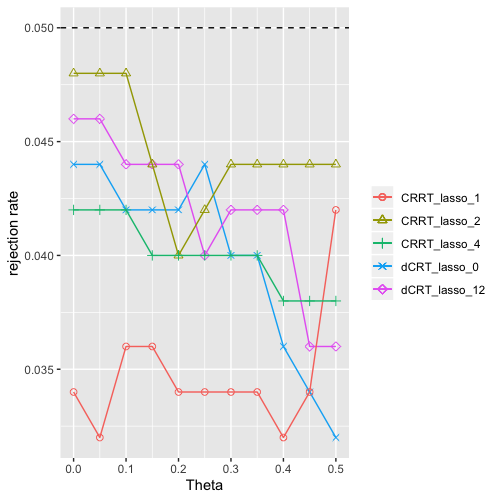} 
        \caption{Quadratic} \label{fig:pll1}
    \end{subfigure}
    \hfill
    \begin{subfigure}[t]{0.32\textwidth}
        \centering
        \includegraphics[width=1\linewidth,height=0.7\linewidth]{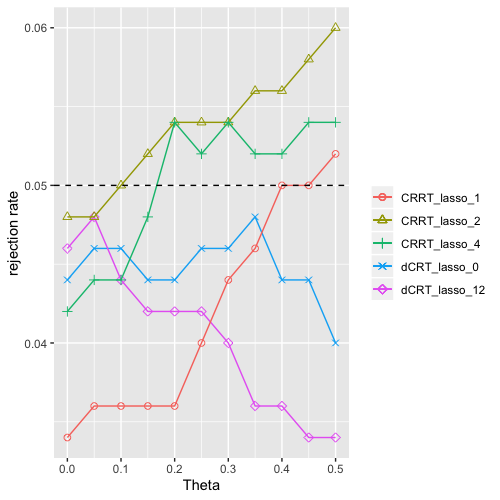} 
        \caption{Cubic} \label{fig:pll2}
    \end{subfigure}
    \hfill
    \begin{subfigure}[t]{0.32\textwidth}
        \centering
        \includegraphics[width=1\linewidth,height=0.7\linewidth]{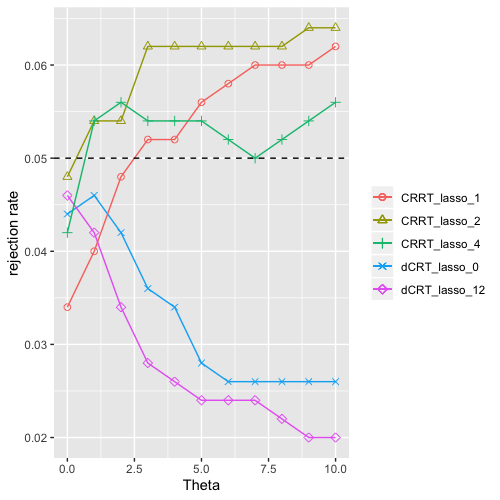} 
        \caption{Tanh} \label{fig:pll3}
    \end{subfigure}
\caption{Rejection rates of 5 Lasso-based tests under $n>p$ linear regression model and 3 non data-dependent misspecification settings: (a) Quadratic, (b) Cubic, (c) Tanh. The y-axis represents the proportion of times the null hypothesis of conditional independence is rejected out of 500 Monte-Carlo simulations. The x-axis is the value of $\theta$, which represents the degree of misspecification of conditional distribution. CRRT\_lasso\_k represents the CRRT using linear Lasso as the test function and with batch size $(b+1)/k$. dCRT\_lasso\_k represents the dCRT with linear Lasso used in the distillation step and keeping k important variables after distillation.}
\label{person_low_lasso}
\end{figure}

\begin{figure}[h]
    \centering
    \begin{subfigure}[t]{0.32\textwidth}
        \centering
        \includegraphics[width=1\linewidth,height=0.7\linewidth]{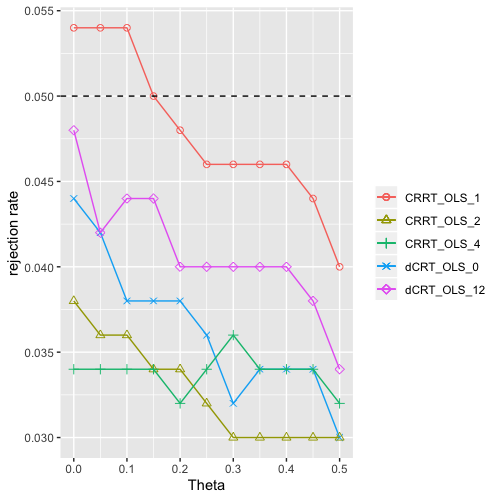} 
        \caption{Quadratic} \label{fig:plo1}
    \end{subfigure}
    \hfill
    \begin{subfigure}[t]{0.32\textwidth}
        \centering
        \includegraphics[width=1\linewidth,height=0.7\linewidth]{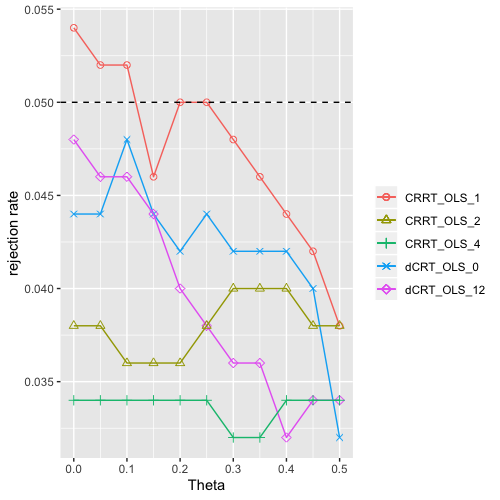} 
        \caption{Cubic} \label{fig:plo2}
    \end{subfigure}
    \hfill
    \begin{subfigure}[t]{0.32\textwidth}
        \centering
        \includegraphics[width=1\linewidth,height=0.7\linewidth]{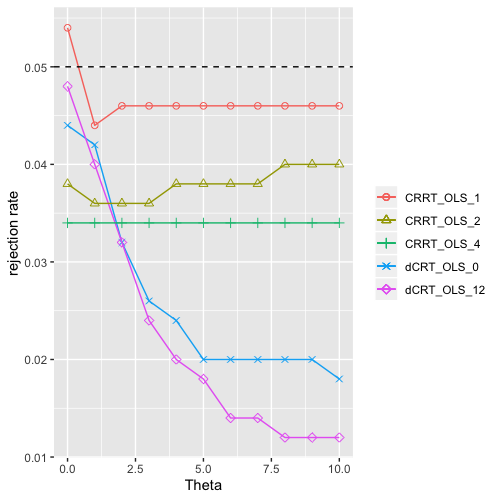} 
        \caption{Tanh} \label{fig:plo3}
    \end{subfigure}
\caption{Rejection rates of 5 OLS-based tests under $n>p$ linear regression model and 3 non data-dependent misspecification settings: (a) Quadratic, (b) Cubic, (c) Tanh. The y-axis represents the proportion of times the null hypothesis of conditional independence is rejected out of 500 Monte-Carlo simulations. The x-axis is the value of $\theta$, which represents the degree of misspecification of conditional distribution. CRRT\_OLS\_k represents the CRRT using OLS linear regression as the test function and with batch size $(b+1)/k$. dCRT\_lasso\_k represents the dCRT with OLS linear regression used in both of the distillation step and the testing step, and keeping k important variables after distillation.}
\label{person_low_ols}
\end{figure}

\begin{figure}[ht]
    \centering
    \begin{subfigure}[t]{0.32\textwidth}
        \centering
        \includegraphics[width=1\linewidth,height=0.7\linewidth]{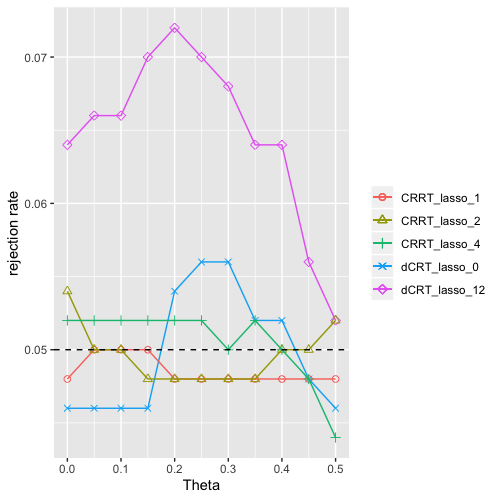} 
        \caption{Quadratic} \label{fig:plw1}
    \end{subfigure}
    \hfill
    \begin{subfigure}[t]{0.32\textwidth}
        \centering
        \includegraphics[width=1\linewidth,height=0.7\linewidth]{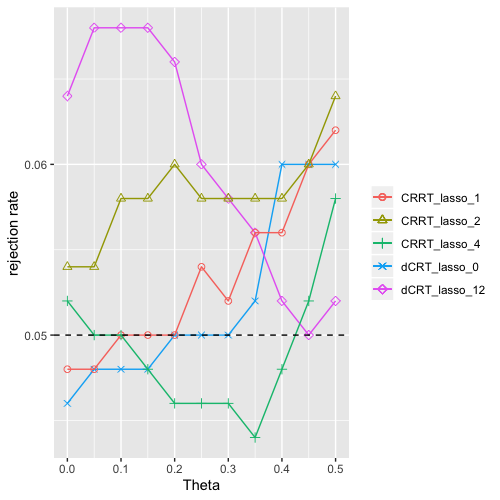} 
        \caption{Cubic} \label{fig:plw2}
    \end{subfigure}
    \hfill
    \begin{subfigure}[t]{0.32\textwidth}
        \centering
        \includegraphics[width=1\linewidth,height=0.7\linewidth]{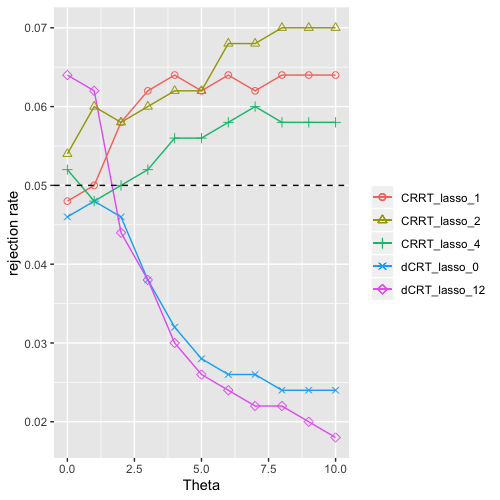} 
        \caption{Tanh} \label{fig:plw3}
    \end{subfigure}
\caption{Rejection rates of 5 Lasso-based tests under $n>p$ cubic regression model and 3 non data-dependent misspecification settings: (a) Quadratic, (b) Cubic, (c) Tanh. The y-axis represents the proportion of times the null hypothesis of conditional independence is rejected out of 500 Monte-Carlo simulations. The x-axis is the value of $\theta$, which represents the degree of misspecification of conditional distribution. Definitions of legends can be found in the Figure \ref{person_low_lasso}}
\label{person_low_wrong}
\end{figure}

\begin{figure}[ht]
    \centering
    \begin{subfigure}[t]{0.32\textwidth}
        \centering
        \includegraphics[width=1\linewidth,height=0.7\linewidth]{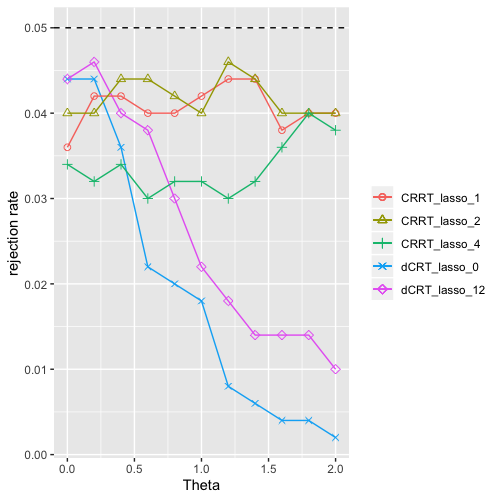} 
        \caption{Quadratic} \label{fig:ph1}
    \end{subfigure}
    \hfill
    \begin{subfigure}[t]{0.32\textwidth}
        \centering
        \includegraphics[width=1\linewidth,height=0.7\linewidth]{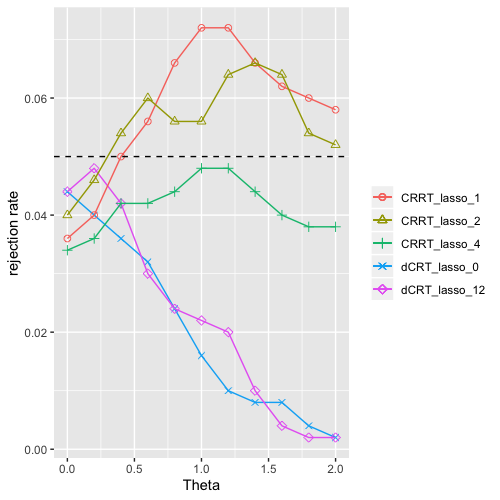} 
        \caption{Cubic} \label{fig:ph2}
    \end{subfigure}
    \hfill
    \begin{subfigure}[t]{0.32\textwidth}
        \centering
        \includegraphics[width=1\linewidth,height=0.7\linewidth]{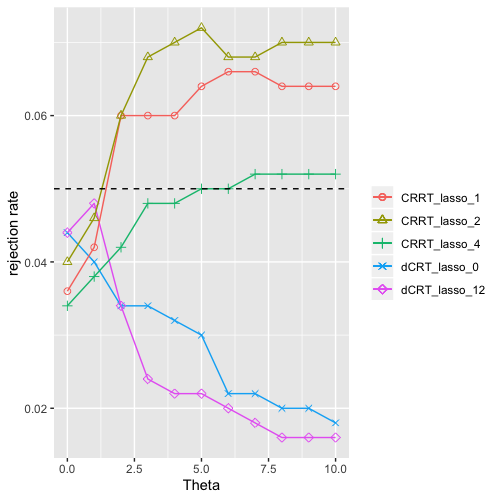} 
        \caption{Tanh} \label{fig:ph3}
    \end{subfigure}
\caption{Rejection rates of 5 Lasso-based tests under $n=p$ linear regression model and 3 non data-dependent misspecification settings: (a) Quadratic, (b) Cubic, (c) Tanh. The y-axis represents the proportion of times the null hypothesis of conditional independence is rejected out of 500 Monte-Carlo simulations. The x-axis is the value of $\theta$, which represents the degree of misspecification of conditional distribution. Definitions of legends can be found in the Figure \ref{person_low_lasso}}
\label{person_high}
\end{figure}

\begin{figure}[h]
    \centering
    \begin{subfigure}[t]{0.32\textwidth}
        \centering
        \includegraphics[width=1\linewidth,height=0.7\linewidth]{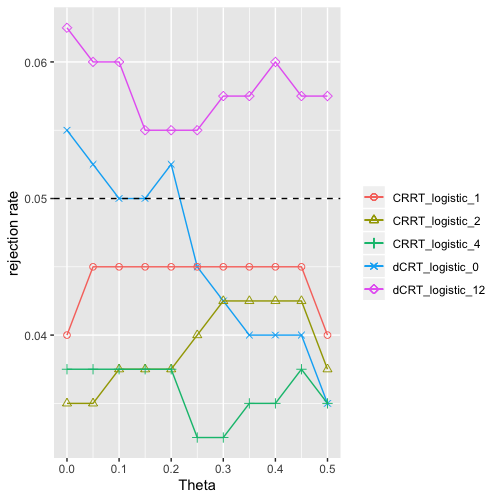} 
        \caption{Quadratic} \label{fig:pl1}
    \end{subfigure}
    \hfill
    \begin{subfigure}[t]{0.32\textwidth}
        \centering
        \includegraphics[width=1\linewidth,height=0.7\linewidth]{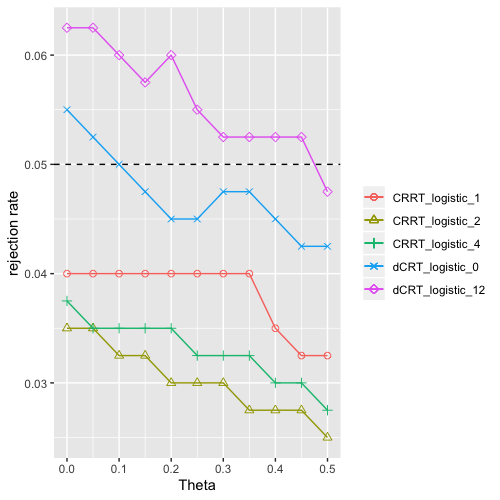} 
        \caption{Cubic} \label{fig:pl2}
    \end{subfigure}
    \hfill
    \begin{subfigure}[t]{0.32\textwidth}
        \centering
        \includegraphics[width=1\linewidth,height=0.7\linewidth]{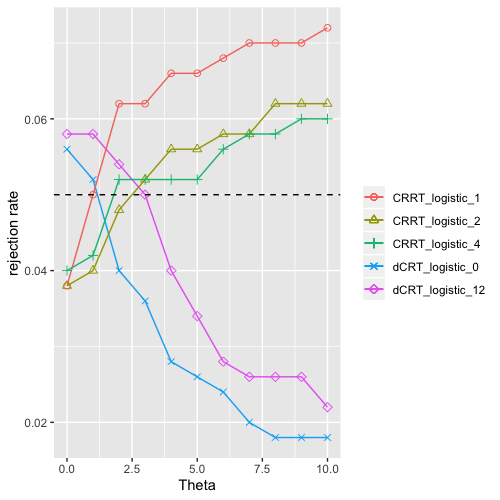} 
        \caption{Tanh} \label{fig:pl3}
    \end{subfigure}
\caption{Rejection rates of 5 logistic Lasso-based tests under $n>p$ logistic regression model and 3 non data-dependent misspecification settings: (a) Quadratic, (b) Cubic, (c) Tanh. The y-axis represents the proportion of times the null hypothesis of conditional independence is rejected out of 500 Monte-Carlo simulations. The x-axis is the value of $\theta$, which represents the degree of misspecification of conditional distribution. CRRT\_logistic\_k represents the CRRT using logistic Lasso as the test function and with batch size $(b+1)/k$. dCRT\_logistic\_k represents the dCRT with logistic Lasso used in the distillation step and keeping k important variables after distillation.}
\label{person_log}
\end{figure}

For each pseudo variable $X^{(k)}$, we let its conditional distribution $Q(\cdot|Z)=N(Z^T\zeta,\sigma^2)$ such that $(X^{(k)},Z^T)^T$ follows a Gaussian $AR(1)$ process with autoregressive coefficient 0.5. This is also the $\sigma^2$ we've mentioned above. We consider 3 types of non data-dependent conditional distribution misspecification by assigning $\mu(Z)$ different from $Z^T\zeta$, which are same as the Berrett, Wang, Barber and Samworth (2019):
\begin{itemize}
\item Quadratic: $\mu(Z)=Z^T\zeta+\theta(Z^T\zeta)^2$,
\item Cubic: $\mu(Z)=Z^T\zeta+\theta(Z^T\zeta)^3$,
\item Tanh: $\mu(Z)=\frac{tanh(\theta Z^T\zeta)}{\theta}\mathbbm{1}\{\theta\ne 0\}+Z^T\zeta\mathbbm{1}\{\theta= 0\}$,
\end{itemize}
where $\theta\in \mathbb{R}$ can reflect the degree of misspecification.

\noindent {\bf Results:} Results for the $n>p$ linear regression model are given in Figures \ref{person_low_lasso} and \ref{person_low_ols}. When we adopt OLS-based CRT's or dCRT's, we get more and more conservative results in general as $Q^{\star}$ deviates more from $Q$. Lasso-based dCRT's also show similar behaviors while Lasso-based CRT's have increasing type 1 error as misspecification is enlarged. 

Results for the $n>p$ cubic regression model are given in Figure \ref{person_low_wrong}. We can see quite different phenomenon in different misspecification settings. When the misspecification is quadratic, the impact of misspecification seems to be non-monotonic. When the misspecification is cubic, dCRT\_lasso\_12 gets more conservative while other tests show opposite behaviors as $\theta$ increases. When the misspecification is tanh, the dCRT's get more conservative while the CRRT's get more aggressive as $\theta$ increases.

Results for the $n=p$ linear regression model are given in Figure \ref{person_high}. No matter under which type of misspecification, the dCRT's always grow conservative as $\theta$ increases. The CRT's have slight increase in type 1 error as $\theta$ increases under quadratic and cubic misspecification and exhibit a relatively more obvious increase under the tanh misspecification.

Results for the $n>p$ logistic regression model are given in Figure \ref{person_log}. Still, the dCRT's become more conservative as $\theta$ increases. The CRT's remain stable under the quadratic misspecification, have decreasing type 1 error under the cubic misspecification and have increasing type 1 error under the tanh misspecification.

To sum up, what we can see here is totally different from what's been shown in the Subsection \ref{robustsim}, where both of the CRT and the dCRT have monotonically increasing rejection rate as the degree of misspecification increases. Results demonstrate that how misspecification of conditional distribution would affect the type 1 error is quiet complicated and unpredictable. It depends on the joint model, the test we use and the type of misspecification.

\end{document}